\documentclass[11pt]{article}

\usepackage[dvipsnames]{xcolor}
\usepackage[margin=0.8in]{geometry}

\usepackage{amsmath,amssymb,amsfonts,amsthm,mathtools}
\usepackage{thmtools,thm-restate}
\makeatletter
\@ifundefined{newcounteralias}{}{%
  \renewcommand\thmt@autorefsetup{%
    \@xa\def\csname\thmt@envname autorefname\@xa\endcsname\@xa{\thmt@thmname}}%
}
\makeatother
\usepackage{braket}
\usepackage{xcolor}
\usepackage{tikz}
\usepackage{comment}
\allowdisplaybreaks[1]

\usepackage[english]{babel}
\usepackage[T1]{fontenc}
\AtBeginDocument{%
  \DeclareFontShape{T1}{cmr}{m}{scit}{<->ssub*cmr/m/sc}{}%
}
\usepackage{enumerate}
\usepackage{enumitem}
\usepackage{float}
\usepackage[title]{appendix}
\usepackage{adjustbox}
\usepackage{indentfirst}
\usepackage{booktabs}
\usepackage{array}
\usepackage{colortbl}
\usepackage{subfigure}
\usepackage{tikz}
\usetikzlibrary{arrows.meta}
\usepackage{caption}
\usepackage[algo2e,ruled,vlined,linesnumbered]{algorithm2e}
\DeclareCaptionLabelSeparator{dqibar}{\enspace|\enspace}

\usepackage[
  pagebackref,
  colorlinks,
  bookmarks=true,
  linkcolor=red,
  citecolor=blue,
  urlcolor=magenta
]{hyperref}
\hypersetup{
  breaklinks=true,
  hypertexnames=false,
  pdfauthor={Kewen Wu and Guangxu Yang},
  pdftitle={Exponential quantum advantages for decoded quantum interferometry in the streaming setting}
}

\usepackage[capitalize,noabbrev,nameinlink]{cleveref}
\usepackage{footnotehyper}
\crefname{claim}{Claim}{Claims}
\Crefname{claim}{Claim}{Claims}
\crefname{fact}{Fact}{Facts}
\Crefname{fact}{Fact}{Facts}
\crefname{problem}{Problem}{Problems}
\Crefname{problem}{Problem}{Problems}
\crefname{algorithm}{Algorithm}{Algorithms}
\Crefname{algorithm}{Algorithm}{Algorithms}
\crefname{algocf}{Algorithm}{Algorithms}
\Crefname{algocf}{Algorithm}{Algorithms}

\declaretheorem[name=Theorem,numberwithin=section,style=plain]{theorem}
\declaretheorem[name=Lemma,sibling=theorem,style=plain]{lemma}

\declaretheorem[name=Claim,sibling=theorem,style=plain]{claim}
\declaretheorem[name=Fact,sibling=theorem,style=plain]{fact}
\declaretheorem[name=Proposition,sibling=theorem,style=plain]{proposition}
\declaretheorem[name=Definition,sibling=theorem,style=plain]{definition}
\declaretheorem[name=Problem,sibling=theorem,style=definition]{problem}

\numberwithin{equation}{section}

\newcommand{\Hrm}{\mathrm{H}}

\newcommand{\Fbb}{\mathbb{F}}
\newcommand{\Zbb}{\mathbb{Z}}

\newcommand{\Fsf}{\mathsf{F}}

\newcommand{\Dcal}{\mathcal{D}}
\newcommand{\Ecal}{\mathcal{E}}

\newcommand{\Gcal}{\mathcal{G}}
\newcommand{\Hcal}{\mathcal{H}}

\newcommand{\Lcal}{\mathcal{L}}
\newcommand{\Pcal}{\mathcal{P}}
\newcommand{\Qcal}{\mathcal{Q}}

\newcommand{\Tcal}{\mathcal{T}}

\newcommand{\abs}[1]{\left\lvert #1\right\rvert}
\newcommand{\1}{\mathbf{1}}

\newcommand{\inc}{\operatorname{inc}}
\newcommand{\npass}{p}
\newcommand{\polylog}{\operatorname{polylog}}
\DeclareMathOperator*{\E}{\mathbb{E}}
\DeclareMathOperator{\im}{im}

\newcommand{\wt}{\operatorname{wt}}

\newcommand{\OPT}{\operatorname{OPT}}
\newcommand{\poly}{\operatorname{poly}}
\newcommand{\err}{\operatorname{err}}
\newcommand{\codim}{\operatorname{codim}}
\newcommand{\score}{\operatorname{score}}
\newcommand{\Dec}{\operatorname{Dec}}
\newcommand{\Var}{\operatorname{Var}}
\newcommand{\ent}{\mathrm{H}}
\newcommand{\rv}[1]{\mathbf{#1}}

\renewcommand{\epsilon}{\varepsilon}

\title{Exponential quantum advantages for decoded quantum interferometry in the streaming setting}

\author{
Kewen Wu\thanks{California Institute of Technology. Email: \texttt{shlw\_kevin@hotmail.com}.}
\and
Guangxu Yang\thanks{University of Southern California. Email: \texttt{guangxuy@usc.edu}. Research supported by NSF CAREER award 2141536.}
}

\date{}

\begin{document}

\maketitle

\begin{abstract}
Decoded quantum interferometry (DQI) is a polynomial-time quantum algorithm introduced by Jordan et al. (Nature 2025). For a natural optimization problem, known as optimal polynomial intersection (OPI), it achieves approximation guarantees in regimes where all known classical algorithms require exponential time.

Besides time, space is another central resource: storing and manipulating a massive input can be very challenging, especially when logical qubits carry substantial fault-tolerant implementation overhead.
This motivates the following question.
\begin{center}
\it
Does DQI yield quantum advantages in memory, and can we prove it unconditionally?
\end{center}
We give an affirmative answer to this question in the streaming setting.
In particular, we consider a natural generalization of OPI using Hermite interpolation and Hasse derivatives, which asks for a low-degree polynomial satisfying as many constraints on its values and derivatives as possible.
As a concrete example, we show
\begin{itemize}
\item \textbf{Quantum efficiency.}
An adaptation of the DQI algorithm produces a polynomial satisfying $93\%$ of the constraints; moreover, it only reads the input stream in one pass, uses polylogarithmic space, and has polylogarithmic computation time per stream entry.
\item \textbf{Classical hardness.}
Any classical algorithm that produces an answer satisfying just $76\%$ of the constraints requires polynomial space, even if it can read the input stream with polynomially many passes and can use unlimited time.
\end{itemize}
Our result provides a complete tradeoff curve for the tunable parameters, and implies that DQI has provable quantum advantages for the original OPI problem.
\end{abstract}

\vspace{10pt}
\par\begingroup
\noindent\begin{minipage}{\linewidth}
\centering
\resizebox{0.90\linewidth}{!}
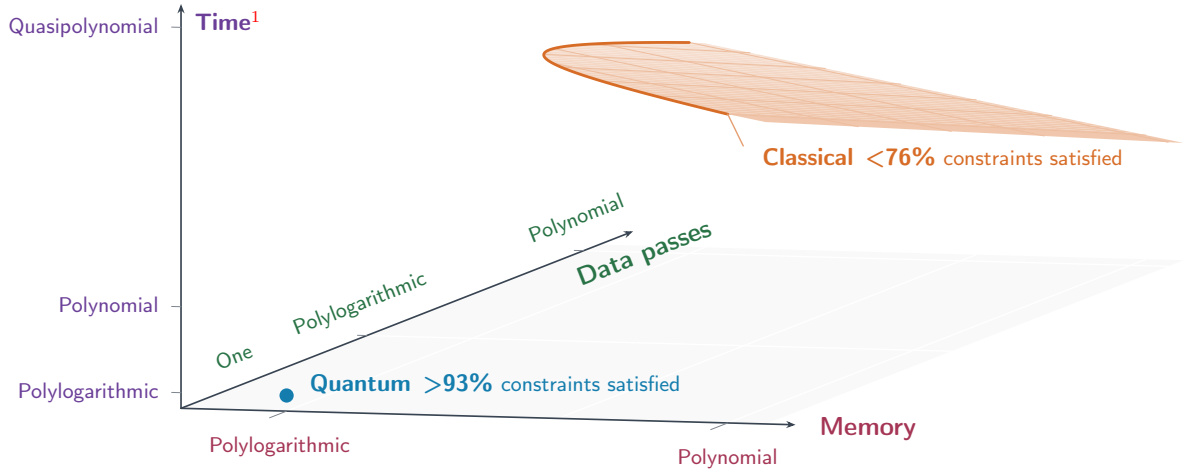
{
\begin{tikzpicture}[
  x={(5.6cm,-0.15cm)}, y={(4.6cm,1.8cm)}, z={(0cm,2.9cm)},
  font=\sffamily\small, line join=round, line cap=round]
\definecolor{titleQuantum}{HTML}{147DAD}
\definecolor{titleClassical}{HTML}{D76C27}
\definecolor{titleInk}{HTML}{384451}
\definecolor{titleMemoryAxis}{HTML}{A43555}
\definecolor{titlePassAxis}{HTML}{297346}
\definecolor{titleTimeAxis}{HTML}{6B3F98}

\def\quantumMemory{0.28}
\def\quantumTime{0.08}
\def\classicalTime#1#2{1.99-0.35*(#1-\quantumMemory)-0.70*(#2)}
\pgfmathsetmacro{\boundaryScale}{1.3/(1/\quantumMemory-1/1.45)}
\pgfmathsetmacro{\boundaryOffset}{\boundaryScale/1.45}
\def\passBoundary#1{\boundaryScale/(#1)-\boundaryOffset}

\fill[titleInk!3] (0,0,0) -- (1.55,0,0) -- (1.55,1.35,0)
  -- (0,1.35,0) -- cycle;
\foreach \v in {\quantumMemory,1,1.45} {
  \draw[white,line width=0.6pt] (\v,0,0) -- (\v,1.35,0);
}
\foreach \v in {0.6,1.3} {
  \draw[white,line width=0.6pt] (0,\v,0) -- (1.55,\v,0);
}

\foreach \i in {0,...,39} {
  \pgfmathsetmacro{\xa}{\quantumMemory+(1.45-\quantumMemory)*\i/40}
  \pgfmathsetmacro{\xb}{\quantumMemory+(1.45-\quantumMemory)*(\i+1)/40}
  \pgfmathsetmacro{\shade}{20+18*\i/39}
  \path[fill=titleClassical!\shade,draw=none]
    (\xa,{\passBoundary{\xa}},{\classicalTime{\xa}{\passBoundary{\xa}}})
    -- (\xb,{\passBoundary{\xb}},{\classicalTime{\xb}{\passBoundary{\xb}}})
    -- (\xb,1.35,{\classicalTime{\xb}{1.35}})
    -- (\xa,1.35,{\classicalTime{\xa}{1.35}}) -- cycle;
}
\fill[titleClassical!38] (1.45,0,{\classicalTime{1.45}{0}})
  -- (1.55,0,{\classicalTime{1.55}{0}})
  -- (1.55,1.35,{\classicalTime{1.55}{1.35}})
  -- (1.45,1.35,{\classicalTime{1.45}{1.35}}) -- cycle;
\foreach \v in {0.4,0.6,0.8,1,1.2,1.45} {
  \draw[titleClassical!45,line width=0.4pt]
    (\v,{\passBoundary{\v}},{\classicalTime{\v}{\passBoundary{\v}}})
    -- (\v,1.35,{\classicalTime{\v}{1.35}});
}
\foreach \v in {0.3,0.6,0.9,1.3} {
  \pgfmathsetmacro{\xm}{\boundaryScale/(\v+\boundaryOffset)}
  \draw[titleClassical!45,line width=0.4pt]
    (\xm,\v,{\classicalTime{\xm}{\v}})
    -- (1.55,\v,{\classicalTime{1.55}{\v}});
}
\draw[titleClassical,line width=1.2pt]
  plot[domain=\quantumMemory:1.45,samples=65]
    (\x,{\passBoundary{\x}},{\classicalTime{\x}{\passBoundary{\x}}});

\draw[titleInk,-{Stealth[length=4pt]},line width=0.7pt]
  (0,0,0) -- (1.63,0,0);
\draw[titleInk,-{Stealth[length=4pt]},line width=0.7pt]
  (0,0,0) -- (0,1.46,0);
\draw[titleInk,-{Stealth[length=4pt]},line width=0.7pt]
  (0,0,0) -- (0,0,2.07);
\coordinate (quantumPoint) at (\quantumMemory,0,\quantumTime);
\coordinate (classicalPoint) at (1.45,0,{\classicalTime{1.45}{0}});
\fill[titleQuantum] (quantumPoint) circle (3pt);

\foreach \v/\lab/\labelShift in {\quantumMemory/Polylogarithmic/-3pt,1.45/Polynomial/0pt} {
  \draw[titleInk!65] (\v,0,0) -- ++(0,-0.055,0);
  \node[text=titleMemoryAxis,anchor=north,xshift=\labelShift,yshift=-7pt,font=\sffamily\footnotesize]
    at (\v,0,0) {\lab};
}
\foreach \v/\lab in {0.6/Polylogarithmic,1.3/Polynomial} {
  \draw[titleInk!65] (0,\v,0) -- ++(-0.025,0,0);
  \node[text=titlePassAxis,anchor=south,rotate=21.4,font=\sffamily\footnotesize]
    at (0,\v,0.07) {\lab};
}
\node[text=titlePassAxis,anchor=south,rotate=21.4,font=\sffamily\footnotesize]
  at (0,0.19,0.07) {One};
\foreach \v/\lab in {0.08/Polylogarithmic,0.52/Polynomial,1.95/Quasipolynomial} {
  \draw[titleInk!65] (0,0,\v) -- ++(-0.025,0,0);
  \node[text=titleTimeAxis,anchor=east,font=\sffamily\footnotesize]
    at (-0.035,0,\v) {\lab};
}

\begin{scope}[x=1cm,y=1cm]
  \node[text=titleMemoryAxis,anchor=north,font=\sffamily\bfseries]
    at (10.2,0) {Memory};
  \node[text=titlePassAxis,anchor=south,rotate=21.4,font=\sffamily\bfseries]
    at (7,2) {Data passes};
  \node[text=titleTimeAxis,anchor=south,align=center] at (0.7,5.475)
    {\textbf{Time}\footnotemark};
  \node[text=titleQuantum,anchor=west,align=left,xshift=6pt,yshift=5pt]
    at (quantumPoint)
    {\textbf{Quantum}\enspace\textbf{>93\%}
     {\footnotesize constraints satisfied}};
  \node[text=titleClassical,anchor=west,align=left] at (8.5,3.725)
    {\textbf{Classical}\enspace\textbf{<76\%}
     {\footnotesize constraints satisfied}};
  \draw[titleClassical!65,line width=0.6pt]
    (8.35,3.895) -- (classicalPoint);
\end{scope}
\end{tikzpicture}
}\par
\end{minipage}
\footnotetext{Time denotes processing per entry, including preprocessing and final postprocessing. All known classical algorithms run
in quasipolynomial time to obtain $>75\%$ approximation.}
\par\endgroup

\newpage
\tableofcontents
\newpage

\section{Introduction}\label{sec:introduction}

Quantum computing offers substantial advantages over classical computing, and the prospect of such advantages for optimization has motivated a number of works~\cite{FarhiEtAl2025HighGirthMaxCut,
PirnayEtAl2024InPrincipleAdvantage,
GilyenHastingsVazirani2021AdiabaticAdvantage}.
One such approach is decoded quantum interferometry (DQI)~\cite{jordan2025dqi}. It applies to optimal polynomial intersection (OPI), a polynomial fitting problem over finite fields: each evaluation point comes with a set of desired values, and the goal is to find a polynomial of bounded degree whose evaluation is acceptable at as many points as possible. 
DQI is a polynomial-time quantum algorithm and achieves optimization guarantees for OPI in regimes where all known classical algorithms require exponential time. Recent improvements further support OPI as a candidate for practical, classically verifiable quantum advantage~\cite{khattar2025circuits,rosmanis2026}, though establishing this unconditionally remains a major open challenge in complexity theory.

Aside from time, space is another central resource, especially when logical qubits carry substantial fault-tolerant implementation overhead~\cite{FowlerMariantoniMartinisCleland2012}. This motivates our first question:
\begin{center}
\textit{Does DQI yield quantum space advantage, and can we prove it unconditionally?}
\end{center}
Streaming algorithms provide a natural framework for studying space-efficient computation~\cite{LeGall2006, NayakTouchette2017}. Motivated by the need to process massive datasets~\cite{BabcockBabuDatarMotwaniWidom2002, Muthukrishnan2005}, the streaming model presents input elements sequentially and has found many applications, including computing statistics over data streams~\cite{FlajoletMartin1985, AlonMatiasSzegedy1999} and estimating parameters of massive graphs~\cite{BarYossefKumarSivakumar2002, McGregor2014}.

The number of data passes is a third important resource: revisiting the input many times can reduce memory requirements. Indeed, one-pass quantum streaming advantages are known for estimating the
optimum values of Max-DiCut~\cite{kallaugher2024exponential} and Max-2SAT~\cite{wang2026exponential}, but
constant-pass classical algorithms can already match their performances~\cite{saxena2025streaming,fei2026dichotomy}, raising a second question:\footnote{Similar questions have been studied in different settings and for contrived problems \cite{ggjl2025,girish2022quantum}. We focus on more natural problems: OPI and its Hermite generalization.}
\begin{center}
\textit{Can one-pass quantum algorithms outperform many-pass classical algorithms for natural problems?}
\end{center}

We answer both questions affirmatively in one shot.\footnote{Obviously, this is not saying that the results are proved by an LLM in one shot. See \Cref{sec:open-problems} for AI disclosure.}
We adapt DQI to the streaming setting and prove rigorous classical lower bounds for OPI. Let $N$ be the stream length and let $\alpha\in(0,1]$ be the approximation guarantee: larger $\alpha$ is better and $\alpha=100\%$ finds the exact optimum.

\begin{theorem}[Quantum advantage for OPI]\label{thm:headline-opi}
For some constant $c>0$, the following holds for finding an $\alpha$-approximate solution for OPI in the streaming setting.
\begin{itemize}[labelsep=5pt,leftmargin=15pt,topsep=4pt,itemsep=0pt]
\item $\alpha=99.999\%$.
A one-pass $\widetilde O(\sqrt N)$-space $\exp\bigl(\widetilde O(\sqrt N)\bigr)$-time quantum algorithm suffices.
\item $\alpha=93.301\%$. 
A one-pass $\widetilde O(\sqrt N)$-space $\poly(N)$-time quantum algorithm suffices.
\item $\alpha=75.001\%$. 
Provably no $N^c$-pass $N^{1/2+c}$-space classical algorithm exists, even with unlimited time.
\end{itemize}
\end{theorem}

Our techniques extend to a natural generalization of OPI using Hermite interpolation. 
In this Hermite OPI, each evaluation point has a set of acceptable values and (high-order) derivatives, and the output polynomial tries to match both the values and the derivatives. 
For this problem, we obtain, qualitatively, the strongest possible quantum streaming advantages.

\begin{theorem}[Stronger advantage for Hermite OPI]\label{thm:headline-hopi}
For some constant $c>0$, the following holds to find an $\alpha$-approximate solution for a version of Hermite OPI in the streaming setting.
\begin{itemize}[labelsep=5pt,leftmargin=15pt,topsep=4pt,itemsep=0pt]
\item $\alpha=99.999\%$. 
A one-pass $\polylog(N)$-space $\exp\bigl(\polylog(N)\bigr)$-time quantum algorithm suffices.
\item $\alpha=93.301\%$. 
A one-pass $\polylog(N)$-space $\polylog(N)$-per-entry-time\footnote{This is the processing time per stream entry arrival, including the time to generate the answer at the end.} quantum algorithm suffices.
\item $\alpha=75.001\%$. 
Provably no $N^c$-pass $N^c$-space classical algorithm exists, even with unlimited time.
\end{itemize}
\end{theorem}

For the rest of the section, we give the formal definition of Hermite OPI, present our formal theorems, and discuss related works.

\subsection{Problem setup}\label{sec:intro-setup}

Throughout, $\log$ denotes the base-two logarithm,
$[n]=\{1,2,\ldots,n\}$ for each positive integer $n$, $\1[E]$ for the indicator of event $E$, $\Fbb_q$ for the field of size $q$, and
$\Fbb_q^\times=\Fbb_q\setminus\{0\}$.
Big-O notations hide absolute constants unless a subscript is used; for example, $O_\kappa(\cdot)$ allows the hidden constant to depend on $\kappa$.
The notation $\widetilde O(\cdot)$ suppresses polylogarithmic factors and $\poly(\cdot)$ hides a fixed polynomial bound.

\paragraph{Hermite interpolation and Hasse derivatives.}
Polynomial interpolation prescribes a polynomial's values at distinct
points. Hermite interpolation also prescribes derivatives.
Over finite fields, these constraints are expressed using Hasse derivatives:
for $f\in\Fbb_q[X]$, its $j$th Hasse derivative $f^{(j)}$ is defined by\footnote{Consider $f(X)=X^3$ over $\Fbb_3$. Ordinary derivatives show that every derivative vanishes, which is not very informative. Hasse derivatives detect the difference as $f(X+T)=X^3+T^3$ and $f^{(3)}(X)=1$.}
\[
 f(X+T)=\sum_{j\ge0}f^{(j)}(X)T^j
 \quad\text{where $T$ is a formal variable.}
\]
Note that $f^{(0)}=f$. For $\ell\ge1$, the length-$\ell$ jet at $a\in\Fbb_q$ records the first $\ell$ Hasse derivatives:
\[
 J_\ell f(a)=\left(f^{(0)}(a),\ldots,f^{(\ell-1)}(a)\right)\in\Fbb_q^\ell.
\]

Similar to polynomial interpolation, any prescribed length-$\ell$ jets at any $n$ distinct points are realized by a unique polynomial of degree less than $\ell n$.
Our Hermite OPI problem is given a set of allowed jets and asks for a polynomial of smaller degree whose jets belong to as many of these sets as possible.

\begin{problem}[Hermite optimal polynomial intersection (HOPI)]\label{prob:hopi}
Let prime $q\ge3$, integer $\ell\ge1$, and $\alpha\in(0,1]$.
Write $n=q-1$, $\Sigma=\Fbb_q^\ell$, and $K=\ell n/2$,
and identify $\Fbb_q^\times$ with $[n]$.
An instance of $(\alpha,q,\ell)$-HOPI consists of
$S_i\subseteq\Sigma$ of size $(|\Sigma|-1)/2$ for each
$i\in[n]$.
The goal is to output a polynomial
$f\in\Fbb_q[X]$ of degree less than $K$ that $\alpha$-approximately maximizes
\[
 \score(f)
 =\abs{\left\{i\in[n]:J_\ell f(i)\in S_i\right\}}.
\]
In other words, $f$ should satisfy $\score(f)\ge\alpha\cdot\score(g)$ for every $g$ of degree less than $K$.
\end{problem}

The $\ell=1$ version of \Cref{prob:hopi} is the OPI problem~\cite{jordan2025dqi} that only concerns polynomial interpolation.
Here is the HOPI problem in the streaming setting.

\begin{problem}[Streaming HOPI]
\label{prob:streaming-hopi}
In the streaming setting, an $(\alpha,q,\ell)$-HOPI instance is given by $N=n\cdot|\Sigma|$ classical entries
\[
 (i,z,\1[z\in S_i])
 \quad\text{for each}\quad
 i\in[n]
 \quad\text{and}\quad 
 z\in\Sigma,
\] 
which arrive one at a time in some order.
A $\npass$-pass streaming algorithm solves $(\alpha,q,\ell)$-HOPI if it reads the stream $\npass$ times in the same order and, with probability at least $0.99$, outputs an $\alpha$-approximate maximizer $f$ at the end.
\end{problem}

We also remark that a number of parameters in \Cref{prob:hopi} and \Cref{prob:streaming-hopi} can be made tunable\footnote{For example, prime powers, fewer evaluation points, different list sizes, degree bounds, and success rates.} and our results generalize; for simplicity, we work with the current setup.
Streaming algorithms and hardness in this paper adopt the following conventions which we will not repeat afterwards.
\begin{itemize}[labelsep=5pt,leftmargin=15pt,topsep=4pt,itemsep=0pt]
\item \textbf{Space complexity} is the maximum number of (qu)bits used at any time in the computation.
\item \textbf{Preprocessing-update-postprocessing time} measures the runtime to prepare before any stream entry, to update per entry arrival, and to produce the answer after the stream, respectively.
\item \textbf{Input instance and streaming order.} 
All our algorithms work for any instance and any streaming order.
All our lower bounds hold for uniformly random instances in a fixed grouped order: each $S_i$ is independently uniform and the entries $(i,z,\1[z\in S_i])$ arrive consecutively as groups of $i$. The group order and order within groups can be arbitrarily fixed in advance.
\end{itemize}

\subsection{Our results}\label{sec:intro-results}

Define
$$
\alpha_\star=\frac{2+\sqrt3}{4}\approx0.93301
\quad\text{and}\quad
\kappa(\alpha)=\frac{4\alpha-3}{4\alpha-2}
\quad\text{for every $\alpha\in[3/4,1]$.}
$$
Our first algorithm adapts DQI to streaming HOPI and achieves every fixed approximation ratio below $\alpha_\star$ with triple efficiency: small space, fast updates, and a single pass over the input.
We prove \Cref{thm:quantum-poly} in \Cref{sec:quantum-streaming}.

\begin{restatable}{theorem}{QuantumStreamingUpper}
\label{thm:quantum-poly}
Let $\alpha\in(0,\alpha_\star)$ be arbitrary.
For every prime $q\ge\Omega_\alpha(1)$ and every integer $\ell\ge1$, there is a one-pass quantum streaming algorithm solving $(\alpha,q,\ell)$-HOPI with
$O(q\ell\log q)$ space complexity and $\poly(q,\ell)$ preprocessing-update-postprocessing time.
\end{restatable}

Allowing exponential-time postprocessing lets us push the approximation ratio arbitrarily close to $1$ with the same low space complexity and a single pass.
Preprocessing and update time remain polynomial.
We prove \Cref{thm:quantum-exp-summary} in \Cref{sec:quantum-near-optimal}.

\begin{restatable}{theorem}{QuantumNearOptimalUpper}
\label{thm:quantum-exp-summary}
Let $\alpha\in(0,1)$ and integer $\ell\ge1$ be arbitrary.
For every prime $q\ge\Omega_\alpha(\ell\log(\ell+1))$, there is a one-pass quantum streaming
algorithm solving $(\alpha,q,\ell)$-HOPI with $O(q\ell\log q)$ space
complexity, $\poly(q,\ell)$ preprocessing-update time, and
$q^{O(q\ell)}$ postprocessing time.
\end{restatable}

Classical algorithms require substantially more space or passes to achieve comparable approximation, even with unlimited computation time.
\Cref{thm:classical-lower-summary}, proved in \Cref{sec:classical}, quantifies this tradeoff for every approximation ratio above $3/4$.

\begin{restatable}{theorem}{ClassicalStreamingLower}
\label{thm:classical-lower-summary}
Let $\alpha\in(3/4,1]$, integers $\ell,\npass\ge1$, and prime $q\ge3$ be arbitrary.
Any randomized $\npass$-pass classical streaming algorithm solving $(\alpha,q,\ell)$-HOPI has space complexity at least
\[
 \frac{q^{1+\ell\cdot\kappa(\alpha)}}
 {\npass}\cdot(\ell\log q)^{-C_\alpha\cdot\ell},
\]
where $C_\alpha>0$ is a constant depending only on $\alpha$.
\end{restatable}

Combining this tradeoff with our quantum upper bounds and choosing concrete parameters gives the headline results \Cref{thm:headline-opi,thm:headline-hopi}.
We present these calculations in \Cref{sec:headline_calc}.

Finally, our classical lower bounds are essentially optimal.
\Cref{thm:classical-upper-summary}, proved in \Cref{sec:classical-upper}, gives a one-pass classical algorithm with nearly matching space complexity for approximation ratios above $3/4$.

\begin{restatable}{theorem}{ClassicalStreamingUpper}
\label{thm:classical-upper-summary}
Let $\alpha\in(3/4,1)$, integer $\ell\ge1$, and prime $q\ge3$ be arbitrary. There is a randomized one-pass classical streaming algorithm solving $(\alpha,q,\ell)$-HOPI with $O_{\ell,\alpha}(q^{1+\ell\kappa(\alpha)}\log q)$ space complexity, $\poly(q,\ell)$ preprocessing-update time, and $q^{O(q\ell)}$ postprocessing time.
\end{restatable}

For approximation ratios at most $3/4$, classical algorithms can match the space and time bounds of our efficient one-pass quantum algorithm.
We prove \Cref{thm:classical-three-quarters-summary} in \Cref{sec:c34}.

\begin{restatable}{theorem}{ClassicalThreeQuartersUpper}
\label{thm:classical-three-quarters-summary}
For every sufficiently large prime $q$, every $\alpha\in(0,3/4]$, and every integer $\ell\ge1$, there is a randomized one-pass classical streaming algorithm solving $(\alpha,q,\ell)$-HOPI with $O(q\ell\log q)$ space complexity and $\poly(q,\ell)$ preprocessing-update-postprocessing time.
\end{restatable}

Our results fit into a broader coding-theoretic picture. At a high level, DQI and related quantum advantage proposals come from Regev's reduction~\cite{Reg09} applied to different code families.
In \Cref{tab:code-advantages}, we summarize them based on the codes, the associated problems, and their corresponding quantum advantages.

\begin{savenotes}
\begin{table}[H]
\centering
\small
\renewcommand{\arraystretch}{1.2}
\begin{tabular}{@{}p{0.4\linewidth}
                  p{0.27\linewidth}
                  p{0.28\linewidth}@{}}
\toprule
Codes (associated problem)
& Time advantages
& Space advantages \\
\midrule

Reed--Solomon (OPI)
& Exponential \cite{jordan2025dqi}
& Polynomial (\textbf{this work}) \\

Multiplicity (Hermite OPI)\footnote{The exponential time advantage uses constant $\ell$, whose space advantage is polynomial. Setting $\ell$ to grow polynomially in $q$ gives exponential space advantage, whose time advantage is quasipolynomial.}
& Exponential (\textbf{this work})
& Exponential (\textbf{this work}) \\

Folded Reed--Solomon (codeword search)
& Exponential \cite{yamakawaZhandry2024}
& Exponential \cite{ggjl2025} \\

Random ternary ($\Fbb_3^n$-Subset-Sum)
& Exponential~\cite{KothariEtAlQuantumSubsetSum}
& Unknown \\

Random $q$-ary ($\ell_\infty$-Short-Integer-Solution)
& None~\cite{ChenLiuZhandry2022Filtering,KothariODonnellWu2025}
& Unknown \\

LDPC (Max-Cut, Max-XORSAT)
& \raggedright None~\cite{parekh2025maxcut,anschuetz2026spin}
& None~\cite{kallaugher2022quantumB}\\

Hermitian (Hermitian OPI)\footnote{Their Hermitian OPI uses Hermitian algebraic-geometry codes, whereas our Hermite OPI uses univariate multiplicity codes. These are different code families and different problems.}
& Exponential~\cite{GuJordan2025Hermitian}
& Unknown \\
\bottomrule
\end{tabular}
\caption{Quantum advantages for code-based problems via Regev's reduction. Space advantages concern the streaming setting; time advantages are relative to known classical algorithms.}
\label{tab:code-advantages}
\end{table}
\end{savenotes}

\vspace{-10pt}
\subsection{Related work}\label{sec:related-works}

\paragraph{Quantum algorithms based on Regev's reduction.}
Regev's quantum reduction~\cite{Reg09}, originally used to establish
hardness of Learning with Errors, combines coherent decoding with
Fourier sampling. Chen, Liu, and
Zhandry~\cite{ChenLiuZhandry2022Filtering} turn this framework into
polynomial-time quantum algorithms for certain average-case
$\mathrm{SIS}^{\infty}$ and related constrained-integer-solution
problems. Recently, Kothari, O'Donnell, and Wu~\cite{KothariODonnellWu2025}
give classical algorithms matching or improving these guarantees,
building in part on earlier classical zero-sum algorithms of Imran and
Ivanyos~\cite{ImranIvanyos2024ZeroSum}.
Very recently, Kothari, Metger, O'Donnell, Shutty, and
Wu~\cite{KothariEtAlQuantumSubsetSum} re-establishes a potential exponential
quantum speedup for one of their proposals.
Coding-theoretic adaptations and extensions are developed
in~\cite{DebrisAlazardRemaudTillich2024,ChaillouxTillich2024QuantumDecoding,BlanvillainChaillouxTillich2025}.
Yamakawa and Zhandry~\cite{yamakawaZhandry2024} use Fourier interference and decoding to obtain an exponential quantum query advantage for codeword search relative to random oracles. Coherent decoding and Fourier sampling also underlie DQI and
our streaming implementation.

\paragraph{DQI and OPI.}
Jordan et al.~\cite{jordan2025dqi} introduce DQI for optimization
and obtain OPI guarantees beyond known efficient classical algorithms.
Subsequent work improves decoding circuits~\cite{khattar2025circuits}
and gives a nearly linear-time OPI implementation in a quantum
random-access model~\cite{rosmanis2026}.
Soft-decoding methods improve guarantees in certain average-case
regimes~\cite{chaillouxTillich2025soft,chailloux2025opi},
while recent work gives existential and algorithmic improvements
beyond the DQI semicircle law for worst-case
OPI~\cite{sunWootters2026,horinagaYamakawa2026,jo2026sampling}. Broader extensions of DQI are studied in~\cite{GuJordan2025Hermitian,bu2026weighted,KrajenbrinkEtAl2026BeyondHamming,SchmidhuberEtAl2025HamiltonianDQI,BuGuLi2026GeneralPauli}, and its computational power and limitations
in~\cite{marwaha2025complexity,parekh2025maxcut,anschuetz2026spin,KramerSchubertEisert2026BoundedDegree}. List decoding already appears in quantum reductions~\cite{chaillouxTillich2025soft,horinagaYamakawa2026,jo2026sampling}.
For fixed $\ell$, our near-optimal algorithm combines sampled-coordinate isolation with exhaustive reversible decoding to implement coherent syndrome compression within the streaming space bound, without further input access; see \Cref{sec:quantum-near-optimal}. Our results establish unconditional space--pass advantages for OPI and Hermite OPI.

\paragraph{Quantum streaming.}
Early quantum space separations were obtained by Le Gall ~\cite{LeGall2006} and Gavinsky et al.~\cite{GavinskyKempeKerenidisRazDeWolf2008}. Subsequent work studied natural problems, including $\mathrm{Dyck}(2)$ and triangle counting~\cite{JainNayak2014,NayakTouchette2017,kallaugher2022quantum}.  Multipass quantum tradeoffs have also been studied for frequency-moment estimation~\cite{Montanaro2016, HamoudiMagniez2019} and Shannon entropy estimation~\cite{FengEtAl2026}.
Zhao et al.~\cite{ZhaoEtAl2026MassiveData} develop quantum oracle sketching,
obtaining exponential quantum space advantages for classification and
dimension reduction from random classical samples.
For optimization, Kallaugher, Parekh, and
Voronova~\cite{kallaugher2024exponential} and Wang and Yang~\cite{wang2026exponential} obtain one-pass quantum space
advantages for estimating the optimum values of Max-DiCut and
Max-$k$SAT, respectively.
Constant-pass logarithmic-space classical algorithms already surpass the corresponding quantum approximation guarantees for Max-DiCut and Max-$k$SAT~\cite{fei2026dichotomy}.
Our algorithms output an explicit polynomial solution, and our lower bounds
persist against many-pass classical algorithms, even with unlimited
computation. The quantum algorithm only uses one pass.

\paragraph{Communication and coding techniques.}
G{\"o}{\"o}s et al.~\cite{ggjl2025} obtain an exponential separation
between quantum simultaneous messages and unrestricted randomized
interaction for a total, classically verifiable search relation,
providing a direct antecedent to our work.
We observe that the Forrelation-based construction
of~\cite{girish2022quantum} also yields a streaming space--pass
separation for indexed-bit inputs.
Our classical lower bound uses structure-versus-randomness
and gadgetless-lifting
techniques~\cite{gpw2020,yangzhang2024,maoyangzhang2024,ggjl2025},
including multiparty min-entropy methods~\cite{huang2025minentropy}.
Our Hermite list-recovery bound extends hidden-derivative
interpolation~\cite{bcpzz2026,jeronimo2026} to jet constraints,
using total-incidence accounting~\cite{jeronimo2026} and
Kopparty's differential-root bound~\cite[Theorem~4.3]{kopparty2015}.
Together with multiplicity-code list
decoding~\cite{KoppartyRonZewiSarafWootters2023}, this yields
the list bound for our near-optimal quantum algorithm.
Sampled-coordinate isolation has a classical antecedent in
the \texttt{PRUNE} procedure
of~\cite[Section~3.1]{KoppartyRonZewiSarafWootters2023}.

\paragraph{Other quantum optimization algorithms.}
Quantum optimization advantages have been studied under computational assumptions~\cite{PirnayEtAl2024InPrincipleAdvantage,Szegedy2022QuantumAdvantageSimplified} and in oracle models~\cite{GilyenHastingsVazirani2021AdiabaticAdvantage,LengEtAl2025OptimizationSpeedup}. Related work studies nonconvex
optimization~\cite{LengZhengWu2023NonconvexSeparation} and constant-depth QAOA for Max-Cut on high-girth 3-regular graphs~\cite{FarhiEtAl2025HighGirthMaxCut}. These works concern time, query complexity, or solution quality.
Our results establish unconditional streaming space--pass separations
for explicitly presented optimization instances, even with unlimited classical computation.

\subsection{Open problems and AI disclosure}
\label{sec:open-problems}

Our quantum upper bounds leave open the cost of obtaining stronger
approximation guarantees.
\Cref{thm:quantum-poly} achieves every approximation ratio below
$\alpha_\star$ with one pass, $O(q\ell\log q)$ space, and polynomial
preprocessing, update, and postprocessing time.
For fixed $\ell$, the near-optimal guarantee in
\Cref{thm:quantum-exp-summary} reaches every ratio below $1$,
but uses exponential-time postprocessing.
Can polynomial-time postprocessing suffice with similar space and
pass bounds?
Alternatively, can one establish computational hardness?
Exact optimization raises a separate question: does finding an optimal
polynomial require more quantum memory, or can a one-pass algorithm
achieve it in small space with unrestricted postprocessing?

On the classical side, the one-pass bounds already match up to low-order terms. The remaining tradeoffs concern additional passes and
computational efficiency.
Can algorithms match the inverse dependence on the number of passes
in \Cref{thm:classical-lower-summary}, particularly when their memory
is comparable to the quantum budget?
Can they achieve the approximation--space frontier without
exponential-time postprocessing?

A broader question is which other code families support quantum
streaming advantages through Regev's reduction.
Can we obtain a coding-theoretic classification? A first step would be completing the missing entries in \Cref{tab:code-advantages}.

\paragraph{AI disclosure.}
The authors formulated the questions in \Cref{sec:introduction} and observed the streaming DQI algorithm (\Cref{thm:quantum-poly}); GPT-6 Astra was used to check and refine the proof. 
The authors subsequently proposed the main idea for \Cref{thm:classical-lower-summary}, whose overall proof builds on the framework developed by one of the authors \cite{yangzhang2024} before the LLM age.  GPT-6 Astra was used to help improve and extend bounds from \cite{bcpzz2026,jeronimo2026} to prove \Cref{hendpoint:lr}. After completing \Cref{thm:quantum-poly} and \Cref{thm:classical-lower-summary}, the authors aimed to complete the tradeoff curve and prompted GPT-6 Astra for classical upper bounds, leading to \Cref{thm:classical-upper-summary} and \Cref{thm:classical-three-quarters-summary}. \Cref{thm:quantum-exp-summary} was the last derived: the authors observed that the streaming DQI algorithm is a soft state identification problem whose information-theoretic answer should be better than what efficient unique decoding can provide; guided by this, GPT-6 Astra completed the proof.

GPT-6 was also used to draft the paper, which was fully rewritten by the authors.
The authors take full responsibility for the correctness, exposition, and attribution in the final manuscript.

\section{Technical overview}\label{sec:overview}

We sketch the main ideas behind our quantum algorithms and classical bounds.
For convenience, let $n=q-1$ and we identify $\Fbb_q^\times$ by $[n]$.
Set $\Sigma=\Fbb_q^\ell$ and $K=\ell n/2$.
The Hermite code\footnote{It is the univariate multiplicity code in \cite{kopparty2015}. We call it Hermite code to be consistent with Hermite OPI.} $C_\Hrm\subseteq\Sigma^n$ consists of the jet vectors
$(J_\ell f(i))_{i\in[n]}$ of polynomials $f\in\Fbb_q[X]$ with $\deg f<K$.
Let $G_\Hrm$ be its generator matrix such that $G_\Hrm^\top u$ is the jets of
$f_u(X)=\sum_{r=0}^{K-1}u_rX^r$.
The goal of HOPI can be equivalently phrased as finding $u\in\Fbb_q^K$ with large 
$$
\score(u):=\score(f_u)=\sum_{i\in[n]}\1[(G_\Hrm^\top u)_i\in S_i],
$$
given allowed jets $S_i\subseteq\Sigma$ for each $i\in[n]$.

\subsection{Triply efficient quantum streaming at \texorpdfstring{$93\%$}{93\%} approximation}\label{subsec:tri}

We adapt decoded quantum interferometry
(DQI)~\cite{jordan2025dqi} to the one-pass streaming model.
The algorithm loads the streaming input into a quantum state, after which it projects onto errors that admit unique syndrome decoding and uses Fourier interference to sample a high-scoring polynomial. 
All stages use small space and are also efficient in time. 
\Cref{fig:streaming-hermite-dqi} depicts the full algorithm.



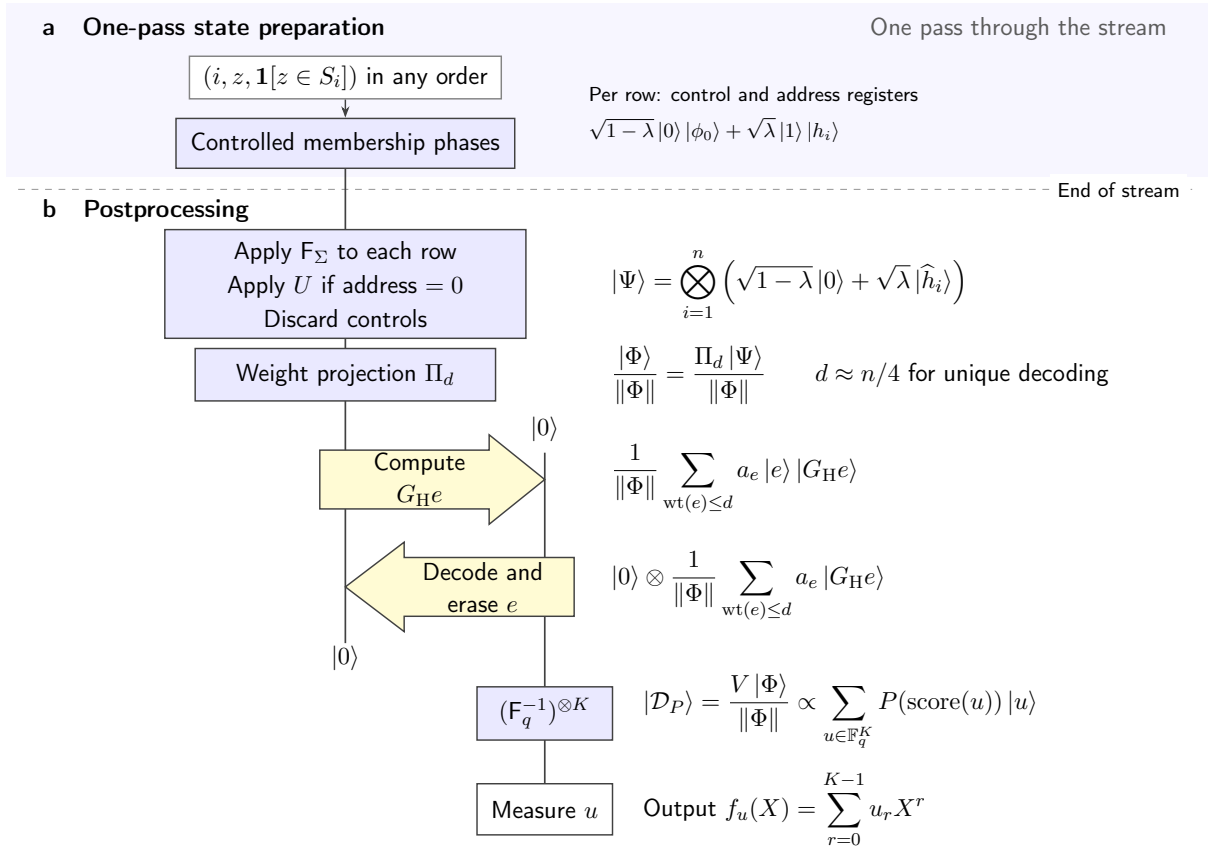
\begin{figure}[htbp]
  \centering
  \resizebox{0.95\linewidth}{!}{\definecolor{dqiOnline}{HTML}{F7F6FF}
\definecolor{dqiGate}{HTML}{EAE9FF}
\definecolor{dqiDecode}{HTML}{FFFBD4}
\definecolor{dqiInk}{HTML}{414141}
\definecolor{dqiNote}{HTML}{5B5B5B}
\definecolor{dqiFilter}{HTML}{1936B5}
\begin{tikzpicture}[
  x=1cm,y=-1cm,
  font=\sffamily\fontsize{10.95}{12.2}\selectfont,
  line width=.65pt,
  dqi gate/.style={
    draw=dqiInk,fill=dqiGate,align=center,
    inner xsep=5pt,inner ysep=5pt},
  dqi state/.style={anchor=west,inner sep=0pt},
  dqi note/.style={
    text=dqiNote,font=\sffamily\fontsize{9}{10.8}\selectfont},
  dqi wire/.style={draw=dqiInk,line width=.75pt},
  dqi arrow/.style={
    draw=dqiInk,-{Stealth[length=1.8mm,width=1.3mm]}}
]
\path[use as bounding box] (0,0) rectangle (18.8,13.9);

\fill[dqiOnline] (0,0) rectangle (18.8,2.72);
\node[anchor=west,font=\sffamily\bfseries] at (.42,.40)
  {a\quad One-pass state preparation};
\node[anchor=east,text=dqiNote] at (18.30,.40)
  {One pass through the stream};
\node[draw=dqiInk!65,fill=white,inner xsep=6pt,inner ysep=4pt]
  (records) at (5.30,1.17)
  {$(i,z,\1[z\in S_i])$ in any order};
\node[dqi gate,minimum width=5.30cm,minimum height=.78cm]
  (load) at (5.30,2.20) {Controlled membership phases};
\draw[dqi arrow] (records.south)--(load.north);
\node[dqi state,align=left,
      font=\sffamily\fontsize{9}{11.5}\selectfont] at (9.10,1.77)
  {Per row: control and address registers\\[4pt]
   $\begin{aligned}
     &\sqrt{1-\lambda}\ket{0}\ket{\phi_0}
     +\sqrt\lambda\ket{1}
       \ket{h_i}
   \end{aligned}$};

\draw[dashed,draw=dqiNote!65,line width=.4pt]
  (.18,2.92)--(18.55,2.92);
\node[anchor=east,fill=white,inner xsep=3pt,
      font=\sffamily\fontsize{8.8}{10.5}\selectfont]
  at (18.48,2.92)
  {End of stream};
\node[anchor=west,font=\sffamily\bfseries] at (.42,3.25)
  {b\quad Postprocessing};

\draw[dqi wire] (load.south)--(5.30,10.02);
\draw[dqi wire] (8.42,7.04)--(8.42,12.62);
\node[dqi gate,minimum width=5.65cm,minimum height=1.30cm]
  at (5.30,4.4)
  {Apply $\Fsf_\Sigma$ to each row\\[3pt]
   Apply $U$ if address $=0$\\[3pt]
   Discard controls};
\node[dqi state] at (9.45,4.4)
  {$\displaystyle \ket{\Psi}=
    \bigotimes_{i=1}^{n}
    \left(\sqrt{1-\lambda}\ket{0}+
          \sqrt\lambda\ket{\widehat h_i}\right)$};

\node[dqi gate,minimum width=4.70cm,minimum height=.82cm]
  at (5.30,5.82) {Weight projection $\Pi_d$};
\node[dqi state] at (9.45,5.82)
  {$\displaystyle\frac{\ket{\Phi}}{\|\Phi\|}
     =\frac{\Pi_d\ket{\Psi}}{\|\Phi\|}$\qquad$d\approx n/4$ for unique decoding};

\node at (8.42,6.67) {$\ket{0}$};
\draw[draw=dqiInk,fill=dqiDecode]
  (4.90,7.00)--(7.55,7.00)--(7.55,6.76)
  --(8.42,7.46)--(7.55,8.16)--(7.55,7.92)
  --(4.90,7.92)--cycle;
\node[align=center] at (6.46,7.46)
  {Compute\\[2pt]$G_{\Hrm}e$};
\node[dqi state] at (9.45,7.46)
  {$\displaystyle\frac1{\|\Phi\|}
    \sum_{\wt(e)\le d}a_e\ket{e}\ket{G_{\Hrm}e}$};

\draw[draw=dqiInk,fill=dqiDecode]
  (8.88,8.68)--(6.22,8.68)--(6.22,8.44)
  --(5.30,9.14)--(6.22,9.84)--(6.22,9.60)
  --(8.88,9.60)--cycle;
\node[align=center] at (7.45,9.14)
  {Decode and\\[2pt]erase $e$};
\node[dqi state] at (9.45,9.14)
  {$\displaystyle \ket{0}\otimes\frac1{\|\Phi\|}
    \sum_{\wt(e)\le d}a_e\ket{G_{\Hrm}e}$};
\node at (5.30,10.25) {$\ket{0}$};

\node[dqi gate,minimum width=2.12cm,minimum height=.82cm]
  at (8.42,11.12) {$(\Fsf_q^{-1})^{\otimes K}$};
\node[dqi state] at (9.95,11.10)
  {$\displaystyle
    \begin{aligned}
      \ket{\Dcal_P}
        &=\frac{V\ket{\Phi}}{\|\Phi\|}
        \propto\sum_{u\in\Fbb_q^K}
        P\!\left(
          \score(u)\right)\ket{u}
    \end{aligned}$};
\node[draw=dqiInk,fill=white,minimum width=2.12cm,
      minimum height=.80cm] at (8.42,12.62) {Measure $u$};
\node[dqi state] at (9.95,12.62)
  {Output $\displaystyle f_u(X)=\sum_{r=0}^{K-1}u_rX^r$};
\end{tikzpicture}}%
  \captionsetup{
    name=Fig.,font=small,labelfont=bf,labelsep=dqibar,
    justification=justified,singlelinecheck=false,skip=7pt
  }
  \caption{Streaming DQI for Hermite OPI.}
  \label{fig:streaming-hermite-dqi}
\end{figure}

Fix $\alpha$ near $\alpha_\star\approx0.93301$ and choose
$\lambda\in(0,1/4)$ such that $\frac12+\sqrt{\lambda(1-\lambda)}=\alpha$.
The choice of $\lambda$ leaves the typical error weight a constant
fraction below the decoding cutoff near $n/4$, while placing the
typical score around $\alpha n$. 
Also for intuition, we assume each $|S_i|=|\Sigma|/2$, which only hides a $1/|\Sigma|$ deviation that goes to zero as $q$ grows.

\paragraph{One-pass state preparation
(\Cref{fig:streaming-hermite-dqi}a and the top of b).}
Our target product state follows the binomial-weight DQI construction in \cite{rosmanis2026}. Define for each $i\in[n]$
\[
 \ket{h_i}=|\Sigma|^{-1/2}\sum_{z\in\Sigma}(-1)^{\1[z\notin S_i]}\ket{z}.
\]
Since we assumed $|S_i|=|\Sigma|/2$ for intuition, $\ket{h_i}$ is a unit vector orthogonal to the uniform state $\ket{\phi_0}=|\Sigma|^{-1/2}\sum_{z\in\Sigma}\ket{z}$.
Define $\ket{\widehat h_i}=\Fsf_\Sigma\ket{h_i}$, where $\Fsf_\Sigma=\Fsf_q^{\otimes\ell}$ is the Fourier transform.
Our goal is to prepare the product error state
\[
 \ket{\Psi}
 =\bigotimes_{i=1}^n
   \left(\sqrt{1-\lambda}\ket{0}
         +\sqrt\lambda\ket{\widehat h_i}\right).
\]
To see why this can be done with one-pass streaming, we initialize an address register to $\ket{\phi_0}$
and a control qubit to
$\sqrt{1-\lambda}\ket{0}+\sqrt\lambda\ket{1}$ for each $i\in[n]$.
Then each stream entry $(i,z,0)$ indicates $z\notin S_i$ and we flip the phase of
$\ket{1}\ket{z}$. After the pass, applying $\Fsf_\Sigma$ to the address register gives $\sqrt{1-\lambda}\ket{0}\ket{0}
 +\sqrt\lambda\ket{1}\ket{\widehat h_i}$.
Since $\braket{0|\widehat h_i}=0$, a controlled-flip produces the desired $\ket\Psi$.

\paragraph{Computing the syndrome and erasing the error
(yellow arrows in \Cref{fig:streaming-hermite-dqi}b).}
Before decoding, set $d\approx n/4$ and project onto errors
of weight at most $d$. Write
\[
 \Pi_d=\sum_{\wt(e)\le d}\ket{e}\bra{e},
 \qquad
 \ket{\Phi}:=\Pi_d\ket{\Psi}
 =\sum_{\wt(e)\le d}a_e\ket{e}.
\]
We measure only the projection's acceptance flag, preserving coherence
among the remaining errors. On acceptance, the state is
$\ket{\Phi}/\|\Phi\|$. Then, following DQI~\cite{jordan2025dqi}, we compute the syndrome and
coherently erase the error register:
\[
 \begin{aligned}
 \frac1{\|\Phi\|}\sum_{\wt(e)\le d}a_e\ket{e}\ket{0}
 &\longmapsto
 \frac1{\|\Phi\|}\sum_{\wt(e)\le d}
   a_e\ket{e}\ket{G_\Hrm e}
 \longmapsto
 \ket{0}\otimes\frac1{\|\Phi\|}
   \sum_{\wt(e)\le d}a_e\ket{G_\Hrm e}.
 \end{aligned}
\]
The erasure is possible by efficient unique decoding (\Cref{thm:hermite-compact-decoder}).

\paragraph{Interference and sampling
(bottom of \Cref{fig:streaming-hermite-dqi}b).}
Applying $(\Fsf_q^{-1})^{\otimes K}$ to the syndrome register
makes the Fourier terms interfere. Let $V$ denote the map that
converts $\ket{e}$ to $\ket{G_\Hrm e}$ and then applies
this inverse Fourier transform. With
$\omega=e^{2\pi\mathrm i/q}$, it satisfies
\[
 V\ket{e}=q^{-K/2}
 \sum_{u\in\Fbb_q^K}\omega^{\langle G_\Hrm^\top u,e\rangle}\ket{u}.
\]
On the projected errors, syndrome compression is an isometry.
Fourier inversion shows that the amplitude depends
on $u$ through a polynomial $P$ of
$\score(u)=\sum_i\1[(G_\Hrm^\top u)_i\in S_i]$ with $\deg P\le d$.
The normalized output state is
\[
 \ket{\Dcal_P}
 =\frac{V\ket{\Phi}}{\|\Phi\|}
 =\frac{q^{-K/2}}{\|\Phi\|}
   \sum_{u\in\Fbb_q^K}P(\score(u))\ket{u}.
\]
Measuring the output register samples $f_u$ with probability
proportional to $|P(\score(u))|^2$.

\paragraph{Why the sampled polynomial is a good approximation.}
The score bias comes from the inverse transform of the product state.
Write $\Fsf_\Hrm=\Fsf_\Sigma^{\otimes n}$.
Measuring $\Fsf_\Hrm^{-1}\ket{\Psi}$ gives independent row
assignments, which need not be the jets of a single polynomial.
Each row satisfies its constraint with probability
\[
\left(\sqrt{\lambda/2}
          +\sqrt{(1-\lambda)/2}\right)^2=\alpha.
\]
Hence their total score concentrates around $\alpha n$.
Since the truncation happens at $d\approx n/4\ge\lambda n$, $\ket\Phi$ is close to $\ket\Psi$. This intuitively shows that the concentration from $\ket\Psi$ transfers to $\ket\Phi$. To make this rigorous, a second-moment calculation similar to \cite{jordan2025dqi} is needed.

Finally to see why $\alpha_\star$ is the right threshold, the decoding cutoff allows any fixed $\lambda<1/4$, and
\[
 \alpha=\frac12+\sqrt{\lambda(1-\lambda)}
 \longrightarrow\alpha_\star=\frac12+\frac{\sqrt3}{4}\approx0.93301
 \qquad\text{as }\lambda\to\frac14.
\]
Hence every fixed approximation ratio $\alpha<\alpha_\star$ is
achievable by choosing $\lambda$ sufficiently close to $1/4$.

\subsection{Doubly efficient quantum streaming at \texorpdfstring{$1-\epsilon$}{1-epsilon} approximation}

Following \Cref{subsec:tri}, we still prepare in one pass the product state
\[
 \ket{\Psi}
 =\bigotimes_{i=1}^n
   \left(\sqrt{1-\lambda}\ket{0}
         +\sqrt\lambda\ket{\widehat h_i}\right).
\]
For approximation ratio $\alpha$ approaching $1$, we need to choose $\lambda\approx1/2$ such that $\frac12+\sqrt{\lambda(1-\lambda)}=\alpha$. Set $d\approx\lambda n$ as the projection cutoff.
The Fourier weight remains concentrated below $d$, and the
$\ket\Psi$ score is still concentrated around $\alpha\cdot n$.
The difficulty is that errors in this larger weight range may
share a syndrome, preventing unique decoding.

\paragraph{Resolving ambiguity with a few coordinates.}
Write $B_t=\{e\in\Sigma^n:\wt(e)\le t\}$.
Our list-decoding analysis bounds every syndrome list in
$B_{d+2}$ by $2^{o(n)}$ (\Cref{lem:nearopt-dual-lists}).
Let $r=o(n)$. Choose a uniform $r$-subset $R\subseteq[n]$ and append the error blocks on $R$:
\[
 H_Re=(G_\Hrm e,e_R),
 \qquad
 H_R=\begin{pmatrix}G_\Hrm\\E_R\end{pmatrix}.
\]
Two distinct errors with the same original syndrome differ on
more than half the blocks, since their difference is a nonzero
dual codeword.
Consequently, for fixed $e\in B_d$, the probability that
some competitor in $B_{d+2}$ also agrees with $e$ on $R$ is at most
$2^{o(n)}\cdot2^{-r}$.
Taking $r=o(n)$ makes this probability negligible.

We accept only errors in $B_d$ that are uniquely identified by their augmented syndrome within $B_{d+2}$. Averaging over $R$, this filter discards only small probability mass.
To decode $\ket e$ from $\ket{H_Re}$ for accepted errors $e$, we do not have an efficient algorithm and  have to use exhaustive search to enumerate all possible $e$'s. This is time inefficient but the space bound remains the same.

\paragraph{Extra coordinates cost little.}
After inverse Fourier transformation, each $(u,v)$
corresponds to
\[
 H_R^\top(u,v)
 =G_\Hrm^\top u+E_R^\top v.
\]
The additional term changes only the coordinates in $R$.
Thus, writing $S(z)=\sum_i\1[z_i\in S_i]$, we can discard
$v$ and lose at most $r$ constraints. Since $r=o(n)$, we have
\[
 \score(u)\ge S\left(H_R^\top(u,v)\right)-r\approx S(H_R^\top(u,v)).
\]
The sampled coordinates therefore resolve decoding ambiguity
and ensure the final score is roughly $\alpha\cdot n$.

\subsection{Classical lower bounds}

We prove the streaming lower bound using multi-party communication complexity. This turns out to be a challenging task. First, most of the current communication lower bound techniques only work for boolean functions (or relations with few solutions), but HOPI has exponentially many solutions for a typical random input. 
To handle this, we use the gadgetless lifting framework developed by \cite{yangzhang2024,maoyangzhang2024}. 
Furthermore, in the gadgetless lifting framework, the choice of the error-correcting code $C_\Hrm$ for HOPI is a Hermite code, whose coding-theoretic properties (in particular, list recovery bounds) are not well established; and we extend the hidden-derivative interpolation of~\cite{bcpzz2026,jeronimo2026} for the analysis.

Below, we give an overview of our proof with presentation adapted from~\cite{ggjl2025}. In the $k$-party communication problem, we partition complete lists $S_1,...,S_n$ among $k$ parties, each holding $n/k$ input lists. For intuition, first consider a simplified input distribution in which the membership bits $\1[z\in S_i]$, indexed by $(i,z)\in[n]\times\Sigma$, are independent and uniform. For $S=(S_1,\ldots,S_n)$ and $c\in C_\Hrm$, write $\err_S(c)=|\{i\in[n]:c_i\notin S_i\}|$. We sketch a lower bound for finding a codeword $c$ with $\err_S(c)\le(1-\alpha)n$. This completes the proof as it is easy to show that the optimum is $n$ for a random $S$ (\Cref{app:classical-random-sat}).

For a communication protocol succeeding with probability at least $2/3$, averaging lets us fix its random coins while maintaining success probability at
least $1/2$ over the input distribution. We require the protocol to announce its output and the membership bits queried by that output, and to return $\bot$ if its score is below $\alpha n$. This verification costs $O_\ell(n\log n)$ additional bits and is performed before the protocol refinement below.
Recall that the stream consists of $N$ entries $(i,z,\1[z\in S_i])$.
For $X\subseteq\{0,1\}^N$ and $I\subseteq[N]$, we use $X_I:=\{x_I\in\{0,1\}^I:x\in X\}$ to denote the set of strings in $X$ restricted to the coordinates in $I$. 

\paragraph{Special case: subcube protocols.}
We start with highly structured protocols. 
For each $a\in[k]$, let $N_a\subset[N]$ be the set of stream bits held by party $a$.
A rectangle $R=X_1\times\cdots\times X_k$, where
$X_a\subseteq\{0,1\}^{N_a}$, is a subcube rectangle if each $X_a$ fixes some coordinates $I_a\subseteq N_a$ and leaves all other coordinates unrestricted. Its codimension is
$\codim(R)=\sum_a|I_a|$.
A subcube protocol is such that every complete-message node of the protocol tree corresponds to a subcube rectangle. This model is closely related to decision trees, but is strictly more general. The following property of the code $C_\Hrm$ will be central to the lower bound.

\begin{definition}[List recovery]
\label{cl:overview-lr-definition}
For $\zeta\in(1/2,1)$ and $T>n$, the code $C_\Hrm\subseteq\Sigma^n$ is
$(\zeta,T,L)$-list-recoverable if,
for every collection $F_1,\ldots,F_n\subseteq\Sigma$ with
$\sum_i|F_i|\le T$,
$$
 \abs{
 \left\{
 x\in C_\Hrm:
 |\{i\in[n]:x_i\in F_i\}|\ge\lceil\zeta n\rceil
 \right\}
 }
 \le L.
$$
\end{definition}

Consider a deterministic subcube protocol $\Pi$ of communication cost $|\Pi|$, run on uniform input distribution. For a rectangle $R$, Let $F_i(R)$ be the set of all $z\in\Sigma$ for which the membership bit $\1[z\in S_i]$ is fixed on every input in $R$. Then
$\sum_i|F_i(R)|=\codim(R)$.
Thus list recovery applies as long as
$\codim(R)\le T$.

The communication bound $|\Pi|$ controls the probability of exceeding this budget $T$. Let $R_{\mathrm{leaf}}$ be the random leaf reached by
$\Pi$. Under uniform input distribution, a leaf $R$ has probability $2^{-\codim(R)}$, and hence $\E[\codim(R_{\mathrm{leaf}})]=\ent(R_{\mathrm{leaf}})\le|\Pi|$, where $\ent$ denotes binary Shannon entropy. Since the exposed sets only grow along an execution, the exposure budget is exceeded if and only if $\codim(R_{\mathrm{leaf}})>T$. By Markov's inequality,
\begin{equation}\label{eq:overview-codim}
 \Pr[\codim(R_{\mathrm{leaf}})>T]
 \le \frac{\E[\codim(R_{\mathrm{leaf}})]}{T}
 \le \frac{|\Pi|}{T}.
\end{equation}

We analyze each execution up to, but excluding, its first node that exceeds the budget $T$. A codeword $c\in C_\Hrm$ is \emph{dangerous} at $R$
if at least $\lceil\zeta n\rceil$ of its
queried addresses $(i,c_i)$ have their membership bits fixed by $R$. We use two key facts.
\begin{itemize}
\item \emph{Few dangerous codewords.}
Fix an execution and let $R_*$ be its last rectangle before the budget $T$ is exceeded, or its leaf if the budget is never exceeded.
Since the exposed sets only grow, every codeword that becomes dangerous up to this point remains dangerous at $R_*$. Moreover, $\sum_i|F_i(R_*)|\le T$. \Cref{cl:overview-lr-definition} therefore implies that at most $L$ distinct codewords become dangerous along the entire
execution before the budget is exceeded.

\item \emph{A newly dangerous codeword is unlikely to be a solution.}
Suppose a message first makes a codeword $c$ dangerous, and let $R$ be the rectangle reached after this message, still within the budget.
Immediately before the message, fewer than $\lceil\zeta n\rceil$ of the membership bits queried by $c$ were fixed. The message changes only the sender's factor. Since the sender holds $n/k$ input lists, at most $n/k$ of the bits queried by $c$ belong to the sender. Hence at least 
\begin{equation}\label{eq:overview-tau}
\tau:=n-\frac nk-\lceil\zeta n\rceil\approx\left(1-\frac1k-\zeta\right)n
\end{equation}
queried membership bits in the other parties' inputs remain unfixed at $R$. Conditional on reaching $R$, these bits are independent and uniform.
Fix any $\tau$ of them. The number of zeros among them has distribution $\operatorname{Bin}(\tau,1/2)$, and each zero contributes one to
$\err_S(c)$. Therefore
\begin{equation}\label{eq:overview-eta}
 \Pr[\err_S(c)\le(1-\alpha)n\mid R]
 \le\eta
 \quad\text{where}\quad
 \eta
 :=\Pr[\operatorname{Bin}(\tau,1/2)\le(1-\alpha)n].
\end{equation}
\end{itemize}

Along each execution, order the codewords that become dangerous, before the budget exceeded, by their first appearance.
For each $j\in[L]$, let $E_j$ be the event that the $j$th codeword in the list recovery bound exists and
has $\err_S$ value at most $(1-\alpha)n$.
By the law of total expectation and \eqref{eq:overview-eta}, we have $\Pr[E_j]\le\eta$.
A union bound therefore gives $\Pr\!\left[\exists j\in[L],~E_j\right] \le L\eta$. In words, the probability that any solution codeword becoming dangerous before exceeding the budget is at most $L\eta$.

At a successful leaf, verification exposes every query of the output codeword. Thus, if the budget remains, its successful output is among the dangerous candidates. Combining this with \eqref{eq:overview-codim}, we obtain
\begin{equation}
 \Pr[\Pi\text{ succeeds}]
 \le \frac{|\Pi|}{T}+L\eta.
 \label{cl:overview-subcube-tradeoff}
\end{equation}
Thus, constant success probability and $L\eta=o(1)$ would imply $|\Pi|=\Omega(T)$.

\paragraph{Structure-vs-randomness dichotomy.}
To prove a lower bound for general protocols, we use the
structure-vs-randomness framework developed by Yang and
Zhang \cite{yangzhang2024}, who build upon query-to-communication lifting techniques~\cite{gpw2020}.
We use this framework to refine any general communication protocol
into a subcube-like execution. The nodes of a subcube-like execution
correspond to subcube-like rectangles
$R=X_1\times\cdots\times X_k$ defined such that
\begin{itemize}[itemsep=0pt,topsep=4pt]
 \item $(X_a)_{I_a}$ are fixed strings for some index sets
 $I_a\subseteq N_a$. \hfill (Structure)

 \item $(X_a)_{N_a\setminus I_a}$ are pseudorandom (they ``look'' like
 the uniform distribution). \hfill (Randomness)
 
 More precisely, they are $\gamma$-dense:
every assignment to any $J\subseteq N_a\setminus I_a$ has mass at most $2^{-\gamma|J|}$. Hence the probability that the $\tau$ unfixed nonsender queries have at most $(1-\alpha)n$ failures is at
most 
\begin{equation}\label{eq:overview-eps-gamma}
\epsilon_\gamma=2^{-\gamma\tau}\binom{\tau}{\le(1-\alpha)n}=2^{(1-\gamma)\tau}\Pr[\operatorname{Bin}(\tau,1/2)\le(1-\alpha)n]=2^{(1-\gamma)\tau}\eta.
\end{equation}
\end{itemize}

This shows that subcube-like protocols behave similarly enough to subcube protocols that we may simply re-do the proof sketched above for subcube protocols, replacing $\eta$ by $\epsilon_\gamma$.
We will also prove that the refinement to obtain the subcube-like protocols incur an $O(1/(1-\gamma))$ blowup in expectation, independently of the number of parties (\Cref{cl:preprocessing}).
By \eqref{eq:overview-tau}, \eqref{eq:overview-eps-gamma}, and \eqref{eq:overview-eta}, we therefore obtain
\begin{equation}\label{eq:after-refine}
 \Pr[\Pi\text{ succeeds}]
 \le
 O\left(\frac{|\Pi|}{(1-\gamma)T}\right)
 +2^{(1-\gamma)\tau}\cdot L\cdot\Pr[\operatorname{Bin}(\tau,1/2)\le(1-\alpha)n],
\end{equation}
where
\begin{equation}\label{eq:overview-tau-2}
\tau\approx\left(1-\frac1k-\zeta\right)n.
\end{equation}

\paragraph{List recovery of the Hermite code.}
We have been ignoring a non-trivial issue---the list recovery bound required by our argument must be proved for the Hermite code $C_\Hrm$.
We obtain it by extending the hidden-derivative interpolation of~\cite{bcpzz2026,jeronimo2026}. The change is that each stream entry now specifies a length-$\ell$ jet, rather than one polynomial value.
We impose these jet constraints on an auxiliary polynomial $Q(X,Y_0,\ldots,Y_r)$, so that every high-score candidate $P$ satisfies $Q(X,P,P^{(1)},\ldots,P^{(r)})=0$. The differential-root bound of~\cite[Theorem~4.3]{kopparty2015} then bounds the number of candidates. 
For \Cref{cl:overview-lr-definition}, our Hermite extension (\Cref{hendpoint:lr}) gives
\begin{equation}\label{eq:overview-lr}
\zeta=2\alpha-1-\Theta\left(\frac1{\log q}\right),
\qquad
 T=
 \widetilde{\Omega}_{\ell,\alpha}
 \left(q^{1+\ell\cdot\kappa(\alpha)} \right),
 \qquad
 L=2^{o_{\ell,\alpha}(n)},
\end{equation}
where $\kappa(\alpha)=(4\alpha-3)/(4\alpha-2)$.

Our discussion so far has been about a classical communication lower bound. To obtain a streaming lower bound, take $k=\Theta(\log q)$ parties, which ensures $\tau>2\cdot(1-\alpha)n$ by \eqref{eq:overview-tau-2}.
Then choosing $\gamma=1-o_{\ell,\alpha}(1)$ appropriately ensures the second term in \eqref{eq:after-refine} is $L\cdot2^{-o_{\ell,\alpha}(n)}=o(1)$ by \eqref{eq:overview-lr}.
Since the verified protocol succeeds with probability at least $1/2$, \eqref{eq:after-refine} requires $|\Pi|=\Omega((1-\gamma)\cdot T)=\widetilde\Omega(T)$ bits. Finally, a $\npass$-pass streaming algorithm using $S$ bits gives a protocol with $(k\npass-1)S$ communicated bits. 
Thus $S=\widetilde{\Omega}_{\ell,\alpha}\left( q^{1+\ell\cdot\kappa(\alpha)}/\npass\right)$ as claimed.

\subsection{Classical upper bounds}

Finally we give a sketch of our classical streaming algorithms.

Our algorithm for \Cref{thm:classical-upper-summary}  independently keeps each positive address $(i,z,1)$ with rate $\eta$ during a single pass over the stream. After the pass, it exhaustively searches for a polynomial matching at least the acceptance threshold $h_0$ sampled ones.  
The analysis has two parts: showing that every accepted polynomial has high score, and showing that some polynomial is accepted. For the first, Hermite list recovery bounds the number of polynomials at each score. Combining this bound with the sampling rate shows that, with high probability, no polynomial of score below $\alpha n$ matches enough sampled ones.  For the second, we consider sets of $h$ positive addresses belonging to high-score polynomials, which we call certificates. There are many such certificates, and a second-moment argument controls their overlap to show that a complete certificate survives with a small but sufficient probability. A concentration inequality then shows that an almost complete certificate survives with high probability. We therefore set the acceptance threshold $h_0$ slightly below $h$. Together, these two claims ensure that the algorithm returns a polynomial of score at least $\alpha n\ge\alpha\OPT$.  Balancing the sampling rate $\eta$ and threshold $h_0$ yields an $\alpha$-approximation claimed in \Cref{thm:classical-upper-summary}.

Our algorithm for \Cref{thm:classical-three-quarters-summary} combines interpolation with random completion in a single pass. First, it collects one allowed jet at each of $k=n/2-4$ distinct evaluation points. Whenever a positive entry $(i,z,1)$ arrives at a point $i$ not yet selected, it selects $i$ and records $y_i=z\in S_i$. Once $k$ points have been selected, it stops collecting constraints; let $I$ be the set of selected points.
Next, it constructs a constant number $R$ of independent random candidate polynomials of degree less than $K$, each satisfying $J_\ell P(i)=y_i$ for $i\in I$. Hermite interpolation imposes $\ell k$ linear constraints on the $K$ coefficients, leaving $K-\ell k=4\ell$ free parameters. Sampling these parameters uniformly gives a uniformly random polynomial satisfying all the selected constraints.
The candidates can now be scored exactly using the rest of the stream. Each candidate already satisfies the $k$ selected constraints. For each remaining positive entry $(i,z,1)$ with $i\notin I$, the algorithm scores if $J_\ell P(i)=z$. The algorithm returns the candidate with the largest counter.
At each remaining point $i$, a random completion has a uniform jet in $\Sigma$, which belongs to $S_i$ with probability roughly $1/2$. The expected score is therefore roughly $k+n/4=3n/4$.

\section{Preliminaries}\label{sec:preliminaries}

Throughout the paper, we reserve $q$ and $\ell$ for the field size and the jet length in \Cref{prob:hopi} respectively. We also reserve $\Sigma=\Fbb_q^\ell$ as the alphabet of the length-$\ell$ jet. We will use $0$ to denote the zero element and its meaning may vary in different settings; for example, $0\in\Sigma$ means the length-$\ell$ vector where each entry is $0\in\Fbb_q$. We use $\propto$ to denote that the LHS quantity is proportional to the RHS. We use $\|\cdot\|$ to denote the $\ell_2$ norm of a vector and operator norm of a matrix.

\paragraph{Fourier transform.}
We identify $\Fbb_q$ as $\Zbb/q\Zbb$ and use $\Fsf_q$ to denote the Fourier transform over $\Fbb_q$, with a negative-sign convention: for each $c\in\Fbb_q$, we have $\Fsf_q\ket{c} =q^{-1/2}\sum_{e\in\Fbb_q}\omega^{-e\cdot c}\ket{e}$ where $\omega=e^{2\pi\mathrm i/q}$.
We write $\Fsf_\Sigma=\Fsf_q^{\otimes\ell}$.

\paragraph{Block weight, distance, and inner product.}
Fix $x\in\Sigma^n$.
For $i\in[n]$, we use $x_i$ to denote its block symbol in $\Sigma=\Fbb_q^\ell$; in addition for $j\in[\ell]$, we use $x_{i,j}\in\Fbb_q$ to denote the $j$th coordinate symbol of $x_i$.
We define $\wt(x)=|\{i\in[n]\colon x_i\ne0\}|$ as the block weight of $x$.

Fix another $y\in\Sigma^n$. We use $d(x,y)=\wt(x-y)$ to denote the block distance between $x$ and $y$.
We use $\langle x,y\rangle=\sum_{i\in[n]}\sum_{j\in[\ell]}x_{i,j}y_{i,j}\in\Fbb_q$ as their inner product.

Fix a matrix $A\in\Sigma^{k\times n}$. For $u\in\Fbb_q^k$, we use $A^\top u\in\Sigma^n$ to denote the linear combination of rows of $A$ using coefficients of $u$: $(A^\top u)_i=\sum_{j\in[k]}u_jA_{j,i}$.
For $x\in\Sigma^n$, we use $Ax\in\Fbb_q^k$ to denote the inner product of rows of $A$ with $x$: $(Ax)_i=\sum_{j\in[n]}\langle A_{i,j},x_j\rangle$.
Define $\im(A)=\left\{A^\top u\colon u\in\Fbb_q^k\right\}$ and $\ker(A)=\left\{x\in\Sigma^n\colon Ax=0\right\}$.

\paragraph{Error-correcting codes.}
All codes in this paper are linear.
A (linear) code of length $n$ over $\Sigma$ is a nonzero $\Fbb_q$-linear subspace $C\subseteq\Sigma^n$, where every $x\in C$ is a codeword. The block distance of $C$ is defined by $d(C)=
\min_{0\ne x\in C}\wt(x)$.
We say $G\in\Sigma^{k\times n}$ is a generator matrix of $C$ if $C=\im(G)$; and say $H\in\Sigma^{k\times n}$ is a parity-check matrix of $C$ if $C=\ker(H)$.
The dual code of $C$ is $C^\perp=\{y\in\Sigma^n\colon\langle x,y\rangle=0\text{ for all }x\in C\}$.
It is standard that if $G$ is the generator matrix of $C$, then $G$ is the parity-check matrix of $C^\perp$.

A radius-$t$ unique decoder for $C$, given input $y=c+e$ with $\wt(e)\le t$, outputs $c\in C$, which is unique whenever $t<d(C)/2$. 
We say $C$ is $(t,L)$-list-decodable if, for every $y\in\Sigma^n$, $|\{c\in C\colon d(c,y)\le t\}|\le L$; and a radius-$t$ list decoder outputs all codewords in this set.

For $\zeta\in(0,1]$ and sets $F=(F_1,\ldots,F_n)$ where $F_1,\ldots,F_n\subseteq\Sigma$, define
\[
 \Lcal_\zeta(F)=\{c\in C:|\{i\in[n]:c_i\in F_i\}|\ge\lceil\zeta n\rceil\}.
\]
We say $C$ is $(\zeta,T,L)$-list-recoverable if $|\Lcal_\zeta(F)|\le L$ holds for every family with $\sum_{i\in[n]}|F_i|\le T$.

\paragraph{Syndrome decoding.}
For each $e\in\Sigma^n$, the syndrome of $e$ given a parity-check matrix $H\in\Sigma^{k\times n}$ is $He\in\Fbb_q^k$ where the $i$th coordinate of $He$ is the inner product of $e$ and the $i$th row of $H$.
It is also immediate that if $H$ is the parity-check matrix of $C$, then $Hx=0$ for every $x\in C$.
At radius $t$ and syndrome $s\in\Fbb_q^k$, the syndrome decoding finds errors in
\[
 \Ecal_t(s)=\{e\in\Sigma^n\colon He=s,\ \wt(e)\le t\}.
\]
Observe that every $e'\in\Sigma^n$ and every $e\in\Ecal_t(He')$ give a codeword $e'-e\in C$. 

\paragraph{Hermite interpolation.}
For each integer $K\ge0$, let
$\Pcal_K=\{f\in\Fbb_q[X]:\deg(f)<K\}$ with
$\Pcal_0=\{0\}$.
The jet $J_\ell f(a)$ defined in \Cref{sec:intro-setup} specifies $f$ modulo $(X-a)^\ell$, which leads to the following simple fact.

\begin{fact}[See e.g., {\cite[Theorem~3.1]{CampsEtAl26}}]\label{fact:hermite-interpolation}
Let $K\ge1$. For any $s$ distinct points $a_1,\ldots,a_s\in\Fbb_q$
with $\ell s\le K$ and any assignment of their length-$\ell$ jets, there is a unique $f_0\in\Pcal_{\ell s}$ such that every $f\in\Pcal_K$ satisfying these jets can be written as 
\[
 f=f_0+v\cdot g
 \quad\text{where}\quad
 v(X)=\prod_{i=1}^s(X-a_i)^\ell
 \quad\text{and}\quad
 g\in\Pcal_{K-\ell s}.
\]
Consequently, the jets of a uniform $f\in\Pcal_K$ at any $\lfloor K/\ell\rfloor$ distinct points are independently uniform in $\Sigma$.
\end{fact}

\paragraph{Hermite code.}
Recall Hermite OPI from \Cref{prob:hopi}. It naturally gives the following Hermite code.

\begin{definition}[Hermite code]\label{def:hermite-code}
For distinct $a_1,\ldots,a_n\in\Fbb_q$ and $K<\ell n$,
define the Hermite code $C_\Hrm$ by
\[
 C_\Hrm
 =\{(J_{\ell}f(a_1),\ldots,J_{\ell}f(a_n)):f\in\Pcal_K\}
 \subseteq\Sigma^n.
\]
We use $G_\Hrm\in\Sigma^{K\times n}$ to denote its generator matrix in the monomial basis: the $r$th\footnote{We conveniently index rows of $G_\Hrm$ by $0,1,\ldots,K-1$.} row of $G_\Hrm$ is the length-$\ell$ jets of $X^r$ evaluated at $a_1,\ldots,a_n$.
Using $G_\Hrm$, the codeword of $f_u(X)=\sum_{r=0}^{K-1}u_rX^r$ is simply $G_\Hrm^\top u=(J_\ell f_u(a_1),\ldots,J_\ell f_u(a_n))$.
\end{definition}

For $\ell=1$, \Cref{def:hermite-code} is the Reed--Solomon code. For general $\ell$, it is the univariate multiplicity code~\cite{kopparty2015}; we call it Hermite code to be consistent with the Hermite OPI problem.
This coding-theoretic view will be very helpful, and we give the following formulation of streaming HOPI, equivalent to \Cref{prob:streaming-hopi}.

\begin{problem}[Streaming HOPI by Hermite code]\label{prob:stream-hopi_code}
Let prime $q\ge3$, integer $\ell\ge1$, and $\alpha\in(0,1]$. Define $n=q-1$.
A streaming $(\alpha,q,\ell)$-HOPI instance is given by $N=n\cdot q^\ell$ classical entries $(i,z,\1[z\in S_i])$ for each $i\in\Fbb_q^\times=[n]$ and $z\in\Fbb_q^\ell$, which arrive one at a time and each $S_i$ has size $(q^\ell-1)/2$.

Let $C_\Hrm$ be the Hermite code with $a_i=i\in\Fbb_q^\times$ and $K=\ell n/2$. Recall its generator matrix $G_\Hrm$ in \Cref{def:hermite-code}.
For each $u\in\Fbb_q^K$, define its score as the number of coordinates $i\in[n]$ of $G_\Hrm^\top u$ that lies in its prescribed $S_i$:
$$
\score(u)=\sum_{i=1}^n\1\left[(G_\Hrm^\top u)_i\in S_i\right].
$$
The goal is to output some $u$ such that $\score(u)\ge\alpha\cdot\OPT$ where $\OPT=\max_{u'}\score(u')$.
\end{problem}

\begin{fact}\label{fct:hopi-equal}
\Cref{prob:streaming-hopi} is equivalent to \Cref{prob:stream-hopi_code}.
\end{fact}

\paragraph{Properties of Hermite codes.}

Unique decoding of univariate multiplicity codes goes back to~\cite{Nie01}; see~\cite[Sections~3.1 and~3.1.1]{Kop14} for its rational-reconstruction formulation. We use the corresponding dual-syndrome formulation, whose implementation uses the extended Euclidean algorithm~\cite{sugiyama1975euclidean} and has the desired time and space bounds.

\begin{theorem}[\cite{Kop14,sugiyama1975euclidean}]\label{thm:hermite-compact-decoder}
In $O(n\ell\log q)$ space and $\poly(n,\ell,\log q)$ time, there is a classical deterministic algorithm for $\Dec_\Hrm\colon\Fbb_q^K\to\Sigma^n$ such that $\Dec_\Hrm(G_\Hrm e)=e$ whenever $e\in\Sigma^n$ satisfies $\wt(e)<d(C_\Hrm^\perp)/2$, where $d(C_\Hrm^{\perp})=\lfloor K/\ell\rfloor+1$ is the block distance of $C_\Hrm^\perp$.
\end{theorem}

We also note a duality between $C_\Hrm$ and $C_\Hrm^\perp$ and a list decoding bound.

\begin{fact}[{\cite[Theorem~5.6]{CampsEtAl26}}]
\label{fct:hermite-duality}
For $1\le K<\ell n$, there are invertible $\Fbb_q$-linear maps $T_i\colon\Sigma\to\Sigma$
such that $C_\Hrm^\perp
 =\left\{\bigl(T_iJ_\ell g(a_i)\bigr)_{i=1}^n:
             g\in\Pcal_{\ell n-K}\right\}$.
The map $T=\bigoplus_iT_i$ preserves block distance and list-decoding bounds. In particular, if $K=\ell n/2$,
then $C_\Hrm^\perp=TC_\Hrm$.
\end{fact}

\begin{theorem}[{\cite[Theorem~3.8]{KoppartyRonZewiSarafWootters2023}}]
\label{thm:krsw-list-recovery}
Let $0<\gamma<1/2$, prime $q\ge4/\gamma$, and integer
$\ell\ge16/\gamma^2$. Set $n=q-1$ and $K=\ell n/2$.
For the Hermite code $C_\Hrm$ evaluated at $\Fbb_q^\times$, every
$y\in\Sigma^n$ satisfies
\[
 \abs{\{c\in C_\Hrm:d(c,y)\le(1/2-\gamma)n\}}
 \le(\ell/\gamma)^{O(\ell/\gamma^2)}.
\]
\end{theorem}

The following list recovery bound will be crucial for both our algorithms and lower bounds. The proof is given in \Cref{app:classical-list-recovery}, extending the hidden-derivative interpolation of~\cite{bcpzz2026,jeronimo2026}
to length-$\ell$ jets.

\begin{restatable}{theorem}{hendpointlr}\label{hendpoint:lr}
Assume $K=\ell(q-1)/2$ in $C_\Hrm$.
For every integer $\ell\ge1$, $\alpha\in(3/4,1]$, and prime $q\ge\Omega_\alpha(1)$, $C_\Hrm$ evaluated at $\Fbb_q^\times$ is
$(\zeta,T,L)$-list-recoverable, where
\begin{equation}
\zeta=2\alpha-1-\frac{10}{\ln q},
\quad
 T=q^{1+\ell\cdot\kappa(\alpha)}
 \cdot(\ell\log q)^{-O_\alpha(\ell)},
 \quad\text{and}\quad
\log L=O\left(\frac q{\log^3q}\right).
 \label{hendpoint:incidence-budget}
\end{equation}
\end{restatable}

\section{Triply efficient quantum streaming below \texorpdfstring{$\alpha_\star$}{alpha-star} approximation}
\label{sec:quantum-streaming}

We prove \Cref{thm:quantum-poly} by adapting DQI~\cite{jordan2025dqi} to the streaming setting.

\QuantumStreamingUpper*

During the streaming phase, the algorithm in \Cref{thm:quantum-poly} encodes the stream data into a carefully designed quantum state.
Then in the postprocessing phase, the quantum state is restricted to errors that can be decoded uniquely, and finally we use Fourier interference to sample an answer with high score.

\paragraph{Section organization.}
In \Cref{subsec:quantum-algorithm}, we explain these steps assuming access to the whole input sets and prove the approximation guarantee. In \Cref{subsec:quantum-loading}, we implement the algorithm in the streaming setting and complete the proof of \Cref{thm:quantum-poly}.

\subsection{The algorithm with access to the whole input}
\label{subsec:quantum-algorithm}

We follow the notation convention from \Cref{prob:stream-hopi_code}: $n=q-1$, $N=n\cdot q^\ell$, $K=\ell n/2$, elements in $\Fbb_q^\times$ are named by $i\in[n]$, and $C_\Hrm$ is the corresponding Hermite code with generator matrix $G_\Hrm$. 

\paragraph{Working backward from the desired state.}
In light of \Cref{prob:stream-hopi_code}, we would like to prepare
$$
 \ket{\Dcal_P}
 \propto\sum_{u\in\Fbb_q^K}P(\score(u))\ket{u},
$$
where $P$ favors higher score.
Measuring $\ket{\Dcal_P}$ samples $u$ with probability proportional to $|P(\score(u))|^2$. We will choose $P$ so that with high probability, the output satisfies
\begin{equation}\label{eq:backward_u_poly}
 \score(u)\ge\alpha\cdot n\ge\alpha\cdot\OPT.
\end{equation}

Let $P$ be a polynomial of degree at most $d\le n$.
Recall that $\score(u)=\sum_{i=1}^n\1[(G_\Hrm^\top u)_i\in S_i]$.
By Fourier decomposition, we can express $P(\score(u))$ as
\begin{align*}
P(\score(u))
&=\sum_{e\in\Sigma^n}a_e\cdot\omega^{\langle e,G_\Hrm^\top u\rangle}
\tag{$\omega=e^{2\pi\mathrm i/q}$ and $a_e$'s are coefficients}\\
&=\sum_{\wt(e)\le d}a_e\cdot\omega^{\langle e,G_\Hrm^\top u\rangle}=\sum_{\wt(e)\le d}a_e\cdot\omega^{\langle u,G_\Hrm e\rangle}.
\tag{since $\deg(P)\le d$}
\end{align*}
Consequently, it suffices to prepare the following dual state
\begin{equation}
 \Fsf_q^{\otimes K}\ket{\Dcal_P}
 \propto\sum_{\wt(e)\le d}a_e\ket{G_\Hrm e}.
 \label{eq:dqi-target-fourier}
\end{equation}
Ideally, this state can be constructed by first preparing $\sum_{\wt(e)\le d}a_e\ket{e}$, then computing $\ket{G_\Hrm e}$ to obtain $\sum_{\wt(e)\le d}a_e\ket{e}\ket{G_\Hrm e}$, and finally using \Cref{thm:hermite-compact-decoder} to erase $\ket{e}$ by $\ket{G_\Hrm e}$. Since $K=\ell n/2$, to make this work, we need
\begin{enumerate}[label=(\arabic*),itemsep=0pt,topsep=4pt]
\item\label{itm:backward_1}
$2d<n/2+1$ due to \Cref{thm:hermite-compact-decoder},
\item\label{itm:backward_2} $\sum_{\wt(e)\le d}a_e\ket{e}$ can be efficiently prepared,
\item\label{itm:backward_3} the resulting $P$ favors higher $\score(u)$ which produces the desired \eqref{eq:backward_u_poly}.
\end{enumerate}

\paragraph{A natural product state.}
We choose the error state directly. For each $i\in[n]$, define
\begin{equation}
 g_i(z)=\frac{\1[z\in S_i]-\rho}{\sqrt{\rho(1-\rho)}}
 \quad\text{and}\quad 
\ket{h_i}=|\Sigma|^{-1/2}\sum_{z\in\Sigma}g_i(z)\ket{z},
\quad\text{where $\rho=\frac{|S_i|}{|\Sigma|}=\frac{q^\ell-1}{2\cdot q^\ell}$.}
 \label{eq:dqi-centered-membership}
\end{equation}
Intuitively, $\ket{h_i}$ puts positive mass on symbols in $S_i$ and penalizes symbols outside $S_i$, which is consistent with the score function. 
Note that $\ket{h_i}$ is a
unit vector orthogonal to the uniform state
$$
\ket{\phi_0}=|\Sigma|^{-1/2}\sum_{z\in\Sigma}\ket{z}.
$$
Thus $\ket{\widehat h_i}:=\Fsf_\Sigma\ket{h_i}$ is a unit
vector orthogonal to $\ket{0}$.
Let $\theta\in[0,1]$ be a parameter to be optimized later.
A natural choice for $\sum_ea_e\ket{e}$ is the product state
\begin{equation}
 \ket{\Psi}=\bigotimes_{i=1}^n
 \left(\sqrt{1-\theta}\ket{0}+
       \sqrt\theta\ket{\widehat h_i}\right),
 \label{eq:dqi-product-state}
\end{equation}
which has block weight distribution $\operatorname{Bin}(n,\theta)$.
Hence most of its mass is concentrated near $\wt(e)=\theta\cdot n$.
In light of \Cref{itm:backward_1} and \Cref{itm:backward_2} above, we should pick $\theta<1/4$.

\paragraph{A heuristic for $\alpha_\star$.}
To explain why $\alpha_\star$ is the right threshold, we assume in the ideal case that $\ket{e}\mapsto\ket{G_\Hrm e}$ is an isometry and \Cref{thm:hermite-compact-decoder} can be applied with the condition in \Cref{itm:backward_1}.
Let $V$ be the transform that implements $\ket{e}\mapsto\ket{G_\Hrm e}$ followed by $(\Fsf_q^{-1})^{\otimes K}$.
Then by \eqref{eq:dqi-target-fourier}, our $\ket{\Dcal_P}$ equals $V\ket{\Psi}$ whose amplitude on $\ket{u}$ is $\braket{u|V|\Psi}$.
To compute this, observe that for each $e\in\Sigma^n$, we have
\begin{equation}\label{eq:Ve-ideal}
 V\ket{e}=(\Fsf_q^{-1})^{\otimes K}\ket{G_\Hrm e}\propto \sum_{u\in\Fbb_q^K}
     \omega^{\langle u,G_\Hrm e\rangle}\ket{u}
 =\sum_{u\in\Fbb_q^K}
     \omega^{\langle G_\Hrm^\top u,e\rangle}\ket{u},
\end{equation}
where $\omega=e^{2\pi\mathrm i/q}$.
Let $\Fsf_\Hrm=\Fsf_\Sigma^{\otimes n}$.
For each $u\in\Fbb_q^K$, we have
\begin{align}
\braket{u|V|\Psi}
&\propto\sum_{e\in\Sigma^n}\braket{u|V|e}\cdot\braket{e|\Psi}
\propto\sum_{e\in\Sigma^n}\omega^{\langle G_\Hrm^\top u,e\rangle}\cdot\braket{e|\Psi}
\propto\braket{G_\Hrm^\top u|\Fsf_\Hrm^{-1}|\Psi}
\label{eq:Vdual}\\
&\propto\bra{G_\Hrm^\top u}\bigotimes_{i=1}^n\left(\sqrt{1-\theta}\ket{\phi_0}+\sqrt\theta\ket{h_i}\right)
\tag{by \eqref{eq:dqi-product-state}}\\
&\propto\prod_{i=1}^n\left(\sqrt{1-\theta}+\sqrt\theta\cdot g_i\left((G_\Hrm^\top u)_i\right)\right)
\propto c_+^{\score(u)}\cdot c_-^{n-\score(u)},
\tag{by \eqref{eq:dqi-centered-membership}}
\end{align}
where
$$
c_+=\sqrt{1-\theta}+\sqrt{\theta(1-\rho)/\rho}
\quad\text{and}\quad 
c_-=\sqrt{1-\theta}-\sqrt{\theta\rho/(1-\rho)}.
$$

To see where the score should concentrate, consider an ideal model of independent row assignments with $\rho=1/2$, the limit as $q\to\infty$, in which $c_\pm=\sqrt{1-\theta}\pm\sqrt\theta$.
In this model, the squared amplitudes give a score distribution
\begin{equation}\label{eq:ideal-score}
 \Pr[\text{score}=s]
 \propto\binom ns\cdot|c_+|^{2s}\cdot|c_-|^{2(n-s)}.
\end{equation}
Since $|c_+|^2+|c_-|^2=2$, this is a binomial distribution with positive rate $c_+^2/2 =1/2+\sqrt{\theta(1-\theta)}$,
which maximizes at $\alpha_\star=(2+\sqrt3)/4$ when $0\le\theta<1/4$ approaches $1/4$.
This suggests that the output score should concentrate near $\alpha_\star\cdot n$ as desired; the proof below accounts for the restriction to codewords and the projection onto decodable errors.

\paragraph{The realistic case.}
To meet \Cref{itm:backward_1}, we project $\ket{\Psi}$ onto the
unique-decoding range.
Fix $0<\alpha<\alpha_\star$.
Choose a constant $\lambda\in(0,1/4)$, depending only on $\alpha$, such that
$$
 \beta:=\frac12+\sqrt{\lambda(1-\lambda)}>\alpha
$$
and define
$$
\theta:=4\lambda\rho(1-\rho)=\lambda(1-|\Sigma|^{-2}).
$$
We will prepare
$\ket{\Psi}$ with this value of $\theta$. Since $\theta\le\lambda<1/4$,
its typical error weight lies a constant fraction below the decoding
threshold, while the heuristic score fraction tends to $\beta>\alpha$.

Set
\begin{equation}
 d=\lfloor n/4\rfloor-1=\lfloor(q-1)/4\rfloor-1
 \quad\text{and}\quad
 \Pi_d=\sum_{\wt(e)\le d}\ket{e}\bra{e}.
 \label{hbq:filter}
\end{equation}
After projection, the unnormalized state is $\ket{\Phi}:=\Pi_d\ket{\Psi}$.

\begin{claim}\label{clm:PhiPsi}
$\|\ket{\Psi}-\ket{\Phi}\|=\sqrt{1-\|\Phi\|^2}=2^{-\Omega_\alpha(n)}$.
\end{claim}
\begin{proof}
Since $\|\Phi\|^2=\Pr[\operatorname{Bin}(n,\theta)\le d]$ and $d-\theta\cdot n=\Omega_\alpha(n)$, a Chernoff bound finishes the proof.
\end{proof}

The Fourier formula \eqref{eq:Ve-ideal} holds for every $e\in\Sigma^n$ with $\wt(e)\le d$. Recall that $V$ is the transform that implements $\ket{e}\mapsto\ket{G_\Hrm e}$, provided $\wt(e)\le d$ and using \Cref{thm:hermite-compact-decoder}, followed by $(\Fsf_q^{-1})^{\otimes K}$. Then the following workflow works out: first prepare $\ket{\Psi}$, then compute the weight and postselect on $\wt(e)\le d$ to obtain $\ket{\Phi}/\|\Phi\|$, finally run $V$ and measure the outcome $V\ket{\Phi}/\|\Phi\|$ to obtain a target answer $u$ for \Cref{prob:stream-hopi_code}.

The following claim proves the concentration guarantee after projection.

\begin{claim}\label{clm:truncation-score}
For each $u\in\Fbb_q^K$, define $p_u=|\braket{u|V|\Phi}|^2=|\braket{u|V\Pi_d|\Psi}|^2$.
Then
$$
\sum_{\score(u)<\alpha\cdot n}p_u\le O_\alpha(1/n).
$$
\end{claim}
\begin{proof}
For any $\mu>\alpha n$, the second-moment method gives
\begin{equation}\label{eq:truncation-second-moment-bound}
 \sum_{\score(u)<\alpha n}p_u
 \le\frac{\sum_u(\score(u)-\mu)^2p_u}{(\mu-\alpha n)^2}.
\end{equation}
We will choose $\mu=\beta n+O_\alpha(1)$ and bound the numerator by
$O_\alpha(n)$. Since $\beta>\alpha$, the denominator is then
$\Omega_\alpha(n^2)$, which proves the claim.
For the analysis, extend $V$ linearly to all errors by
$V\ket{e}=(\Fsf_q^{-1})^{\otimes K}\ket{G_\Hrm e}$ for
$e\in\Sigma^n$.
This extension is an isometry among errors with $\wt(e)\le d+1$ by
$2\cdot(d+1)<d(C_\Hrm^\perp)=n/2+1$ and \Cref{thm:hermite-compact-decoder}.
In addition, the calculation in \eqref{eq:Vdual} shows that 
\begin{equation}\label{eq:clm:truncation-score_1}
\bra{u}V=q^{(\ell n-K)/2}\bra{G_\Hrm^\top u}\Fsf_\Hrm^{-1}.
\end{equation}

To express the score in the error basis, write
$S(z)=\sum_i\1[z_i\in S_i]$ for $z\in\Sigma^n$ and define
\begin{equation}\label{eq:clm:truncation-score_2}
 M_0\ket{z}=S(z)\ket{z}
 \quad\text{and}\quad
 M=\Fsf_\Hrm M_0\Fsf_\Hrm^{-1}.
\end{equation}
Thus $M$ computes the score after changing to the error basis.
Since $S(G_\Hrm^\top u)=\score(u)$, we have
\begin{align*}
 \bra{u}VM
 &=q^{(\ell n-K)/2}\bra{G_\Hrm^\top u}M_0\Fsf_\Hrm^{-1}
 \tag{by \eqref{eq:clm:truncation-score_1}}\\
 &=\score(u)\cdot q^{(\ell n-K)/2}\bra{G_\Hrm^\top u}\Fsf_\Hrm^{-1}
 \tag{by \eqref{eq:clm:truncation-score_2}}\\
 &=\score(u)\bra{u}V.
 \tag{by \eqref{eq:clm:truncation-score_1}}
\end{align*}
Consequently, we express the numerator as
\begin{align}
\sum_u(\score(u)-\mu)^2p_u
&=\left\|\sum_u(\score(u)-\mu)\braket{u|V|\Phi}\cdot\ket{u}\right\|^2=\left\|\sum_u(\bra{u}VM-\mu\bra{u}V)\ket{\Phi}\cdot\ket{u}\right\|^2\notag\\
&=\left\|\sum_u\bra{u}V(M-\mu I)\ket{\Phi}\cdot\ket{u}\right\|^2=\|V(M-\mu I)\ket{\Phi}\|^2\notag\\
&=\|(M-\mu I)\ket{\Phi}\|^2.
\label{eq:truncation-moment-identity}
\end{align}
For the last equality, we observe that $M$ is a sum of operators on individual blocks and $(M-\mu I)\ket{\Phi}$ is supported on errors with $\wt(e)\le d+1$, on which $V$ is an isometry.

We now choose $\mu$ using $\ket{\Psi}$. Recall from \eqref{eq:dqi-product-state} that $\Fsf_\Hrm^{-1}\ket{\Psi}$ is a product state whose coordinate lies in the prescribed set with probability
\[
 \tau
 =\left(\sqrt{\rho(1-\theta)}+
              \sqrt{(1-\rho)\theta}\right)^2
 =\beta+O_\alpha(1/n),
\]
where $|\rho-1/2|=O_\alpha(1/n^\ell)=O_\alpha(1/n)$, $|\theta-\lambda|=O_\alpha(1/n^2)$, and $\beta=1/2+\sqrt{\lambda(1-\lambda)}$.

Set $\mu:=\tau n=\beta n+O_\alpha(1)$. Then we have $\|(M-\mu I)\ket{\Psi}\|^2=n\tau(1-\tau)\le n/4$.
Since $\|M-\mu I\|\le n$, \Cref{clm:PhiPsi} implies
\begin{equation}\label{eq:hermite-filter-deviation}
 \|(M-\mu I)\ket{\Phi}\|
 \le\|(M-\mu I)\ket{\Psi}\|+n\cdot\left\|\ket{\Psi}-\ket{\Phi}\right\|
 =O_\alpha(\sqrt n).
\end{equation}
By \eqref{eq:truncation-moment-identity}, the numerator in \eqref{eq:truncation-second-moment-bound} is
$O_\alpha(n)$. Since $\mu-\alpha n=\Omega_\alpha(n)$,
this completes the whole proof.
\end{proof}

\subsection{Implementation in one streaming pass}
\label{subsec:quantum-loading}

To complete the proof of \Cref{thm:quantum-poly}, we need to prepare the product state $\ket{\Psi}$ in the streaming setting. This is adapted from \cite{rosmanis2026}. We assume for simplicity arbitrary single- and two-qubit gates are allowed.

\begin{lemma}\label{lem:quantum-loading}
Given one pass over the stream in \Cref{prob:stream-hopi_code}, one can exactly prepare $\ket{\Psi}$ from \eqref{eq:dqi-product-state} using $O(q\ell\log q)$ space and $\poly(q,\ell)$ preprocessing-update-postprocessing time.
\end{lemma}
\begin{proof}
Recall that $\theta=4\lambda\rho(1-\rho)$, $\rho=(q^\ell-1)/(2\cdot q^\ell)$, and $\lambda\in(0,1/4)$ is fixed.
For each $i\in[n]$, initialize
\[
 \left(\sqrt{1-\lambda}\ket{0}+\sqrt\lambda\ket{1}\right)\ket{\phi_0}
 \quad\text{where}\quad
 \ket{\phi_0}=\frac1{|\Sigma|^{1/2}}\sum_{z\in\Sigma}\ket{z}.
\]
On an entry $(i,z,0)$, negate the amplitude of $\ket{1}\ket{z}$
in row $i$; entries $(i,z,1)$ do nothing.
By \eqref{eq:dqi-centered-membership}, the $i$th state is
$$
\sqrt{1-\lambda}\ket{0}\ket{\phi_0}+\sqrt{\lambda}\ket{1}\left((2\rho-1)\ket{\phi_0}+2\sqrt{\rho(1-\rho)}\ket{h_i}\right).
$$
After the pass, apply $\Fsf_\Sigma$ and obtain
\begin{equation}
 \left(\sqrt{1-\lambda}\ket{0}+\sqrt\lambda\cdot(2\rho-1)\ket{1}\right)
 \ket{0}+\sqrt\theta\ket{1}\ket{\widehat h_i}.
 \label{eq:stream-loaded}
\end{equation}
Choose an arbitrary single-qubit unitary $U$, depending only on $\lambda$ and $\rho$, such that
\[
 U\left(\sqrt{1-\lambda}\ket{0}+\sqrt\lambda\cdot(2\rho-1)\ket{1}\right)
 =\sqrt{1-\theta}\ket{1}.
\]
Apply $U$ to the control qubit conditioned on the rest being $0$.
Since $\ket{\widehat h_i}$ is orthogonal to $\ket{0}$, this
acts only on the first term of \eqref{eq:stream-loaded}, yielding $\ket{1}\otimes
 \left(\sqrt{1-\theta}\ket{0}+\sqrt\theta\ket{\widehat h_i}\right)$.
Discarding the control qubits therefore leaves exactly
$\ket{\Psi}$ from \eqref{eq:dqi-product-state}.
The time and space bounds are immediate.
\end{proof}

Finally we complete the proof of \Cref{thm:quantum-poly}.

\begin{proof}[Proof of \Cref{thm:quantum-poly}]
The algorithm in \Cref{subsec:quantum-algorithm} prepares
$\ket{\Psi}$ from \eqref{eq:dqi-product-state}, measures the
projection \eqref{hbq:filter}, and, on success, applies $V$ from
\eqref{eq:Ve-ideal} and measures the output polynomial.
On rejection, it outputs the zero polynomial.
In the streaming setting, we use \Cref{lem:quantum-loading} to supply $\ket{\Psi}$. All other operations happen during the postprocessing phase, which matches the claimed space and time bounds by \Cref{thm:hermite-compact-decoder}
and its coherent implementation~\cite{khattar2025circuits}. 
\Cref{clm:PhiPsi} and \Cref{clm:truncation-score} bound the failure probability by $O_\alpha(1/n)+2^{-\Omega_\alpha(n)}=O_\alpha(1/q)\le0.01$ for $q\ge\Omega_\alpha(1)$ sufficiently large, as desired.
\end{proof}

\section{Doubly efficient quantum streaming at every approximation}
\label{sec:quantum-near-optimal}

We prove \Cref{thm:quantum-exp-summary} by extending the DQI construction
of \Cref{sec:quantum-streaming} beyond the unique-decoding radius.

\QuantumNearOptimalUpper*

During the streaming phase, we prepare the same product state as in
\Cref{sec:quantum-streaming}, with a larger typical error weight.
Then in postprocessing, we append a few error blocks to the syndrome
so that it identifies almost every retained error uniquely. Fourier
interference again samples an answer with high score. We implement
the decoding by exhaustive search, using small space and exponential time.

\paragraph{Section organization.}
In \Cref{subsec:nearopt-algorithm}, we explain the algorithm assuming
access to the whole input sets and prove the approximation guarantee.
In \Cref{subsec:nearopt-implementation}, we implement the algorithm
in one streaming pass.

\subsection{The algorithm with access to the whole input}
\label{subsec:nearopt-algorithm}

We follow the notation of \Cref{prob:stream-hopi_code} as in \Cref{subsec:quantum-algorithm}. In particular,
$n=q-1$, $N=nq^\ell$, $K=\ell n/2$, and
$\rho=(1-|\Sigma|^{-1})/2$. Fix $\alpha\in(0,1)$ and put
$\epsilon=1-\alpha$. Our goal is to sample $u$ satisfying
\[
 \score(u)\ge(1-\epsilon)\cdot n\ge\alpha\cdot\OPT.
\]

\paragraph{Working backward from a near-optimal score.}
The heuristic in \Cref{subsec:quantum-algorithm} suggests using the
product state with $\theta$ close to $1/2$: the ideal score fraction
$1/2+\sqrt{\theta(1-\theta)}$ then approaches one.
Choose a constant $\lambda\in(0,1/2)$ such that
\[
 \beta:=\frac12+\sqrt{\lambda(1-\lambda)}>1-\frac\epsilon8.
\]
As before, we use
\begin{equation}
 \ket{\Psi}=\bigotimes_{i=1}^n
 \left(\sqrt{1-\theta}\ket{0}+
       \sqrt\theta\ket{\widehat h_i}\right)
\quad\text{where}\quad
\theta=4\lambda\rho(1-\rho)=\lambda(1-|\Sigma|^{-2}).
 \label{eq:nearopt-product-state}
\end{equation}
Its block weight has distribution $\operatorname{Bin}(n,\theta)$.
Pick a constant $\tau\in(\lambda,1/2)$ depending only on $\alpha$,
and define $d=\lfloor\tau n\rfloor$.
Since $\theta\le\lambda<\tau$, almost all the mass of $\ket{\Psi}$ is at weight at most $d$.
The one-pass construction in \Cref{lem:quantum-loading} directly applies to this choice of $\lambda$ as well.

\paragraph{Augmenting the syndrome.}
The difficulty is decoding. Errors of weight close to $n/2$ need not
be determined by their syndrome, so we cannot directly implement
$\ket{e}\mapsto\ket{G_\Hrm e}$. Write
\[
 B_t=\{e\in\Sigma^n:\wt(e)\le t\}
 \quad\text{and}\quad
 \Ecal_t(s)=\{e\in B_t:G_\Hrm e=s\}.
\]
The key observation here is the following \Cref{clm:dualwt}, which implies that a few random coordinates should distinguish different errors, provided the syndrome list is small enough.

\begin{claim}\label{clm:dualwt}
If $e\ne e'\in\Sigma^n$ satisfy $G_\Hrm e=G_\Hrm e'$, then $\wt(e-e')\ge n/2+1$.
\end{claim}
\begin{proof}
Apply \Cref{thm:hermite-compact-decoder} with $K=\ell n/2$.
\end{proof}

Choose a uniform $r$-subset $R\subseteq[n]$, independently of
$\ket{\Psi}$, with $r$ specified below. Let
$E_R\colon\Sigma^n\to\Fbb_q^{\ell r}$ extract the coordinates in $R$ and flatten $\Sigma$ into $\Fbb_q^\ell$.
Define
\[
 H_R=\begin{pmatrix}G_\Hrm\\E_R\end{pmatrix},
 \qquad H_Re=(G_\Hrm e,e_R)\in\Fbb_q^{K+\ell r}.
\]
After inverse Fourier transformation, a label
$(u,v)\in\Fbb_q^K\times\Fbb_q^{\ell r}$ corresponds to
\[
 H_R^\top(u,v)=G_\Hrm^\top u+E_R^\top v.
\]
The extra term $E_R^\top v$ changes only the coordinates in $R$. Thus, with
$S(z)=\sum_i\1[z_i\in S_i]$ as in
\Cref{subsec:quantum-algorithm}, discarding $v$ costs at most $r$
satisfied constraints:
\begin{equation}\label{eq:nearopt-score-loss}
 \score(u)\ge S\bigl(H_R^\top(u,v)\bigr)-r.
\end{equation}

\paragraph{Projecting onto isolated errors.}
Call an error $e\in B_d$ \emph{isolated} for $R$ if it is the unique error in $B_{d+2}$ with augmented syndrome $H_Re$. Define
\[
 \Gcal_R=\{e\in B_d:\text{no $e'\in B_{d+2}\setminus\{e\}$ with }H_Re'=H_Re\}
 \quad\text{and}\quad \Pi_R=\sum_{e\in\Gcal_R}\ket{e}\bra{e}.
\]
The $+2$ leaves room for the second-moment calculation, just as the condition
$2d+2<d(C_\Hrm^\perp)$ in \Cref{clm:truncation-score} does.
We measure the projection flag. On success, the unnormalized state is
$\ket{\Phi_R}=\Pi_R\ket{\Psi}$.

To choose $r$, we first bound the maximum syndrome-list size.

\begin{lemma}\label{lem:nearopt-dual-lists}
For every $\ell\ge1$ and $q\ge\Omega_\tau(1)$, we have
\begin{equation}\label{eq:nearopt-list-bound}
 \max_{s\in\Fbb_q^K}|\Ecal_{d+2}(s)|\le\exp\left(O_\tau\left(
       \ell\log(\ell+1)+\frac{n}{(\log n)^3}\right)\right).
\end{equation}
\end{lemma}
\begin{proof}
We first reduce the syndrome list to a list-decoding ball in $C_\Hrm$. By \Cref{fct:hermite-duality}, there is an invertible block-diagonal map $T$ such that $C_\Hrm^\perp=TC_\Hrm$. Fix a syndrome $s$. If $\Ecal_{d+2}(s)$ is empty, there is nothing to prove. Otherwise, choose $e_0$ with $G_\Hrm e_0=s$ and put $y=-T^{-1}e_0$. Since $\ker G_\Hrm=C_\Hrm^\perp$, every vector with syndrome $s$ has a unique representation
\[
 e=e_0+Tc=T(c-y),\qquad c\in C_\Hrm.
\]
Each block of $T$ is invertible, so $\wt(e)=d(c,y)$. Consequently,
\[
 |\Ecal_{d+2}(s)|
 =|\{c\in C_\Hrm:d(c,y)\le d+2\}|.
\]
Put $\gamma=(1/2-\tau)/2>0$ and
$\ell_0=\lceil16/\gamma^2\rceil$.
\begin{itemize}
 \item Large multiplicity: $\ell\ge\ell_0$.
The multiplicity condition in \Cref{thm:krsw-list-recovery} holds
by the definition of $\ell_0$.
For sufficiently large $q$, we also have $q\ge4/\gamma$ and
$2/n\le\gamma$, so
$\frac{d+2}{n}
 \le\tau+\frac2n
 \le\tau+\gamma
 =\frac12-\gamma$.
The theorem therefore gives
\[
 |\Ecal_{d+2}(s)|
 \le(\ell/\gamma)^{O(\ell/\gamma^2)}
 =\exp\bigl(O_\tau(\ell\log(\ell+1))\bigr).
\]

\item Small multiplicity: $1\le\ell<\ell_0$.
Here we use \Cref{hendpoint:lr} with $\alpha'=1-\tau/2>3/4$
and singleton lists $A_i=\{y_i\}$. Their total incidence is $n$.
Indeed, writing $\kappa(\alpha')=\frac{1-2\tau}{2-2\tau}>0$, the theorem's incidence budget can be taken in the form
$T=\frac{q^{1+\ell\kappa(\alpha')}}{(\ell\log q)^{B_{\alpha'}\ell}}$
for a constant $B_{\alpha'}>0$. Uniformly over $1\le\ell<\ell_0$,
\[
 \frac{T}{n}\ge
 \left(\frac{q^{\kappa(\alpha')}}
              {(\ell_0\log q)^{B_{\alpha'}}}\right)^\ell
 \ge1
\]
for sufficiently large $q$. Moreover, $10n/\ln q\ge2$ for
sufficiently large $q$. Since $n-d-2$ is an integer, this implies
\[
 n-d-2
 \ge\left\lceil\left(1-\tau-\frac{10}{\ln q}\right)n\right\rceil,
\]
which is exactly the required agreement threshold. Therefore,
\[
 \log|\Ecal_{d+2}(s)|
 =O_{\ell,\tau}\left(\frac{q}{(\log n)^3}\right)
 =O_{\ell,\tau}\left(\frac{n}{(\log n)^3}\right).
\]
There are only finitely many integers $1\le\ell<\ell_0$, and
$\ell_0$ depends only on $\tau$. Taking the maximum of their
constants and field-size thresholds makes the last bound uniform
in $\ell$, with dependence only on $\tau$.
\end{itemize}
Both bounds are uniform in the received word $y$ and hence in $s$. Combining the two cases proves \eqref{eq:nearopt-list-bound}.
\end{proof}

The dual distance makes each competing error unlikely to survive the extra coordinates.

\begin{lemma}\label{lem:nearopt-isolation}
For $q\ge\Omega_\tau(\ell\log(\ell+1))$, there is an integer
$r\le\epsilon n/4$ such that, for a uniform $r$-subset $R\subseteq[n]$
and every fixed $e\in B_d$, we have $\Pr_R[e\notin\Gcal_R]\le1/n$.
Consequently, every fixed unit vector $\ket{\psi}$ satisfies
\begin{equation}\label{eq:nearopt-isolation-mass}
 \E_R\left[\|(I-\Pi_R)\ket{\psi}\|^2\right]
 \le\sum_{e\notin B_d}|\braket{e|\psi}|^2+1/n.
\end{equation}
\end{lemma}
\begin{proof}
By \Cref{lem:nearopt-dual-lists}, choosing
\[
 r=\left\lceil C_\tau\left(\ell\log(\ell+1)
       +\frac{n}{(\log n)^3}\right)+\log n\right\rceil
\]
with a sufficiently large constant $C_\tau$ ensures
$2^{-r}\max_s|\Ecal_{d+2}(s)|\le1/n$.
Our choice of $\lambda<\tau$ gives
$\epsilon>8\cdot(1/2-\sqrt{\tau(1-\tau)})>0$.
Thus $\epsilon$ is bounded below by a positive constant depending only
on $\tau$. The first term in $r$ is $O_\tau(\ell\log(\ell+1))$
and the remaining terms are $o_\tau(n)$, so
$q\ge\Omega_\tau(\ell\log(\ell+1))$ with a sufficiently large
constant ensures $r\le\epsilon n/4$.

Fix $e\in B_d$. For any competitor
$f\in\Ecal_{d+2}(G_\Hrm e)\setminus\{e\}$,
\Cref{clm:dualwt} gives $\wt(f-e)\ge n/2+1$.
The augmented syndromes collide exactly when $R$ misses all
coordinates on which $e$ and $f$ differ, so
\[
 \Pr_R[H_Rf=H_Re]
 =\frac{\binom{n-\wt(f-e)}r}{\binom nr}\le2^{-r}.
\]
A union bound over the competitors gives
\[
 \Pr_R[e\notin\Gcal_R]
 \le\bigl(|\Ecal_{d+2}(G_\Hrm e)|-1\bigr)2^{-r}\le1/n.
\]
Finally, since $\Pi_R$ is diagonal in the error basis, every fixed
unit vector $\ket{\psi}$ satisfies
\[
 \E_R\|(I-\Pi_R)\ket{\psi}\|^2
 =\sum_{e\notin B_d}|\braket{e|\psi}|^2
  +\sum_{e\in B_d}|\braket{e|\psi}|^2
       \Pr_R[e\notin\Gcal_R].
\]
The preceding bound and $\|\psi\|=1$ give
\eqref{eq:nearopt-isolation-mass}.
\end{proof}

Since the product-state weight has distribution
$\operatorname{Bin}(n,\theta)$ and $\theta\le\lambda<\tau$, a Chernoff
bound gives
\begin{equation}\label{eq:nearopt-weight-tail}
 \sum_{e\notin B_d}|\braket{e|\Psi}|^2
 \le e^{-2(\tau-\lambda)^2n}.
\end{equation}
Thus, by \Cref{lem:nearopt-isolation}, the average rejected mass satisfies
\begin{equation}\label{eq:nearopt-rejection}
 \delta_n:=\E_R\|(I-\Pi_R)\ket{\Psi}\|^2
 \le e^{-2(\tau-\lambda)^2n}+1/n.
\end{equation}

\paragraph{Decoding and sampling.}
On the successful branch, every retained $e$ belongs to
$\Gcal_R$, so $H_Re$ determines $e$ uniquely within $B_{d+2}$.
We can therefore perform the transform $\ket{e}\to\ket{H_Re}$ by exhaustive search.
Applying $(\Fsf_q^{-1})^{\otimes(K+\ell r)}$ to the augmented syndrome then gives the isometry
\begin{equation}\label{eq:nearopt-isometry}
 V_R\ket{e}=q^{-(K+\ell r)/2}
   \sum_{y\in\Fbb_q^{K+\ell r}}
     \omega^{\langle y,H_Re\rangle}\ket{y}
 \quad\text{for every}\quad e\in\Gcal_R.
\end{equation}
Finally we measure $V_R\ket{\Phi_R}/\|\Phi_R\|$ to obtain $(u,v)$,
and only output $u$ for \Cref{prob:stream-hopi_code}.

\paragraph{Approximation guarantee.}
To derive the approximation guarantee, we follow the second-moment argument in
\Cref{clm:truncation-score}, averaging over the independently chosen $R$.

\begin{claim}\label{clm:nearopt-score}
The probability that the projection succeeds but the output satisfies
$\score(u)<\alpha n$ is $O_\alpha(1/n)$.
\end{claim}
\begin{proof}
Follow the proof of \Cref{clm:truncation-score}, using its operator
$M$ and product-state mean $\mu$ at the present $\lambda$.
We only highlight the changes.
First, replace $G_\Hrm$ by $H_R$, extend $V_R$ to all errors by
\eqref{eq:nearopt-isometry}, and set
$\widetilde M_R\ket{y}=S(H_R^\top y)\ket{y}$.
The same Fourier calculation gives $\widetilde M_RV_R=V_RM$.
Isolation in $B_{d+2}$ ensures that the linear extension of $V_R$
is an isometry on errors within one block of $\Gcal_R$:
if $H_R(e+a)=H_R(f+b)$ with $e,f\in\Gcal_R$ and
$\wt(a),\wt(b)\le1$, isolation of $f$ forces $e+a-b=f$.
Hence the moment identity still holds:
\begin{equation}\label{eq:nearopt-moment-identity}
 \braket{V_R\xi|(\widetilde M_R-\mu I)^2|V_R\xi}
 =\|(M-\mu I)\ket{\xi}\|^2
\end{equation}
for every $\mu$ and every $\ket{\xi}$ supported on $\Gcal_R$. Second, average the perturbation estimate over $R$. Using
\eqref{eq:nearopt-rejection} and the same product-state variance bound,
\[
 \E_R\|(M-\mu I)\ket{\Phi_R}\|^2
 \le n/2+2n^2\delta_n=O_\alpha(n).
\]
Third, discarding the extra coordinates costs at most $r$ in score
by \eqref{eq:nearopt-score-loss}. Since
$\mu=\beta n+O_\alpha(1)\ge(1-\epsilon/4)n$ and
$r\le\epsilon n/4$, a bad output has augmented score less than
$\mu-\epsilon n/2$. The same second-moment bound, applied to the
unnormalized states $\ket{\Phi_R}$ and averaged over $R$, gives
joint probability $O_\alpha(n)/(\epsilon n/2)^2=O_\alpha(1/n)$.
\end{proof}

Adding the rejection probability $\delta_n$, the algorithm fails
with probability at most $O_\alpha(1/n)=O_\alpha(1/q)$.
This calculation uses joint probabilities throughout; it does not
require a lower bound on $\|\Phi_R\|$ for each individual $R$.

\subsection{Implementation in one streaming pass}
\label{subsec:nearopt-implementation}

To implement the algorithm, the streaming phase is the same as in \Cref{subsec:quantum-loading} for $\ket{\Psi}$ and only the postprocessing changes: we use exhaustive search to test isolation and decode the augmented syndrome.
The proof of \Cref{lem:quantum-loading} applies unchanged to the
present $\lambda\in(0,1/2)$, and we use the same arbitrary single-
and two-qubit gate model.

\begin{lemma}\label{lem:nearopt-decoder}
For every $R\subseteq[n]$, the projection $\Pi_R$ and, on success,
the map $\ket{e}\to\ket{H_Re}$ can be implemented with
$O(q\ell\log q)$ space and $q^{O(q\ell)}$ time.
\end{lemma}
\begin{proof}
To test membership in $\Gcal_R$, enumerate $f\in B_{d+2}$,
counting those with $H_Rf=H_Re$.
Accept exactly when $\wt(e)\le d$ and this count is one. Compute this flag reversibly,
uncompute the work registers, and measure the flag to implement $\Pi_R$.

On acceptance, perform
$\ket{e}\ket{0}\to\ket{e}\ket{H_Re}\to\ket{0}\ket{H_Re}$.
The first step computes the augmented syndrome. For the second,
enumerate $f\in B_{d+2}$ with the given augmented syndrome; the
unique match is $e$, so this search erases the error register.
Both searches examine at most $|\Sigma|^n=q^{O(q\ell)}$ candidates.
The candidate, counters, and syndrome checks use $O(q\ell\log q)$
space, reusing the work registers after each check. The same bounds
hold for the reversible computations and their inverses.
\end{proof}

Now we are ready to prove \Cref{thm:quantum-exp-summary}.

\begin{proof}[Proof of \Cref{thm:quantum-exp-summary}]
In the streaming phase, we use \Cref{lem:quantum-loading} to prepare $\ket{\Psi}$ with the claimed space and time bounds.
In the postprocessing phase, choose $r$ as in \Cref{lem:nearopt-isolation}
and sample a uniform $r$-subset $R$. Apply \Cref{lem:nearopt-decoder}
to measure $\Pi_R$ and, on acceptance, decode the augmented syndrome.
On acceptance, inverse Fourier-transform the
augmented syndrome, measure $(u,v)$, and output $u$ for an answer to \Cref{prob:stream-hopi_code}.
On rejection, output the zero polynomial.
Sampling $R$ and the coordinatewise Fourier transforms use
$O(q\ell\log q)$ space and $\poly(q,\ell)$ time. Capping the
integer-sampling rejection trials at $O(\log n)$ per draw adds at
most $1/n$ failure probability, while keeping $R$ exactly
uniform conditional on success.
Together with \eqref{eq:nearopt-rejection} and
\Cref{clm:nearopt-score}, this gives total failure probability
$O_\alpha(1/q)\le0.01$ for
$q\ge\Omega_\alpha(\ell\log(\ell+1))$ sufficiently large.
\end{proof}

\section{Classical streaming lower bounds}
\label{sec:classical}
In this section, we prove the classical streaming lower bound for
$(\alpha,q,\ell)$-HOPI.

\ClassicalStreamingLower*

Fix $\ell$ and $\alpha$ as above. We prove the streaming lower bound by using the $k$-party communication complexity. Index the rows in their fixed group order, write $F_i=S_{a_i}$, and partition them into $k$ consecutive blocks of size at most $\lceil n/k\rceil$, one for each party. We use public-blackboard communication and write $|\Pi|$ for the worst-case communication cost. The communication bound holds even with success probability $2/3$, and therefore applies to the $0.99$ success requirement in \Cref{prob:streaming-hopi}.

Throughout this section, put $h=\ln n$,$\zeta=2\alpha-1-\frac{10}{\ln q}$ and $k=\lceil h\rceil$.  For $F=(F_1,\ldots,F_n)$ and $c\in\Sigma^n$, write $\err_F(c)=|\{i\in[n]:c_i\notin F_i\}|$. 

Recall that the Hermite code is $C_\Hrm=
\left\{\bigl(J_{\ell}P(a_1),\ldots,J_{\ell}P(a_n)\bigr):
 P\in\Pcal_K \right\}\subseteq\Sigma^n$
where $\Sigma=\Fbb_q^\ell$ and $K=\ell n/2$. This code is fixed independently of the input lists.
Using the list-recovery bound in \Cref{hendpoint:lr}, it remains
to prove the following communication lower bound.
\begin{lemma}
\label{cl:communication}
Suppose $C_\Hrm$ is $(\zeta,T,L)$-list-recoverable, where
$T\ge2n$ and $\log L\le A n/h^3$ for some $A>0$.
For $q\ge\Omega_{\alpha,A}(1)$, there is an input distribution such that any communication protocol $\Pi$ giving an $\alpha$-approximation with probability at least $2/3$ over the distribution has
\begin{equation}
 |\Pi|\ge \frac{c\cdot T}{h^2}-\frac{\ell n\log q}c,
 \label{cl:communication-interface}
\end{equation}
where $c>0$ is an absolute constant.
\end{lemma}

We prove the communication lower bound by using the gadgetless lifting idea of \cite{yangzhang2024,maoyangzhang2024} and present the proof by following~\cite[Section~5]{ggjl2025}. We first (\Cref{subsec:classical-subcube}) prove the lower bound for highly structured communication protocols that we call subcube protocols. We then (\Cref{subsec:classical-dense}) generalize
that proof to subcube-like protocols. Finally
(\Cref{cl:preprocessing}), we use the structure-vs-randomness paradigm~\cite{gpw2020}
to reduce general protocols to this case. 

\subsection{Communication lower bound for subcube protocols}
\label{subsec:classical-subcube}
We start by defining the input distribution $\mu$ and its properties. Let $\mu$ be the distribution of independent uniform lists $F_i$ of size $(|\Sigma|-1)/2$. Under $\mu$, we have $\OPT=n$ with probability $1-o(1)$ (See \Cref{hb:random-sat}). Thus, on almost all inputs from $\mu$, an $\alpha$-approximation must satisfy at least $\alpha n$ rows.
In each row, choose a uniform excluded jet $\xi_i$ and a uniform perfect matching of $\Sigma\setminus\{\xi_i\}$, and select one endpoint of every matching edge using an independent fair bit. The resulting list has size
$(|\Sigma|-1)/2$. Its distribution is invariant under all permutations of $\Sigma$, which act transitively on lists of this size; hence it is uniform. Therefore, conditional on the excluded jets and matchings, the orientation bits are uniform and independent, and the rows are sampled
independently. 

Since the hard distribution is satisfiable with probability $1-o(1)$, the protocol outputs a polynomial satisfying at least $\alpha n$ rows with probability at least $2/3-o(1)$. By averaging, we fix the matchings, excluded jets, and random coins while retaining success probability at least $1/2$ over uniform orientation bits. We append a verification step at the end. The output party sends the polynomial's $K$ coefficients, and each party reports the membership bit of the queried jet in each of its rows. We reject invalid polynomials or those satisfying fewer than $\alpha n$ rows. The resulting deterministic protocol is \emph{fully verified}: every accepting transcript determines the output polynomial and all queried membership bits, hence the corresponding orientation bits for nonexcluded jets. This adds only $K\lceil\log q\rceil+n+O(1)=O(\ell nh)$ bits of communication, so it suffices to analyze fully verified protocols. 

\paragraph{Proof for subcube protocols.}
We first define the notion of subcube protocols.
Let $N_a$ be the set of input bits held by party $a$.
A rectangle $R=X_1\times\cdots\times X_k$, where
$X_a\subseteq\{0,1\}^{N_a}$, is a subcube rectangle if each factor $X_a$ fixes some coordinates $I_a\subseteq N_a$ and leaves all other coordinates unrestricted. Its codimension is
$\codim(R)=\sum_a|I_a|$.
A subcube protocol is such that every complete-message node of
the protocol tree corresponds to a subcube rectangle.
This model is closely related to decision trees, but is strictly more general. 

Fix a fully verified subcube protocol $\Pi$ and a matching structure.  For a node $v$, let $F_i(v)$ consist of $\xi_i$ and both endpoints
of every edge whose orientation is fixed by $R_v$. Thus $F_i(v)$ contains addresses with known membership status, whether positive or negative; it is not the input list $F_i$. Since matching edges are disjoint,
\begin{equation}
 \sum_{i\in[n]}|F_i(v)|=n+2\codim(R_v).
 \label{cl:touched}
\end{equation}
Stop the analysis of a path just before its first node with $\sum_{i\in[n]}|F_i(v)|>T$, and call this an abort. For the original random leaf
$R$, uniformity gives $\E\codim(R)=\ent(R)\le|\Pi|$, where $\ent$ is Shannon entropy. By monotonicity and Markov's inequality,
\begin{equation}
 \Pr[\text{abort}]
 \le\frac{2\E[\codim(R)]}{T-n}
 \le\frac{2|\Pi|}{T-n}.
 \label{cl:abort}
\end{equation}

Call a codeword \emph{dangerous} at $v$ if at least
$\lceil\zeta n\rceil$ of its addresses belong to $F(v)$.
For the successive nodes $v_1,\ldots,v_d$ in the stopped path, let
\[
 Q_i=\Lcal_\zeta(F(v_i))
 =\{c\in C_\Hrm:
       |\{r\in[n]:c_r\in F_r(v_i)\}|\ge\lceil\zeta n\rceil\}.
\]
We observe three properties:
\begin{enumerate}[label=(\roman*),nosep]
 \item $Q_i\subseteq Q_{i+1}$ for all $i$.
 \item Every successful output on a non-aborting path belongs to $Q_d$.
 \item $|Q_d|\le L$.
\end{enumerate}
Here (i) follows from monotonicity of the fixed coordinates. For (ii), full verification makes every output query's membership status known, so every output address belongs to $F(v_d)$.
For (iii), apply list recovery to the exposed
address sets, rather than to the input lists:
\begin{equation}
 \sum_{r\in[n]}|F_r(v_i)|\le T
 \quad\Longrightarrow\quad
 |Q_i|=|\Lcal_\zeta(F(v_i))|\le L.
 \label{cl:list-recovery-use}
\end{equation}
It holds at every node, even though $F(v_i)$ depends on the transcript, because list recovery applies to every family of address sets within the incidence budget. Moreover, (i) gives $\abs{\bigcup_{i=1}^d Q_i}=|Q_d|\le L$. Thus at most $L$ distinct codewords become dangerous along the entire stopped path, not merely at any single node. This is the use of Hermite list recovery in the subcube proof: it bounds the number of candidates in the final union bound. No list-recovery algorithm is run by the protocol.

A codeword dangerous at the root queries at least $\lceil\zeta n\rceil$
excluded jets and cannot satisfy $\alpha n$ rows, since $\zeta>1-\alpha$ for sufficiently large $n$. Order the other dangerous codewords by first appearance, breaking ties by a fixed order, and pad the list with $\bot$ to obtain
$c_1,\ldots,c_L$. Once an index is assigned, its codeword does not change.

\begin{claim}
\label{cl:first-appearance}
For every $j\in[L]$, $\Pr[c_j\ne\bot\text{ and }\err_F(c_j)\le(1-\alpha)n] \le e^{-32n/h^2}$.
\end{claim}
\begin{proof}
Fix $j$ and an entry node $v$ where $c_j$ first becomes dangerous. Before the sender's message, fewer than $\lceil\zeta n\rceil$ of its addresses were exposed. The other parties therefore still hold $t$ unexposed queried bits, where
\begin{equation}
 t\ge n-\lceil n/k\rceil-\lceil\zeta n\rceil
   \ge2(1-\alpha)n+\frac{8n}{h}.
 \label{cl:unfixed-queries}
\end{equation}
Conditional on $R_v$, these bits are independent and fair, possibly complemented. Since $t\le n$, Hoeffding's inequality gives
\[
 \Pr[\err_F(c_j)\le(1-\alpha)n\mid R_v]
 \le\Pr[\operatorname{Bin}(t,1/2)\le(1-\alpha)n]
 \le e^{-32n/h^2}.
\]
The candidate and its index are determined at $v$. For a fixed $j$,
its possible entry nodes form an antichain, so their input events
are disjoint. Averaging the conditional bound over these events proves
the claim. No condition is imposed on the future transcript.
\end{proof}

Every successful non-aborting output is among these at most $L$
candidates. Using the claim and properties (ii)--(iii), we obtain
\begin{align}
 \Pr[\Pi\text{ succeeds and does not abort}]
 &\le\sum_{j=1}^L
   \Pr[c_j\ne\bot\text{ and }\err_F(c_j)\le(1-\alpha)n]
 \notag\\
 &\le L e^{-32n/h^2}
 \notag=\exp\left(O_{\ell,\alpha}(n/h^3)-32n/h^2\right)=o(1).
 \label{cl:subcube}
\end{align}
The list-recovery bound controls the number of candidates, whereas
\Cref{cl:first-appearance} controls the success probability of
each candidate through its still-unexposed queries.
Together with~\eqref{cl:abort}, this gives
\[
 \Pr[\Pi\text{ succeeds}]
 \le\frac{2|\Pi|}{T-n}+o(1).
\]
Therefore, any such protocol with success probability at least
$1/2$ has $|\Pi|=\Omega(T)$.

\subsection{Communication lower bound for general protocols}
\label{subsec:classical-dense}
We will now prove the lower bound for a more general class of
subcube-like protocols. We use dense random variables as a notion of pseudorandomness. In this section, fix $\gamma=1-h^{-2}$.

\begin{definition}[Min-entropy and density]
The min-entropy of a random variable $X$ is
\[
 \ent_\infty(X)=-\log \max_x\Pr[X=x].
\]
A distribution on $\{0,1\}^U$ is $\gamma$-dense if
$\ent_\infty(X_J)\ge\gamma|J|$ for every $J\subseteq U$.
\end{definition}

\begin{definition}[Subcube-like protocols]
A rectangle $R=X_1\times\cdots\times X_k$ is $\gamma$-subcube-like
with respect to sets $I_1,\ldots,I_k$ if $(X_a)_{I_a}$ is fixed and
the uniform distribution on the remaining coordinates of $X_a$ is
$\gamma$-dense, for every party $a$.
Its codimension is $\sum_a|I_a|$.
A protocol is subcube-like if every complete-message node is subcube-like.
\end{definition}

\paragraph{Proof for subcube-like protocols.}
The argument is nearly identical to that for subcube protocols.
Define $Q_1,\ldots,Q_d$ and $c_1,\ldots,c_L$ as before.
The only change to \Cref{cl:first-appearance} is its conditional
probability bound. The $t$ unexposed queries in
\eqref{cl:unfixed-queries} lie in the nonsending parties.
Density within each party and independence across party factors imply
that every assignment to these queries has probability at most
$2^{-\gamma t}$. Hence
\begin{align}
 \Pr[\err_F(c_j)\le(1-\alpha)n\mid R_v]
 \le2^{-\gamma t}
   \sum_{s=0}^{\lfloor(1-\alpha)n\rfloor}\binom ts
 \notag =2^{t/h^2}
   \Pr[\operatorname{Bin}(t,1/2)\le(1-\alpha)n]
 \le e^{-30n/h^2}.
 \label{cl:dense-tail}
\end{align}
At a fully verified leaf, every queried orientation bit is constant and must belong to its owner's fixed set: a constant bit cannot lie in a $\gamma$-dense distribution. Thus property (ii) still holds. The same first-appearance and union-bound argument, followed by Markov's inequality, gives
\begin{equation}
 \Pr[\Pi\text{ succeeds}]
 \le\frac{2\E[\codim(R)]}{T-n}+L e^{-30n/h^2},
 \label{cl:success-exposure}
\end{equation}
where $R$ is the original random leaf. It remains to bound its expected codimension for a refinement of a general protocol.

\paragraph{From general protocols to subcube-like protocols}
We will now prove \Cref{cl:communication} for general protocols. The proof is by a reduction to the subcube-like case using the density-restoring partition of~\cite{gpw2020,yangzhang2024,maoyangzhang2024,ggjl2025}.
The following lemma is the multiparty version needed here.
\begin{lemma}[Density-restoring partition \cite{gpw2020}]
\label{cl:density-restoring}
Let $X\subseteq\{0,1\}^{U}$ be nonempty and equip it with the uniform distribution. For every $0<\gamma<1$, there is a partition $X=\bigsqcup_j X^j$ and, for every cell, a set $I_j\subseteq U$ and a string
$a_j\in\{0,1\}^{I_j}$ such that:
\begin{enumerate}[label=(\roman*),nosep]
 \item every $x\in X^j$ satisfies $x_{I_j}=a_j$;
 \item the uniform distribution on the residual set
   $X'_j:=\{x_{U\setminus I_j}:x\in X^j\}$ is $\gamma$-dense;
 \item if $j$ is sampled with probability
   $p_j=|X^j|/|X|$, then
   \begin{equation}
    \E_j[\operatorname{def}(X'_j)]
    \le \operatorname{def}(X)+(\gamma-1)\E_j[|I_j|] + \log \mathrm e.
    \label{cl:restoration-step}
   \end{equation}
\end{enumerate}
\end{lemma}

\begin{lemma}
\label{cl:preprocessing}
Every deterministic protocol $\Pi$ on uniform independent input bits
has a $\gamma$-subcube-like refinement $\Pi'$ with the same output
such that its random leaf $R$ satisfies
\begin{equation}
 \E[\codim(R)]=O(h^2|\Pi|).
 \label{cl:fixed-communication}
\end{equation}
The implicit constant is independent of the number of parties.
\end{lemma}
\begin{proof}
Put $\Pi$ in binary normal form, so that each original
communication round sends one bit. The refinement $\Pi'$ is defined by
\Cref{algo:preprocess}. It maintains the invariant that the
current rectangle is $\gamma$-subcube-like with respect to the
coordinates fixed so far.

\begin{algorithm2e}[ht]
  \caption{Protocol preprocessing}
  \label{algo:preprocess}
  \DontPrintSemicolon
  \KwIn{A deterministic $k$-party protocol $\Pi$ in binary normal form
    and an input $(x_1,\ldots,x_k)$}
  \KwOut{The output of $\Pi$ on $(x_1,\ldots,x_k)$}
  Initialize $v\leftarrow$ the root of $\Pi$ and
    $I_a\leftarrow\emptyset$ for every party $a$.\;
  \While{$v$ is not a leaf}
  {
    Let party $a$ send the next bit at $v$, and write
      $X_a=X_{a,0}\sqcup X_{a,1}$ for the partition of its current
      factor according to this bit.\;
    Let $b$ be the bit determined by $x_a$, update
      $X_a\leftarrow X_{a,b}$, and move to the child $v_b$.\;
    Set $U_a=N_a\setminus I_a$ and
      $Y_a=\{x_{U_a}:x\in X_a\}$.\;
    Apply the density-restoring partition
      $Y_a=\bigsqcup_j Y_a^j$, with fixed sets $J_{a,j}\subseteq U_a$.\;
    Lift each cell to
      $X_a^j=\{x\in X_a:x_{U_a}\in Y_a^j\}$.\;
    Let $C(\cdot)$ be a prefix-free encoding of the cell index $j$.\;
    Let $j$ be such that $x_a\in X_a^j$, update
      $X_a\leftarrow X_a^j$ and $I_a\leftarrow I_a\cup J_{a,j}$,
      and send $(b,C(j))$.\;
  }
  Output the label of $v$.\;
\end{algorithm2e}

The cell index can, for example, be encoded by a Huffman code for the cell probabilities. This makes the refinement an ordinary binary protocol.
For a nonempty set $X\subseteq\{0,1\}^{U}$, let
$\operatorname{def}(X):=|U|-\log |X|$ denote its min-entropy deficiency. If $I\subseteq U$ and
$X'\subseteq\{0,1\}^{U\setminus I}$, the same notation is used with
$U\setminus I$ as the ambient coordinate set.
For the current factor of party $a$, define its residual deficiency by
\[
 D_a:=\operatorname{def}\bigl((X_a)_{N_a\setminus I_a}\bigr)
 =|N_a|-|I_a|-\log|X_a|.
\]
The equality holds because the coordinates in $I_a$ are fixed, so
projection onto the remaining coordinates is injective on $X_a$.

We now verify the codimension bound. Conditional on the current transcript, the input is uniform on a product rectangle.
Restricting the sender's factor according to one original message bit
increases its expected $D_a$ by the conditional Shannon entropy of that
bit, which is at most $1$. Applying \Cref{cl:density-restoring} to $Y_a$
afterwards fixes exactly $|J_{a,j}|$ new coordinates, since
$J_{a,j}\cap I_a=\varnothing$, and restores $\gamma$-density.
By \eqref{cl:restoration-step}, this restoration increases the expected
value of $D_a+(1-\gamma)|I_a|$ by at most $\log\mathrm e$.
The potential $\sum_a\bigl(D_a+(1-\gamma)|I_a|\bigr)$ is initially zero.
Summing over all original message bits gives
\[
 \E\left[\sum_a D_a\right]
 +h^{-2}\E[\codim(R)]
 \le(1+\log \mathrm e)|\Pi|.
\]
Here $R$ is the random leaf of the refined protocol and we used
$1-\gamma=h^{-2}$. Since all deficiencies are nonnegative,
\[
 \E[\codim(R)]
 \le \frac{1+\log \mathrm e}{1-\gamma}|\Pi|
 =O(h^2|\Pi|),
\]
which proves the lemma. Notice that the partition is applied only to
the sender's factor at each step, so the implicit constant is
independent of the number of parties.
\end{proof}

\begin{proof}[Proof of \Cref{cl:communication}]
The original protocol has a fully
verified deterministic augmentation of cost
$|\Pi|+O(\ell n\log q)$ and success probability at least $1/2$.
Apply \Cref{cl:preprocessing} to this augmentation.
Its refinement is still fully verified, so
\eqref{cl:success-exposure} gives
\[
 \frac12
 \le O\left(\frac{h^2(|\Pi|+\ell n\log q)}{T}\right)
   +\exp\left(A(\ln2)n/h^3-30n/h^2\right),
\]
where we used $T-n\ge T/2$. The last term is at most $1/4$
for sufficiently large $q$ depending only on $\alpha,A$.
Rearranging proves the claimed bound.
\end{proof}

Now, we are ready to prove \Cref{thm:classical-lower-summary}.
\begin{proof}[Proof of \Cref{thm:classical-lower-summary}]
Write $\kappa=\kappa(\alpha)>0$. By \Cref{hendpoint:lr}, choose $B=B_\alpha\ge40$ such that $T=q^{1+\kappa\ell}/(\ell\log q)^{B\ell}$ and $\log L=O(n/h^3)$
for $q\ge q_0(\alpha)$. For sufficiently large $C_\alpha$, the desired bound follows from
$s\ge1$ unless $q\ge q_0(\alpha)$ and $(\ell\log q)^B\le q^{\kappa/2}$. Assume these conditions. Then $T\ge q^{1+\kappa\ell/2}$ and $\ell\le q^{\kappa/(2B)}$, so $T=\omega(\ell n h^3)$ uniformly in $\ell$. Increasing $q_0(\alpha)$ absorbs the verification overhead
in \Cref{cl:communication}, giving $|\Pi|=\Omega(T/h^2)$.

A $\npass$-pass algorithm using $s$ bits yields a protocol of cost at most $(k\npass-1)s$: each party processes its block and passes the memory to the next party, including between passes. The final party outputs the polynomial with the same success probability of at least $0.99$. Since $k=\lceil h\rceil$, we obtain
\[
 s\ge\Omega\!\left(\frac{T}{\npass h^3}\right)
 \ge\frac{q^{1+\kappa\ell}}{\npass \cdot (\ell\log q)^{C_\alpha\ell}},
\]
after increasing $C_\alpha$ if necessary.
\end{proof}

\bibliographystyle{alphaurl}
\bibliography{references}

\appendix
\section{Calculations for the headline theorems}\label{sec:headline_calc}

We specialize
\Cref{thm:quantum-poly}, \Cref{thm:quantum-exp-summary}, and \Cref{thm:classical-lower-summary} from \Cref{sec:intro-results} to prove \Cref{thm:headline-opi} and \Cref{thm:headline-hopi}. In fact, they are immediate corollaries of the more general \Cref{thm:headline-opi-formal} and \Cref{thm:headline-hopi-formal} below. Recall $\alpha_\star=(2+\sqrt3)/4>0.93301$.

\begin{theorem}\label{thm:headline-opi-formal}
For every $\epsilon\in(0,1/4)$ and every prime $q\ge\Omega_\epsilon(1)$, the following
holds for streaming OPI with $\ell=1$ and $N=q(q-1)$.
\begin{enumerate}
 \item There is a one-pass quantum algorithm achieving $(1-\epsilon)$-approximation that has
 $O(\sqrt N\log N)$ space complexity, $\poly(N)$ preprocessing-update time, and 
 $\exp\bigl(O(\sqrt N\log N)\bigr)$ postprocessing time.
 \item There is a one-pass quantum algorithm achieving $(\alpha_\star-\epsilon)$-approximation that has
 $O(\sqrt N\log N)$ space complexity and $\poly(N)$ preprocessing-update-postprocessing time.
 \item For every integer $\npass\ge1$, every randomized $\npass$-pass
 classical $(3/4+\epsilon)$-approximation has space complexity at least
 \[
  \frac{N^{\frac12+\frac{2\epsilon}{1+4\epsilon}}}{\npass\cdot(\log N)^{O_\epsilon(1)}}.
 \]
\end{enumerate}
\end{theorem}
\begin{proof}
Set $\ell=1$, so $N=q(q-1)=\Theta(q^2)$.
For the first item, apply \Cref{thm:quantum-exp-summary} with $\alpha=1-\epsilon$.
Its space bound is $O(q\log q)=O(\sqrt N\log N)$, its
preprocessing-update time is $\poly(q)=\poly(N)$, and its
postprocessing time is $q^{O(q)}=\exp\bigl(O(\sqrt N\log N)\bigr)$.
For the second item,
\Cref{thm:quantum-poly} applies with $\alpha=\alpha_\star-\epsilon$; and the desired space and time bounds follow.

For the third item, apply \Cref{thm:classical-lower-summary} with $\alpha=3/4+\epsilon\in(3/4,1)$. Since
\[
 \kappa(3/4+\epsilon)=\frac{4\epsilon}{1+4\epsilon},
\]
every such $\npass$-pass algorithm using $s$ bits satisfies, for some
constant $B_\epsilon>0$,
\[
 s\ge\frac{q^{1+\frac{4\epsilon}{1+4\epsilon}}}{\npass\cdot(\log q)^{B_\epsilon}}
 \ge\frac{N^{\frac12+\frac{2\epsilon}{1+4\epsilon}}}{\npass\cdot(\log N)^{B_\epsilon}},
\]
where the second inequality uses $N=\Theta(q^2)$.
\end{proof}

\begin{theorem}\label{thm:headline-hopi-formal}
For every $\epsilon\in(0,1/4)$, there is an integer $C=\Theta_\epsilon(1)$
such that, for every prime $q\ge\Omega_\epsilon(1)$, the following
holds for streaming Hermite OPI with
$\ell=\lfloor q^{1/C}\rfloor$ and $N=(q-1)q^\ell$.
\begin{enumerate}
 \item There is a one-pass quantum algorithm achieving $(1-\epsilon)$-approximation that has
 $O((\log N)^{C+1})$ space complexity, $(\log N)^{O(C)}$ preprocessing-update time, and
 $\exp\bigl(O((\log N)^{C+1})\bigr)$ postprocessing time.
 \item There is a one-pass quantum algorithm achieving $(\alpha_\star-\epsilon)$-approximation that has
 $O((\log N)^{C+1})$ space complexity and $(\log N)^{O(C)}$ preprocessing-update-postprocessing time.
 \item For every integer $\npass\ge1$, every randomized $\npass$-pass
 classical $(3/4+\epsilon)$-approximation has space complexity at least $N^{\frac{2\epsilon}{1+4\epsilon}}/\npass$.
\end{enumerate}
\end{theorem}

\begin{proof}
Write $\kappa=\kappa(3/4+\epsilon)=4\epsilon/(1+4\epsilon)$.
By \Cref{thm:classical-lower-summary}, there is a constant
$B_\epsilon\ge1$, such that every
$\npass$-pass classical $(3/4+\epsilon)$-approximation using $s$ bits satisfies
\begin{equation}
 \npass\cdot s\ge\frac{q^{1+\ell\kappa}}{(\ell\log q)^{B_\epsilon\ell}}.
 \label{eq:headline-classical-tradeoff}
\end{equation}
Choose an integer $C\ge\max\{2,4B_\epsilon/\kappa\}$. With $\ell=\lfloor q^{1/C}\rfloor$, for sufficiently
large $q$ we have
\[
 q\le(2\ell)^C\le(\ell\log q)^C\le(\log N)^C
 \quad\text{and}\quad \ell\log q\le\log N.
\]
Consequently,
\begin{equation}
 q\ell\log q\le(\log N)^{C+1},
 \qquad \poly(q,\ell)=(\log N)^{O(C)}.
 \label{eq:headline-hermite-resources}
\end{equation}

Since $C>1$, we have $\ell\log(\ell+1)=o(q)$.
Thus $q\ge\Omega_{1-\epsilon}(\ell\log(\ell+1))$ for all $q\ge\Omega_\epsilon(1)$.
Apply \Cref{thm:quantum-exp-summary} with $\alpha=1-\epsilon$.
\eqref{eq:headline-hermite-resources} gives the claimed space, and preprocessing-update time, while its postprocessing time is
\[
 q^{O(q\ell)}=\exp\bigl(O(q\ell\log q)\bigr)
 =\exp\bigl(O((\log N)^{C+1})\bigr).
\]
This proves the first item.

For the second item, we simply apply \Cref{thm:quantum-poly} with
$\alpha=\alpha_\star-\epsilon$.
For the third item, for sufficiently large $q$, both $\ell\le q^{1/C}$ and
$\log q\le q^{1/C}$. 
Hence \eqref{eq:headline-classical-tradeoff} implies
\[
 \npass\cdot s\ge\frac{q^{1+\ell\kappa}}{(\ell\log q)^{B_\epsilon\ell}}
 \ge q^{1+\ell\kappa-2B_\epsilon\ell/C}
 \ge q^{1+\ell\kappa/2}
 \ge N^{\kappa/2}.
\]
The third inequality uses $C\ge4B_\epsilon/\kappa$, and the last uses $N\le q^{\ell+1}$ and $\kappa/2<1$.
\end{proof}

\section{Optimum of random Hermite OPI instances}
\label{app:classical-random-sat}

Recall that $\mu$ samples the lists independently and uniformly at their prescribed size, and that $\rho=(1-q^{-\ell})/2$. Set $n=q-1$.
\begin{lemma}\label{hb:random-sat}
Under the hard input distribution $\mu$,
\begin{equation}
 \Pr[\OPT<n]
 \le (q^K\rho^n)^{-1}
 =e^{-\Omega(n\ell\log q)}.
 \label{hb:random-sat-bound}
\end{equation}
\end{lemma}
\begin{proof}
Let $Z$ count the codewords satisfying all rows. Each codeword succeeds with probability $\rho^n$, so $\E[Z]=q^K\rho^n$. We show that $\Var(Z)\le\E[Z]$. Choose a uniform ordered pair of distinct codewords $c,c'$. Their difference corresponds to a uniform nonzero polynomial of degree less than $K$. For every $I\subseteq[n]$,
\[
 \Pr[c_I=c'_I]=
 \begin{cases}
  \displaystyle\frac{q^{K-\ell|I|}-1}{q^K-1},&\ell|I|\le K,\\[5pt]
  0,&\ell|I|>K,
 \end{cases}
\]
Both cases follow from the extension count and uniqueness in \Cref{fact:hermite-interpolation}, excluding the zero polynomial. In the first case, the displayed ratio is at most $q^{-\ell|I|}$; in the second case, the probability is zero. Thus, in both cases, $\Pr[c_I=c'_I]\le |\Sigma|^{-|I|}$.

Put $\theta=\rho(|\Sigma|\rho-1)/(|\Sigma|-1)$. For fixed $c,c'$, the probability
that row $i$ accepts both symbols is
$\theta+(\rho-\theta)\1[c_i=c'_i]$. Both coefficients are
nonnegative. Expanding the product over the independent rows and then
averaging over distinct codeword pairs gives
\begin{align*}
 \E_{c\ne c'}\left[\Pr_F[c,c'\text{ satisfy all rows}]\right]
 &\le\sum_{I\subseteq[n]}
       \theta^{n-|I|}(\rho-\theta)^{|I|}|\Sigma|^{-|I|}
 =\left(\theta+\frac{\rho-\theta}{|\Sigma|}\right)^n
 =\rho^{2n}.
\end{align*}
Thus,
\[
 \E[Z(Z-1)] = \sum_{c\ne c'}\Pr_F[c,c'\text{ satisfy all rows}] \le q^K(q^K-1)\rho^{2n}\le(\E [Z])^2,
\]
which proves $\Var(Z) = \E[Z^2] - (\E[Z])^2 =  \E[Z(Z-1)] + \E[Z] - (\E[Z])^2 \le\E[Z]$. Finally, $\OPT<n$ exactly when $Z=0$. Chebyshev's
inequality gives
\[
 \Pr[Z=0]
 \le\frac{\Var(Z)}{(\E[Z])^2}
 \le\frac1{\E[Z]}
 =(q^K\rho^n)^{-1} = e^{-\Omega(n\ell\log q)}.
\]
The last line follows from $K=\ell n/2$ and
$\rho=(1-q^{-\ell})/2$.
\end{proof}

\section{List-recovery bound of Hermite codes}
\label{app:classical-list-recovery}

Fix $3/4<\alpha\le1$ and an integer $\ell\ge1$.
Throughout this proof, $q=n+1$ is prime.
All derivatives are Hasse derivatives; $P^{(j)}$ is the notation
used here for the derivative denoted by $P^{[j]}$ in
\cite{bcpzz2026,jeronimo2026}.
We extend the hidden-derivative construction of~\cite[Section~3]{bcpzz2026}
and the total-incidence interpolation framework
of~\cite[Section~3.6]{jeronimo2026} to length-$\ell$ jets.
The auxiliary results are stated at the steps where they are used;
the proofs of the local jet extension and the partition estimates
follow the main argument.

\hendpointlr*

\begin{proof}
All field-size thresholds below depend only on $\alpha$.
We divide the proof into four steps. First, we reduce list recovery
to counting candidate polynomials. Second, we establish the local
multiplicity conditions and an interpolation criterion for a common
differential equation. Third, we construct a space satisfying this
criterion for the incidence budget $T$, obtaining one equation for
all candidates. Fourth, we count its solutions to obtain $L$.

\paragraph{Step 1: Reduce to counting candidate polynomials.}
Put $\kappa(\alpha)=(4\alpha-3)/(4\alpha-2)$.
Fix a sufficiently large constant $B_\alpha\ge40$ and set $T=\frac{q^{1+\ell\kappa(\alpha)}}{(\ell\ln q)^{B_\alpha\ell}}$. Let $K=\ell n/2$. Put $h=\ln n$, $\zeta=2\alpha-1-10/{\ln q}$,
and $t=\lceil\zeta n\rceil$.
Let $F_i\subseteq\Sigma$ be arbitrary lists with $N_{\inc}:=\sum_i|F_i|\le T$
and, writing $S=(F_1,\ldots,F_n)$, define
\[
 \Pcal(F)=\left\{P\in\Pcal_K:\quad
 \bigl|\{i:J_\ell P(a_i)\in F_i\}\bigr|\ge t\right\}.
\]
Every codeword counted by list recovery is the encoding of a polynomial
in $\Pcal(F)$, so it suffices to prove $|\Pcal(F)|\le L$.
An incidence is a pair $(a_i,y)$ with $y\in S_i$.
If $\ell>n^{1/10}$, then
$T\le n^{2+\ell-B_\alpha\ell/10}<1$ for large $n$.
Every list is then empty, and $\Pcal(S)=\varnothing$ since $t\ge1$,
so the conclusion holds for every $L\ge1$.
We henceforth assume $1\le\ell\le n^{1/10}$.

\paragraph{Step 2: Establish the conditions for a common equation.}
Let $r\ge0$ be an integer to be chosen in Step~3.
To bound $\Pcal(F)$, we seek a nonzero auxiliary polynomial
$Q(X,Y_0,\ldots,Y_r)$ such that
\[
 F_{Q,P}(X):=Q(X,P,P^{(1)},\ldots,P^{(r)})\equiv0
 \qquad\text{for every }P\in\Pcal(F).
\]
This step establishes the required local conditions; Step~3 will choose the space and parameters and obtain $Q$. Let $D\ge1$ be an integer target multiplicity, also to be chosen in Step~3. For a finite-dimensional space $\Qcal$, we impose at most $M_{\mathrm{loc}}$ homogeneous linear conditions per incidence, ensuring
\[
 J_\ell P(a_i)=y\quad\Longrightarrow\quad
 (X-a_i)^D\mid F_{Q,P}(X).
\]
The sufficient dimension and degree bounds are
\begin{equation}
 \dim\Qcal>N_{\inc}M_{\mathrm{loc}},
 \qquad
 \deg F_{Q,P}<Dt\quad(Q\in\Qcal,\ \deg P<K).
 \label{hfixed:proof-targets}
\end{equation}
The following criterion is immediate from rank--nullity and the fact
that a nonzero polynomial of degree less than $Dt$ cannot have $t$
distinct roots of multiplicity $D$.

\begin{proposition}[{Adapted from \cite[Proposition~3.7]{jeronimo2026}}]
\label{hfixed:interpolation-criterion}
Let $a_1,\ldots,a_n\in\Fbb_q$ be distinct, let $F_i\subseteq\Sigma$ be finite, and put $N=\sum_i|F_i|$. Let $D,t,K\ge1$ and $r\ge0$ be integers, let $M\ge0$, and let $\Qcal\subseteq\Fbb_q[X,Y_0,\ldots,Y_r]$ be a finite-dimensional linear subspace.
Suppose $\deg F_{Q,P}<Dt$ for every $Q\in\Qcal$ and every polynomial $P$ with $\deg P<K$.
Suppose also that, for each incidence $(a_i,y)$, there is a system of at most $M$
homogeneous linear conditions on $Q$, independent of $P$, whose satisfaction guarantees
\[
 J_\ell P(a_i)=y\quad\Longrightarrow\quad
 (X-a_i)^D\mid F_{Q,P}(X)
 \qquad(\deg P<K).
\]
If $\dim\Qcal>NM$, a nonzero $Q\in\Qcal$ satisfies the local conditions at every incidence $(a_i,y)$, $y\in F_i$. For this $Q$, every polynomial $P$ with $\deg P<K$ whose jet hits at least $t$ lists satisfies
\[
 Q\bigl(X,P(X),P^{(1)}(X),\ldots,P^{(r)}(X)\bigr)\equiv0.
\]
\end{proposition}

To supply the local conditions in this criterion, we use the backward
Taylor identity~\cite[Eq.~(8)]{bcpzz2026} to relate derivatives at $u$
and $u+Z$. For any field $\Fbb$, integer $d\ge0$, and polynomial
$P\in\Fbb[X]$ of degree at most $d$, the following identity holds
in $\Fbb[X,Z]$:
\begin{equation}
 P(X)=\sum_{j=0}^{d}P^{(j)}(X+Z)(-Z)^j.
 \label{hfixed:backward-taylor}
\end{equation}

Applied to the $\ell$ prescribed derivatives, this identity leads to
$\ell$ remainder variables. To count the resulting coefficient
conditions, we also track the exponents of $Y_{\ell+1},\ldots,Y_r$.
For $r\ge\ell$, put $R=r-\ell+1$ and define the weight of a tail
$v=(v_{\ell+1},\ldots,v_r)$ by
$\omega(v)=\sum_{k=\ell+1}^r(k-\ell)v_k$.
Under the reindexing $v_{\ell+j}\mapsto v_j$, the same weight is
$\omega(v)=\sum_{j=1}^{R-1}jv_j$. Let
\[
 \Lambda_R(w)=
 \#\left\{v\in\Zbb_{\ge0}^{R-1}:
             \sum_{j=1}^{R-1}jv_j\le w\right\},
 \qquad \Lambda_R(w)=0\quad(w<0).
\]
The variable $Y_\ell$ has weight zero and is counted separately.
For any nonnegative integer multi-index $v$, write $|v|=\sum_jv_j$;
for a tail, also write $Y^v=\prod_{k=\ell+1}^rY_k^{v_k}$.

The next lemma extends the cited substitution and multiplicity argument
to $\ell$ prescribed derivatives. With $X=u+Z$, its congruence is exactly
the local divisibility condition in \Cref{hfixed:interpolation-criterion}.

\begin{lemma}[{\cite[Lemma~3.1]{bcpzz2026}} and {\cite[Lemma~3.2]{jeronimo2026}}]
\label{hfixed:constraints-lemma}
Let $r\ge\ell$ be an integer, put $R=r-\ell+1$, and let
$D\ge1$, $B,W\ge0$ be integers.
Suppose $\Qcal\subseteq\Fbb_q[X,Y_0,\ldots,Y_r]$ is a finite-dimensional
space spanned by monomials of total $Y$-degree at
most $B$ and tail weight at most $W$.
For each incidence $(u,y)$, there is a system of at most
\[
 D(B+1)\binom{\lceil D/R\rceil+\ell}{\ell}\Lambda_R(W+D)
\]
homogeneous linear conditions on the coefficients of $Q\in\Qcal$.
If $J_\ell P(u)=y$ and $Q$ satisfies these conditions at $(u,y)$, then
\begin{equation}
 Q\bigl(u+Z,P(u+Z),P^{(1)}(u+Z),\ldots,P^{(r)}(u+Z)\bigr)
 \equiv0\pmod{Z^D}.
 \label{hfixed:local-projection}
\end{equation}
\end{lemma}

We will prove \Cref{hfixed:constraints-lemma} after completing the proof
of \Cref{hendpoint:lr}. For now, it supplies the local conditions for
any monomial space with the stated degree and weight bounds.
It remains to choose such a space satisfying both inequalities in
\eqref{hfixed:proof-targets}.

\paragraph{Step 3: Construct the space and verify the incidence budget $T$.} We choose the tail exponents first. Their ordinary degrees control
the specialization degree, while the local condition count involves
$\Lambda_R(W+D)$. The next lemma retains many tails of bounded ordinary
degree and controls the growth from $\Lambda_R(W)$ to $\Lambda_R(W+D)$.
The corresponding estimate in~\cite[Lemma~3.4]{jeronimo2026}
is stated for cubic multiplicity. We use the following version with
the smaller multiplicity $m=\lceil R\ln^2(eR)\rceil$.

\begin{lemma}[Weighted-partition estimates]
\label{part:growth}
Let $R\ge2$ be an integer and $a>1$.
Put $H=\ln(eR)$,$m=\lceil RH^2\rceil$,$W=\left\lfloor\frac{aRm}{H}\right\rfloor$. For every integer $b\ge0$,
\begin{equation}
\Lambda_R(W+bm)\le(eR)^{b/a}\Lambda_R(W).
 \label{part:ratio}
\end{equation}
If $V$ is uniform on the vectors counted by $\Lambda_R(W)$, then
\begin{equation}
 \E|V|\le am+\frac{m}{2H}.
 \label{part:mean}
\end{equation}
Consequently, for every $a_+>a$, the set $\Tcal_{a,a_+}
 =\{v\in\Zbb_{\ge0}^{R-1}:\omega(v)\le W,\ |v|\le\lceil a_+m\rceil\}$ satisfies
\begin{equation}
 |\Tcal_{a,a_+}|
 \ge\frac{a_+-a-(2H)^{-1}}{a_+}\Lambda_R(W).
 \label{part:retained-fraction}
\end{equation}
\end{lemma}

We will also prove \Cref{part:growth} after completing the main argument.
To apply \Cref{part:growth} here, set $R=\left\lfloor\frac{n}{100\ell^2h^4}\right\rfloor$,
$r=R+\ell-1$, and $H=\ln(eR)$.
For all sufficiently large $n$, uniformly over $1\le\ell\le n^{1/10}$,
\[
 \frac{n}{200\ell^2h^4}\le R\le\frac{n}{100\ell^2h^4},
 \qquad \frac h2\le H\le h,\qquad r\le2R.
\]
In particular, $2\le R\le n$.
Set $m=\lceil RH^2\rceil$ and $a=2\zeta-\frac8H$; then
$a=4\alpha-2-O(1/h)>1$.
With $W=\left\lfloor\frac{aRm}{H}\right\rfloor$, define
\[
 \Tcal=\left\{v\in\Zbb_{\ge0}^{R-1}:
 \omega(v)\le W,\quad
 |v|\le\left\lceil\left(2\zeta-\frac4H\right)m\right\rceil\right\}.
\]
The retained-fraction bound in \Cref{part:growth}, with
$a_+=2\zeta-4/H$, gives
\begin{equation}
 |\Tcal|
 \ge\frac{4-1/2}{H(2\zeta-4/H)}\Lambda_R(W)
 \ge\frac{\Lambda_R(W)}H.
 \label{hendpoint:partition-ratio}
\end{equation}

We turn this tail set into a rectangular monomial space, following
the constructions in the proofs of \cite[Lemma~3.2]{bcpzz2026} and
\cite[Lemma~3.3]{jeronimo2026}. The degree bounds below follow directly
from the exponent ranges, and the dimension counts the allowed monomials.

\begin{lemma}[{\cite[Lemma~3.2]{bcpzz2026}} and {\cite[Lemma~3.3]{jeronimo2026}}]
\label{hfixed:box-lemma}
Let $K\ge2$, $r\ge\ell\ge1$ be integers, and let $s>0$ and $A\ge0$.
Let $\Tcal\subseteq\Zbb_{\ge0}^{r-\ell}$ be a finite set
with $|v|\le A$ for every $v\in\Tcal$.
Define $\Qcal$ to be the span of the monomials
\[
 X^{d_X}\prod_{j=0}^{\ell}Y_j^{e_j}Y^v,
 \qquad v\in\Tcal,\quad
 0\le e_j\le s,\quad 0\le d_X<(K-1)s,
\]
where all exponents are integers and
$Y^v=\prod_{j=\ell+1}^rY_j^{v_j}$.
This space has the following properties:
\begin{enumerate}
\item For $Q\in\Qcal$, the total $Y$-degree is at most
$A+(\ell+1)s$. Hence
\[
 \deg_{Y_j}Q\le A+(\ell+1)s\qquad(0\le j\le r).
\]
\item For $Q\in\Qcal$ and $P\in\Fbb_q[X]$ with $\deg P\le K-1$,
\[
 \deg_X Q\bigl(X,P(X),P^{(1)}(X),\ldots,P^{(r)}(X)\bigr)
 <(K-1)\bigl(A+(\ell+2)s\bigr).
\]
\item The dimension is
\[
 \dim\Qcal
 =\lceil(K-1)s\rceil(\lfloor s\rfloor+1)^{\ell+1}|\Tcal|
 \ge(K-1)s^{\ell+2}|\Tcal|.
\]
\end{enumerate}
\end{lemma}

Apply \Cref{hfixed:box-lemma} to the tail set and parameters above, with
$s=m/(\ell H)$ and $A=\lceil(2\zeta-4/H)m\rceil$, to define
$\Qcal$. Every defining monomial satisfies
\[
 \begin{aligned}
 \frac{d_X}{K-1}+\sum_{j=0}^{\ell}e_j+|v|
 &<A+(\ell+2)s
 \le2\zeta m-\frac mH+1<2\zeta m,
 \end{aligned}
\]
where the last inequality uses $m/H\ge RH>1$.
Set $D=\ell m$.
Thus every $Q\in\Qcal$ has total $Y$-degree at most $2m$ and
$\deg F_{Q,P}<2\zeta(K-1)m<Dt$ for every polynomial $P$ with $\deg P<K$.
Moreover,
\begin{equation}
 \dim\Qcal
 \ge(K-1)s^{\ell+2}|\Tcal|
 \ge\frac{Km^{\ell+2}}
           {2\ell^{\ell+2}H^{\ell+3}}\Lambda_R(W).
 \label{hendpoint:dimension}
\end{equation}

Every defining monomial of $\Qcal$ has tail weight at most $W$
by the choice of $\Tcal$. Together with its total $Y$-degree bound
$2m$, this verifies the hypotheses of \Cref{hfixed:constraints-lemma}
with $B=2m$ and $D=\ell m$. Each incidence therefore requires at most
\[
 M_{\mathrm{loc}}
 \le3\ell(3e)^\ell m^{\ell+2}R^{-\ell}
       \Lambda_R(W+\ell m)
\]
conditions, using
$\binom{\lceil\ell m/R\rceil+\ell}{\ell}\le(3em/R)^\ell$.
The growth bound in \Cref{part:growth}, with $b=\ell$, gives
$\Lambda_R(W+\ell m)\le(eR)^{\ell/a}\Lambda_R(W)$.
Combining these estimates with \eqref{hendpoint:dimension}, we obtain
\begin{equation}
 \frac{\dim\Qcal}{M_{\mathrm{loc}}}
 \ge\frac{nR^{\ell(1-1/a)}}
           {C^\ell\ell^{\ell+2}H^{\ell+3}}
 \ge\frac{nR^{\ell\kappa(\alpha)}}
           {C_0^\ell\ell^{\ell+2}H^{\ell+3}}
 \label{hendpoint:interpolation-ratio}
\end{equation}
for absolute constants $C,C_0$.
For the second inequality, $a=4\alpha-2-20/{\ln q}-8/H>1$ and $R\le n$ imply
\[
 0\le\left(\kappa(\alpha)-1+\frac1a\right)\ln R
 =\frac{(20/{\ln q}+8/H)\ln R}{(4\alpha-2)a}\le28.
\]

Since $0<\kappa(\alpha)\le1/2$, substituting our choice of $R$ in
\eqref{hendpoint:interpolation-ratio} gives
\[
 \frac{nR^{\ell\kappa(\alpha)}}
      {C_0^\ell\ell^{\ell+2}H^{\ell+3}}
 \ge\frac{q^{1+\ell\kappa(\alpha)}}
           {(C_1\ell\ln q)^{6\ell}}
 \ge2T>N_{\inc}
\]
for an absolute constant $C_1$ and sufficiently large $n$, by the
choice of $B_\alpha$ in Step~1.
This verifies the dimension inequality in \eqref{hfixed:proof-targets};
the specialization degree bound and local conditions were established
above. Applying \Cref{hfixed:interpolation-criterion} with $N=N_{\inc}$ and $M=M_{\mathrm{loc}}$ now supplies a single
nonzero $Q\in\Qcal$ for which $F_{Q,P}=0$ for every $P\in\Pcal(F)$.

\paragraph{Step 4: Bound the number of solutions to obtain $L$.}
It remains to bound the number of low-degree polynomial solutions of
the common equation. We use the cardinality conclusion of
\cite[Theorem~4.3]{kopparty2015}.
The individual-degree parameter is denoted by $B$ here, rather than
the source's $t$, to distinguish it from the agreement threshold.

\begin{theorem}[Differential-root bound {\cite[Theorem~4.3]{kopparty2015}}]
\label{hfixed:differential-roots}
Let $q$ be prime and let $d,r,B$ be integers with
$d\ge1$, $0\le r<\min\{d,q\}$, and $B>0$.
Suppose $0\ne Q\in\Fbb_q[X,Y_0,\ldots,Y_r]$ satisfies:
\begin{enumerate}
\item $\deg_{Y_j}Q\le B<q$ for each $0\le j\le r$;
\item the $(1,d,d-1,\ldots,d-r)$-weighted degree of $Q$ is less than $q^2$.
\end{enumerate}
Define its set of polynomial solutions of degree at most $d$ by
\[
 \Lcal=\left\{P\in\Fbb_q[X]:\deg P\le d,\quad
 Q\bigl(X,P(X),P^{(1)}(X),\ldots,P^{(r)}(X)\bigr)=0\right\}.
\]
Then
\[
 |\Lcal|\le B(r+1)q^{2r\lfloor d/q\rfloor+4r+4}.
\]
\end{theorem}

We apply this theorem to the common nonzero $Q$ from Step~3, with
$d=K-1$ and $B=2m$.
Our parameter choice gives $r\le2R$, $m=O(n/(\ell^2h^2))$, and
$Dt=O(n^2/(\ell h^2))$. Hence, for all sufficiently large $n$,
\begin{equation}
 r<\min\{K-1,q\},\qquad
 2m<q,\qquad Dt<q^2.
 \label{hendpoint:root-conditions}
\end{equation}
Also, the normalized monomial bound in Step~3 gives
$\deg_{(1,K-1,K-2,\ldots,K-1-r)}Q<Dt$, since each weight
$K-1-j$ is at most $K-1$.
Together with the total $Y$-degree bound $2m$, these inequalities verify
all hypotheses of \Cref{hfixed:differential-roots}.
For this same $Q$ and $d=K-1$, every $P\in\Pcal(F)$ belongs to
the solution set $\Lcal$ in that theorem. Hence
\begin{equation}
 |\Pcal(F)|\le|\Lcal|
 \le 2m(r+1)q^{2r\lfloor(K-1)/q\rfloor+4r+4}.
 \label{hendpoint:candidate-count}
\end{equation}

Since $\lfloor(K-1)/q\rfloor=O(\ell)$ and
$r=O(n/(\ell^2h^4))$, the bound in \eqref{hendpoint:candidate-count} is at most
\[
 q^{O(\ell r)}
 =\exp\left(O\left(\frac{n}{\ell h^3}\right)\right)
 \le2^{C_2 n/h^3}
\]
for an absolute constant $C_2$.
Taking $L=\lceil2^{C_2 n/h^3}\rceil$ also covers the empty-list case
in Step~1 and gives
\[
 T=q^{1+\ell\kappa(\alpha)}(\ell\ln q)^{-O_\alpha(\ell)},
 \qquad
 \log L=O(n/h^3).
\]
Since the lists were arbitrary, $|\Pcal(F)|\le L$ establishes $(\zeta,T,L)$-list recovery.
\end{proof}

We now prove the two deferred lemmas, starting with
\Cref{hfixed:constraints-lemma}.

\begin{proof}[Proof of \Cref{hfixed:constraints-lemma}]
Fix an incidence $(u,y)$, with $y=(y_0,\ldots,y_{\ell-1})$.
Put $Z=X-u$ and $A_y(Z)=\sum_{j<\ell}y_jZ^j$.
Suppose $J_\ell P(u)=y$.
Apply \eqref{hfixed:backward-taylor} to $P^{(j)}$ at $X=u$, using the Hasse derivative
identity $(P^{(j)})^{(s)}=\binom{j+s}{j}P^{(j+s)}$. For $j<\ell$,
\[
 y_j=\sum_{k=j}^r(-1)^{k-j}\binom kj
       Z^{k-j}P^{(k)}(u+Z)\pmod{Z^{r+1-j}}.
\]
These equations are unitriangular in the first $\ell$ derivatives.
Eliminating them in decreasing order of $j$ yields integer coefficients
$c_{j,k}$, depending only on $\ell,r,j,k$, for which
\[
 P^{(j)}(u+Z)=A_y^{(j)}(Z)
       +\sum_{k=\ell}^r c_{j,k}Z^{k-j}P^{(k)}(u+Z)
       \pmod{Z^{r+1-j}}.
\]
Each substitution preserves the modulus: an error divisible by
$Z^{r+1-k}$ is multiplied by $Z^{k-j}$.
The prescribed-jet terms equal $A_y^{(j)}(Z)$, as is seen by applying
the same elimination to $P(X)=A_y(X-u)$.
All diagonal coefficients are $1$, so no division is used.

Introduce remainder variables $E_0,\ldots,E_{\ell-1}$.
In $Q$, substitute $X=u+Z$ and
\[
 Y_j=A_y^{(j)}(Z)+Z^{\ell-j}E_j
       +\sum_{k=\ell}^r c_{j,k}Z^{k-j}Y_k
 \qquad(j<\ell).
\]
Setting $Y_k=P^{(k)}(u+Z)$ makes each $E_j$ a polynomial divisible
by $Z^{(r+1-j)-(\ell-j)}=Z^R$.

Writing $E=(E_0,\ldots,E_{\ell-1})$, expand the transformed polynomial as
\[
 Q_{u,y}(Z,E,Y_\ell,\ldots,Y_r)
 =\sum_{\nu,k,v}q_{\nu,k,v}(Z)E^\nu Y_\ell^kY^v,
 \qquad E^\nu=\prod_{j<\ell}E_j^{\nu_j}.
\]
For every $\nu$ with $R|\nu|<D$, impose
\begin{equation}
 q_{\nu,k,v}(Z)\equiv0\pmod{Z^{D-R|\nu|}}
 \qquad\text{for all }k,v.
 \label{hfixed:coefficient-constraints}
\end{equation}
This is the jet version of the coefficient conditions in
\cite[Eq.~(15)]{bcpzz2026,jeronimo2026}.
After substituting the derivatives of $P$, the factor $E^\nu$ is
divisible by $Z^{R|\nu|}$. If $R|\nu|<D$, the coefficient in
\eqref{hfixed:coefficient-constraints} supplies the remaining power;
if $R|\nu|\ge D$, the remainder factor suffices.
Thus every summand vanishes modulo $Z^D$, proving~\eqref{hfixed:local-projection}.

It remains to count the scalar coefficient conditions. They set to zero
the coefficients of $Z^dE^\nu Y_\ell^kY^v$ with $d+R|\nu|<D$.
The substitution does not increase total degree in the $E$- and $Y$-variables.
Moreover, each newly introduced tail variable $Y_k$ is accompanied
by a power $Z^{k-j}$ whose exponent is at least its tail weight $k-\ell$.
Hence every recorded monomial satisfies
\[
 0\le d<D,\qquad k\le B,\qquad
 |\nu|<D/R,\qquad \omega(v)\le W+D.
\]
Counting these coefficient positions gives the claimed bound.
\end{proof}

We next prove \Cref{part:growth}, which gives the weighted-partition
estimates used in Step~3.

\begin{proof}[Proof of \Cref{part:growth}]
Let $p(s)$ count vectors of weight exactly $s$, setting $p(s)=0$
for $s<0$. Adding a weight-one part shows that $p(s)$ is nondecreasing
for $s\ge0$. Counting each vector with multiplicity equal to its weight gives
\[
 s\,p(s)=\sum_{j=1}^{R-1}j\sum_{k\ge1}p(s-jk)
 \le(R-1)\Lambda_R(s-1)\qquad(s\ge1).
\]
Indeed, for each $j$ and $1\le k\le\lfloor s/j\rfloor$,
\[
 j\,p(s-jk)
 \le\sum_{u=(k-1)j+1}^{kj}p(s-u),
\]
and these blocks are disjoint.
Since $\Lambda_R(s)=\Lambda_R(s-1)+p(s)$, for every integer $u\ge0$,
\[
 \frac{\Lambda_R(W+u)}{\Lambda_R(W)}
 \le\prod_{s=W+1}^{W+u}\left(1+\frac{R-1}{s}\right)
 \le\exp\left(\frac{(R-1)u}{W}\right).
\]
Also $am/H\ge1$, so
$W\ge a(R-1)m/H$.
Taking $u=bm$ proves~\eqref{part:ratio}. For the mean bound, add the slack variable
$X_0=W-\sum_{j=1}^{R-1}jV_j$.
Conditional on all $V_i$ with $i\ne j$, write $X_0+jV_j=B_j$.
Then $V_j$ is uniform on $\{0,\ldots,\lfloor B_j/j\rfloor\}$, and thus
\[
 0\le\E[X_0-jV_j]\le j-1.
\]
Summing over $j$ and using $\E X_0+\sum_jj\E V_j=W$ gives
\[
 \E X_0
 \le\frac{W+(R-1)(R-2)/2}{R}.
\]
Writing $H_{R-1}=\sum_{j=1}^{R-1}1/j\le H$, we conclude that
\[
 \E[|V|]
 \le H_{R-1}\E X_0
 \le am+\frac{RH}{2}
 \le am+\frac{m}{2H}.
\]
Finally, \eqref{part:retained-fraction} follows from \eqref{part:mean}
by Markov's inequality.
\end{proof}

\section{Classical streaming at every approximation}
\label{sec:classical-upper}
We use the notation of \Cref{sec:preliminaries}. In particular, $\kappa(\alpha)=(4\alpha-3)/(4\alpha-2)$. We set $n=q-1$ and prove \Cref{thm:classical-upper-summary}. Recall $K=\ell n/2$ and $\rho=(|\Sigma|-1)/(2|\Sigma|)$. For $P=f_u\in\Pcal_K$, write $\score(P)=\score(u)$.

\ClassicalStreamingUpper*

For fixed $\ell,\alpha$, choose $Q_{\ell,\alpha}$ so that the analysis
below applies whenever $q\ge Q_{\ell,\alpha}$. For
$3\le q<Q_{\ell,\alpha}$, store all positive addresses during the
single pass, then enumerate $\Pcal_K$ and return a polynomial of maximum
exact score. This uses $O(\ell q^{\ell+1}\log q)$ space,
$\poly(q,\ell)$ preprocessing-update time, and $q^{O(q\ell)}$
postprocessing time. Since the range of $q$ is finite for fixed
$\ell,\alpha$, its space usage is absorbed into the claimed
$O_{\ell,\alpha}(q^{1+\ell\kappa(\alpha)}\log q)$ bound.
Henceforth assume $q\ge Q_{\ell,\alpha}$.

The algorithm is given in \Cref{cu:algorithm}. It stores a random sample of positive addresses and then searches for a polynomial that matches sufficiently many of them. It uses two public parameters: a sampling probability $\eta$ and an acceptance threshold $h_0$. For all sufficiently large $n$, the parameters are chosen below
\begin{equation}
 h=\left\lceil\left(2\alpha-1+\frac1{\ln n}\right)n\right\rceil,
 \qquad d=\left\lceil\frac n{(\ln n)^4}\right\rceil,
 \qquad h_0=h-d,\qquad r=\frac n{(\ln n)^{10}}.
 \label{cu:parameters}
\end{equation}
Choose a dyadic probability $\eta$ with
$\eta_*= \left(\frac{\exp\bigl(n/[100(\ln n)^2]\bigr)} {q^K\binom nh\rho^h}\right)^{1/h}$ and $\eta/\eta_*=1+O(n^{-2})$. The dyadic probability can be represented with $O_{\ell,\alpha}(\log n)$ binary digits.

\begin{algorithm2e}[t]
\caption{One-pass classical algorithm for $3/4<\alpha<1$}
\label{cu:algorithm}
\normalfont
\DontPrintSemicolon

\KwIn{The membership stream for $(\alpha,q,\ell)$-HOPI.}
\KwOut{The coefficient vector of a polynomial in $\Pcal_K$.}

\BlankLine
\tcp{Initialize}
$W\gets\varnothing$, $B_W\gets\lceil n(|\Sigma|-1)\eta\rceil$\;
Reserve $B_W$ address slots\;
$\mathsf{failed}\gets\mathsf{false}$\;

\BlankLine
\tcp{Sample during the pass}
\ForEach{record $(i,z,c)$ in the stream}{
  \If{$c=1$ and $\mathsf{failed}=\mathsf{false}$}{
    \If{an independent Bernoulli$(\eta)$ trial succeeds}{
      \eIf{$|W|<B_W$}{
        $W\gets W\cup\{(i,z)\}$\;
      }{
        $\mathsf{failed}\gets\mathsf{true}$\;
      }
    }
  }
}

\BlankLine
\tcp{Search after the pass}
\If{$\mathsf{failed}=\mathsf{false}$}{
  \ForEach{$u\in\Fbb_q^K$ in a fixed order}{
    $P(X)\gets\sum_{r=0}^{K-1}u_rX^r$\;
    Compute the sampled score $v_W(P) = \bigl|\{i\in[n]:(i,J_\ell P(a_i))\in W\}\bigr|$\label{cu:accept}
    \If{$v_W(P)\ge h_0$}{
      \KwRet{$u$}\;
    }
  }
}

\BlankLine
\tcp{Fallback}
\KwRet{$\rv{0}\in\Fbb_q^K$}\;
\end{algorithm2e}

\subsection{Analysis of the algorithm}
\label{cu:analysis}
Throughout this argument, $W$ denotes the uncapped independent samples. The actual algorithm agrees with it unless the storage cap is exceeded;
we bound its probability after the two acceptance claims. The samples $W$ must satisfy two properties: no polynomial of score below $\alpha n$ passes the threshold, and at least one polynomial passes. For the first, a union bound suffices once we control the number of polynomials at each score. For the second, we count \emph{certificates}:
sets of $h$ positive addresses belonging to a high-score polynomial. A second-moment bound shows that some certificate survives in full with probability at least $e^{-O(r)}$. Concentration then shows that, with high probability, some certificate loses at most $d$ addresses, so its polynomial passes the lower threshold $h_0=h-d$. Since $\OPT\le n$, an output of score at least $\alpha n$ is an
$\alpha$-approximation.

\paragraph{Score estimates.}
We first state the score estimates used in both parts of the argument.
They follow from the Hermite list-recovery bound in \Cref{hendpoint:lr}
and the isolation and moment-transfer arguments in
\Cref{lem:nearopt-isolation} and \eqref{eq:nearopt-moment-identity}.
We defer their proof until after the two acceptance claims.

\begin{lemma}[Score estimates under prescribed jet constraints]
\label{cu:score-distribution}
Fix a compact interval $J\subset(1/2,1)$. There are constants
$\delta,C>0$ such that the following holds for all sufficiently large $n$.
Prescribe consistent jets at any $j\le\delta n$ rows, and choose a
polynomial uniformly among all polynomials in $\Pcal_K$ satisfying these
constraints. Let $\pi$ be its score distribution on the remaining
$N=n-j$ rows, and put
$\mu_N(w)=\Pr[\operatorname{Bin}(N,\rho)=w]$ and
$r=n/(\ln n)^{10}$. Uniformly over the input and the prescribed jets,
\begin{align}
 \pi(w)&\le F_n\mu_N(w)\quad(w/N\in J),
 &F_n&=C\sqrt n\exp\left(C\frac n{(\ln n)^{30}}\right),
 \label{cu:strong-envelope}\\
 \pi\{w:|w-zN|\le r\}&\ge e^{-O(r)}\mu_N(\lfloor zN\rfloor)
 \quad(z\in J).
 \label{cu:strong-interval}
\end{align}
The second bound also holds for any fixed positive multiple of $r$.
\end{lemma}

Using the lemma, we prove the two acceptance claims.
Choose $J$ with $[2\alpha-1,\alpha]\subset\operatorname{int}(J)$.
We will use the elementary identity
\begin{equation}
 \mu_N(w)\binom wt
 =\binom Nt\rho^t\mu_{N-t}(w-t),\qquad 0\le t\le w\le N.
 \label{cu:binomialtilt}
\end{equation}

\paragraph{No low-score polynomial is accepted.}
A polynomial of score $w$ passes only if some $h_0$ of its positive
addresses are sampled. There are $q^K$ polynomials in total. Hence
\eqref{cu:strong-envelope} and \eqref{cu:binomialtilt} and a union bound give
\begin{align}
 &\Pr[\exists P:\score(P)<\alpha n,\ v_W(P)\ge h_0]\le F_n q^K\binom n{h_0}(\rho\eta)^{h_0}
 \Pr[\operatorname{Bin}(n-h_0,\rho)<\alpha n-h_0].
 \label{cu:badmass}
\end{align}
Here only $h_0\le w<\alpha n$ contribute, and our parameters ensure $w/n\in J$ for large $n$.  The choice of $h_0$ gives $h_0+(n-h_0)\rho\ge\alpha n+n/(3\ln n)$ for large $n$. Hence the binomial tail in \eqref{cu:badmass} is at most $\exp\bigl(-n/(10(\ln n)^2)\bigr)$.  By the choice of $\eta$ and $h-h_0=d$, the prefactor satisfies
\[
 F_n q^K\binom n{h_0}(\rho\eta)^{h_0}
 \le \exp\left(\frac n{100(\ln n)^2}
       +O(d\ln n)+o\left(\frac n{(\ln n)^2}\right)\right).
\]
Combining these bounds in \eqref{cu:badmass} and using
$d\ln n=o(n/(\ln n)^2)$ gives
\[
 \Pr[\exists P:\score(P)<\alpha n,\ v_W(P)\ge h_0]
 \le \exp\left(-\left(\frac9{100}-o(1)\right)
                    \frac n{(\ln n)^2}\right)
 =e^{-\Omega(n/(\ln n)^2)}.
\]
\paragraph{A high-score polynomial is accepted.}
Since $h>K/\ell$, we can use $h$ prescribed jets to identify a polynomial. Put $\bar w=h+(n-h)\rho $ and $I=\{w:|w-\bar w|\le r\}$. Our choice of $h$ gives $\bar w=\alpha n+n/(2\ln n)+O(1)$. Since $r=o(n/\ln n)$, the entire interval $I$ lies above $\alpha n$ and below $n$ for large $n$.
Let $\Hcal$ contain every $h$-subset of the positive addresses of
a polynomial whose score lies in $I$, and let $Z$ count those contained
in $W$. 

First, we can construct certificates for high-score polynomials. Since $h>K/\ell$, Hermite interpolation makes each certificate's polynomial unique.

Then we can bound the probability that a complete certificate survives. Applying \eqref{cu:strong-interval} at $\bar w/n$
and then \eqref{cu:binomialtilt} at $w_*=\lfloor\bar w\rfloor$ gives
\begin{equation}
 |\Hcal|\ge q^K\binom nh\rho^h e^{-O(r)},\qquad
 \E[Z]=|\Hcal|\eta^h
 \ge q^K\binom nh(\rho\eta)^h e^{-O(r)}.
 \label{cu:firstmoment}
\end{equation}
Indeed, throughout $I$ the certificate count per polynomial differs
from $\binom{w_*}h$ by at most $e^{O(r)}$, and
$\mu_{n-h}(w_*-h)=\Theta(n^{-1/2})$, absorbed into $e^{-O(r)}$. By the choice of the sampling probability,
\begin{equation}
 q^K\binom nh(\rho\eta)^h
 =\exp\left(\frac n{100(\ln n)^2}+o(1)\right).
 \label{cu:balance}
\end{equation}

For completeness, the second-moment calculation only needs the
following three overlap bounds. If $d_{\Hcal}(S)$ counts
certificates containing a fixed $j$-address set $S$, then
\[
 d_{\Hcal}(S)\le
 \begin{cases}
 F_n q^{K-\ell j}\binom{n-j}{h-j}\rho^{h-j},&0\le j\le\delta n,\\
 q^{K-\ell j}\binom{n-j}{h-j},&\delta n<j<K/\ell,\\
 \binom{n-j}{h-j},&K/\ell\le j\le h.
 \end{cases}
\]
The first follows from the conditional score estimate and
\eqref{cu:binomialtilt}; shrink $\delta$ so all residual scores from
$I$ remain in $J$. The other two are immediate from Hermite
interpolation. Inconsistent prescriptions contribute zero.
Expanding the overlap of two certificates yields
\begin{equation}
 \frac{\E[Z^2]}{(\E[Z])^2}
 \le\sum_{j=0}^h\eta^{-j}\binom hj
       \frac{\max_{|S|=j}d_{\Hcal}(S)}{|\Hcal|}.
 \label{cu:overlap}
\end{equation}
The three ranges are bounded, respectively, by
\[
 e^{O(r)}F_n\left(1+\frac1{|\Sigma|\rho\eta}\right)^h,
 \qquad e^{O(n)}\sum_{j>\delta n}(|\Sigma|\eta)^{-j},
 \qquad \frac{(1+\sqrt\eta)^{n+h}}{\E[Z]}.
\]
For the last expression, set $a=h-j$ and use
$\binom{n-h+a}{a}\le\binom na$ followed by the binomial theorem. To bound these terms, take logarithms in \eqref{cu:balance}.
Using $h=(2\alpha-1)n+O(n/\ln n)$ and $\ln\binom nh=O(n)$ gives
\begin{equation}
 \eta=\Theta_{\ell,\alpha}(n^{-\ell/(4\alpha-2)}),
 \qquad |\Sigma|\eta=\Theta_{\ell,\alpha}(n^{\ell\kappa(\alpha)}).
 \label{cu:size}
\end{equation}
In particular, $n/(|\Sigma|\eta)=o(r)$, $\ln F_n=o(r)$, and
$r+n\sqrt\eta=o(n/(\ln n)^2)$. Together with
\eqref{cu:firstmoment} and \eqref{cu:balance}, these estimates bound the three ranges by $e^{O(r)}$, $e^{-\Omega(n\ln n)}$, and $e^{-\Omega(n/(\ln n)^2)}$, respectively. Thus Cauchy--Schwarz gives $\Pr[Z>0]\ge e^{-O(r)}$.

We now show that a high-score polynomial is accepted with high
probability. Recall that $W$ denotes the uncapped independent sample. We identify $W$ with its vector of independent Bernoulli$(\eta)$ indicators, indexed by the positive addresses,
so its distribution is a product measure.
Let $Y(W)=\max_{E\in\Hcal}|E\cap W|$,$
 A=\{W:Y(W)\le h-d\}$,$
 B=\{W:Y(W)=h\}$. The second-moment bound above gives $\Pr[W\in B]\ge e^{-O(r)}$.
We may assume $\Pr[W\in A]>0$, since otherwise the desired failure bound is immediate. Fix $W\in B$ and choose a certificate $E\in\Hcal$ with $E\subseteq W$. Since $|E|=h$, every $W'\in A$ satisfies $|E\setminus W'|=h-|E\cap W'|
 \ge h-Y(W')\ge d$ Thus every $s\in U_A(W)$ satisfies $\sum_{e\in E}s_e\ge d$. This linear inequality remains valid on the convex hull
$V_A(W)$. Hence, by Cauchy--Schwarz, every $v\in V_A(W)$ satisfies $d\le\sum_{e\in E}v_e
 \le\sqrt{h}\,\|v\|_2$ Taking the infimum over $v\in V_A(W)$ yields
\[
 f_c(A,W)\ge\frac{d}{\sqrt h}
 \qquad\text{for every }W\in B.
\]
Applying the following theorem of
Talagrand~\cite[Theorem~4.1.1]{talagrand1995}
with this set $A$ and $t=d/\sqrt h$.
\begin{theorem}[Talagrand's convex-distance inequality
    {\cite[Theorem~4.1.1]{talagrand1995}}]
\label{thm:talagrand-convex-distance}
Let $\Omega$ be a finite set equipped with a probability measure
$\mu$, and let $\mathbb{P}=\mu^{\otimes M}$.
For $A\subseteq\Omega^M$ with $\mathbb{P}(A)>0$ and
$x\in\Omega^M$, define
\[
 U_A(x)
 =
 \left\{
 s\in\{0,1\}^M:
 \exists y\in A,\ 
 s_i=0\Rightarrow x_i=y_i
 \text{ for every }i\in[M]
 \right\},
\]
and set $V_A(x)=\operatorname{conv}U_A(x)$ and $
 f_c(A,x)=\inf_{v\in V_A(x)}\|v\|_2$. Then
\[
 \int_{\Omega^M}
 \exp\!\left(\frac{f_c(A,x)^2}{4}\right)
 \,d\mathbb{P}(x)
 \le \frac{1}{\mathbb{P}(A)}.
\]
Consequently, for every $t\ge0$, $\mathbb{P}(A)\,
 \mathbb{P}\bigl(\{x:f_c(A,x)\ge t\}\bigr)
 \le \exp(-t^2/4)$.
\end{theorem}
We can obtain
\[
 \Pr[Y\le h-d]
 \le
 \frac{\exp(-d^2/(4h))}{\Pr[Y=h]}
 \le
 \exp\!\left(O(r)-\frac{d^2}{4h}\right).
\]

If $Y(W)\ge h_0=h-d$, some certificate $E\in\Hcal$ has at
least $h_0$ sampled addresses. Its associated polynomial $P$
has $\score(P)\in I$ and
$v_W(P)\ge |E\cap W|\ge h_0$.
Consequently,
\begin{equation}
 \begin{aligned}
 &\Pr[\nexists P:\score(P)\in I,\ v_W(P)\ge h_0]
\le \Pr[Y<h-d]
 \le \Pr[Y\le h-d]
 \le \exp\!\left(O(r)-\frac{d^2}{4h}\right).
 \end{aligned}
 \label{cu:existence}
\end{equation}
Since $h=\Theta(n)$ and
$d=\lceil n/(\ln n)^4\rceil$, we have
$d^2/h=\Theta(n/(\ln n)^8)$.
Moreover, $r=n/(\ln n)^{10}=o(d^2/h)$.
Thus the probability that no high-score polynomial is accepted
is at most $e^{-\Omega(n/(\ln n)^8)}$.

The uncapped sample size has distribution $\operatorname{Bin}(n(|\Sigma|-1)/2,\eta)$, and $B_W$ is at least twice its mean. A Chernoff bound gives overflow probability $e^{-\Omega(n|\Sigma|\eta)}$. Combining this with the two acceptance bounds, a union bound shows that, with probability $1-o(1)\ge0.99$, no overflow occurs, some candidate passes, and every passing candidate has score at least $\alpha n\ge\alpha\OPT$.

\paragraph{Space complexity:} The algorithm stores at most $B_W=\lceil n(|\Sigma|-1)\eta\rceil$ addresses,
each using $O(\log n+\ell\log q)$ bits. It enumerates candidates
one at a time: their coefficients and reusable jet-evaluation
workspace occupy $O(n\ell\log q)$ bits. Counters and the dyadic
sampler use $O_{\ell,\alpha}(\log n)$ additional bits. Thus, for
fixed $\ell,\alpha$, \eqref{cu:size} gives
the space bound $O_{\ell,\alpha}\left((n|\Sigma|\eta+n\ell)\log n\right)=O_{\ell,\alpha}\left(q^{1+\ell\kappa(\alpha)}\log q\right)$ since $n=q-1$.

\paragraph{Time complexity:}
Computing the public parameters to the required precision takes
polynomial time for fixed $\ell,\alpha$ and fits in the space bound
above. During the pass, each of the $nq^\ell$ records requires a
membership-bit check and, if positive, at most $O_{\ell,\alpha}(\log n)$ random bits and an address insertion.
Addresses can be appended to an array, since each occurs only once.
Thus preprocessing-update time is polynomial in $n$ and $q$ for fixed $\ell,\alpha$. 

After the pass, there are $q^K=q^{n\ell/2}$ coefficient vectors to inspect. For each, Horner evaluation of $P(a_i+T)\bmod T^\ell$
computes all $n$ jets in $O(nK\ell)$ field operations. Scanning the stored addresses against these jets then computes $v_W(P)$ in
$O(B_W\ell\log q)$ bit operations. Since $B_W$ is polynomial in $n$ for fixed $\ell$, each candidate takes polynomial time. Therefore, the postprocessing time is $q^{n\ell/2}\poly_{\ell,\alpha}(n,\log q) = q^{O(q\ell)}$.

Last, we prove the technical lemma \Cref{cu:score-distribution}.
\begin{proof}[Proof of \Cref{cu:score-distribution}]
We now derive the stated estimates. 

We first express the polynomials satisfying the prescribed jet constraints in terms of a shorter Hermite code. By Hermite interpolation, conditioning on $j$ jets leaves exactly
$q^{K-\ell j}$ polynomials, of the form $P_0+VQ$, where
$V(X)=\prod_{i\text{ prescribed}}(X-a_i)^\ell$ and
$\deg Q<K-\ell j$. Translation by $P_0$ and division by the nonzero jet of $V$ on each remaining row preserve the list density $\rho$.
Thus, after this change of coordinates, the remaining jet vectors form a linear Hermite code $C$ of length $N$.

We next use list recovery to bound the number of low-weight vectors in each dual coset. Choose a constant $\tau<1/2$ larger than $\max_{z\in J}(1/2-\sqrt{z(1-z)})$ and put $t=\lfloor\tau N\rfloor$. \Cref{fct:hermite-duality} identifies $C^\perp$
with an invertible block transformation of the degree-$(<K)$ Hermite code on the remaining rows. Its minimum block distance is $N-K/\ell+1=K/\ell-j+1$. As in \Cref{lem:nearopt-dual-lists}, a word of weight
at most $t+2$ in a fixed dual coset corresponds to a polynomial matching singleton prescriptions on at least $(1-\tau)N-2$ rows. Take $\delta,\gamma>0$ so that $(1-\tau)(1-\delta)>1/2+2\gamma$.

We need a stronger logarithmic saving than the stated list-size bound in \Cref{hendpoint:lr}. In its proof in \Cref{app:classical-list-recovery}, set
$2\alpha_{\rm LR}-1=1/2+\gamma$ and use
\[
 R=\left\lfloor\frac{n}{100\ell^2(\ln n)^{32}}\right\rfloor,
 \qquad m=\lceil R(\ln(eR))^2\rceil,
 \qquad r_{\rm der}=R+\ell-1.
\]
The dimension-to-local-rank ratio is still at least
$n^{1+\ell\kappa(\alpha_{\rm LR})}/(\ln n)^{O_{\ell,\gamma}(1)}>n$, so it covers all singleton incidences. The same degree checks and root count give at most $q^{(\ell+4)r_{\rm der}+6}$ solutions. Consequently every dual coset has at most $L_n$ words of weight at most $t+2$, with $\ln L_n=O(n/(\ln n)^{30})$.

We then prove the pointwise upper bound by truncating a product weight and controlling its squared norm on the code. Under the uniform product measure on the $N$ jet blocks, let $I_i$ indicate membership in row $i$'s list and set $g_i=(I_i-\rho)/\sqrt{\rho(1-\rho)}$. For $z\in J$, write
\[
 a=\sqrt{\rho z}+\sqrt{(1-\rho)(1-z)},\qquad
 b=\sqrt{(1-\rho)z}-\sqrt{\rho(1-z)}.
\]
Let $B_z=\prod_i(a+bg_i)$, and let $A_z$ retain only terms involving
at most $t$ blocks. Both depend only on the score $w$, and
\[
 \mu_N(w)B_z(w)^2=\Pr[\operatorname{Bin}(N,z)=w],\qquad
 \E_{\mu_N}A_z^2\le1,\qquad
 \E_{\mu_N}(A_z-B_z)^2\le e^{-cN}.
\]
The last inequality follows from the binomial tail with parameter $b^2<\tau$. Grouping Fourier coefficients by their syndrome and applying Cauchy--Schwarz gives $\E_C A_z^2\le L_n\E_{\mu_N}A_z^2\le L_n$. For $z=w/N$, the binomial mass at its mean is $\Theta(N^{-1/2})$;
the displayed error bound therefore implies
$A_z(w)^2\ge B_z(w)^2/4$. This proves
\eqref{cu:strong-envelope}.

To prove the window lower bound, we first isolate Fourier indices by augmenting the syndrome. Let $\Psi$ be the Fourier coefficient vector of
$B_z$ and $\rv{a}$ that of $A_z$. Choose a uniformly random set $R$ of
$m_0=\lceil C_0(\ln L_n+\ln n)\rceil=o(r)$ remaining coordinates, and append their blocks to the syndrome map. Let $\Pi_R$ retain the indices of weight at most $t$ that are unique in their augmented syndrome through weight $t+2$. The proof of \Cref{lem:nearopt-isolation}, using distance $K/\ell-j+1$ in place of $K/\ell+1$, gives
\[
 \Delta:=\E_R\|(I-\Pi_R)\Psi\|^2
 \le e^{-cN}+L_ne^{-c_0m_0}.
\]
Choose $C_0$ so that $L_n\Delta\le n^{-100}$.

We next transfer the second moment to show that the weighted mass is concentrated near the target score. Here $S(x)$ counts satisfied lists on the remaining rows.
For an augmented generator matrix $H_R$, use the character map $U_R\xi(y)=\sum_e\xi_e\omega^{\langle y,H_Re\rangle}$, with the \emph{uniform} measure on $y$. This map is an isometry on the isolated
support and has squared norm at most $L_n$ on weight-$\le t$ support. By the argument for \eqref{eq:nearopt-moment-identity} and the product-score variance
$Nz(1-z)\le N/4$, the isolated output has average squared mass $1-\Delta$ and score second moment at most $N/2+2N^2\Delta$ about $zN$. Since $N/r^2=o(1)$, Chebyshev's inequality puts $1-o(1)$ of this mass within $r/3$ of $zN$. Also $\E_R\|U_R(\rv{a}-\Pi_R\Psi)\|^2\le L_n\Delta=o(1)$. Using $U_R\rv{a}(y)=A_z(S(H_R^\top y))$, we obtain
\[
 \E_R\E_y\left[
 |A_z(S(H_R^\top y))|^2
 \1\{|S(H_R^\top y)-zN|\le r/3\}\right]\ge\frac12.
\]

Recover the window lower bound for the original code. Since $\mu_N(w)A_z(w)^2\le1$, the probability of this augmented window is at least $\tfrac12\min_{|w-zN|\le r/3}\mu_N(w)
 \ge e^{-O(r)}\mu_N(\lfloor zN\rfloor)$.
Removing the perturbations on the appended blocks changes the score by at most $m_0=o(r)$ and leaves a uniform codeword in $C$. This proves \eqref{cu:strong-interval}; the same argument works for a fixed multiple of $r$.
\end{proof}

\section{Efficient classical streaming at \texorpdfstring{$3/4$}{3/4} approximation}
\label{sec:c34}
We use the notation of \Cref{sec:preliminaries}. In this section, we set $n=q-1$ and prove \Cref{thm:classical-three-quarters-summary}. Recall $K=\ell n/2$. For $P=f_u\in\Pcal_K$, write $\score(P)=\score(u)$.

\ClassicalThreeQuartersUpper*

The algorithm is given in \Cref{alg:c34}. It combines interpolation with random completion, an approach related to the Prange baseline~\cite{prange1962information} discussed in~\cite{jordan2025dqi}. Interpolation guarantees satisfaction on nearly half of the rows, while random completion satisfies about half of the remaining rows.
The key streaming idea is to select interpolation constraints when their first positive records arrive, ensuring that the resulting candidates can be scored exactly without revisiting the stream.

We use parameters $r=4$ and $r\ell$ is the number of free polynomial coefficients, $R=512$ is the number of candidates, and $B=\left\lceil\log(400Rr\ell)\right\rceil$ is a sampling cutoff. Set $k_0=\frac n2-r$ to be the number of rows guaranteed to be satisfied by interpolation and $m=n-k_0=\frac n2+r$ is the number of rows not selected as interpolation constraints.

The procedure $\mathsf{Sample}_B$ attempts to sample a uniform field element using at most $B$ independent rejection trials. Each trial draws $\lceil\log q\rceil$ fair bits and accepts the result if it is less than $q$, interpreting accepted integers through a fixed bijection with $\Fbb_q$; if every trial is rejected, the procedure returns $\bot$.

\begin{algorithm2e}[t]
\small
\DontPrintSemicolon
\caption{One-pass classical algorithm for the $3/4$ approximation ratio}
\label{alg:c34}

\KwIn{The $(3/4,q,\ell)$-HOPI membership stream $(i,z,\beta)$;
      parameters $r,R,B$.}
\KwOut{The coefficient vector of a polynomial in $\Pcal_K$.}

$I\gets\varnothing$; $k_0\gets n/2-r$\;

\While{$|I|<k_0$}{
    Read the next record $(i,z,\beta)$\;
    \If{$\beta=1$ and $i\notin I$}{
        $I\gets I\cup\{i\}$; $y_i\gets z$\;
    }
}

Construct $P_0\in\Pcal_{\ell k_0}$ satisfying
$J_\ell P_0(a_i)=y_i$ for all $i\in I$\;

$V(X)\gets\prod_{i\in I}(X-a_i)^\ell$\;

\For{$j=1,\ldots,R$}{
    Sample the $r\ell$ coefficients of $Q_j$ independently
    using $\mathsf{Sample}_B$\;

    \If{a sampling call returns $\bot$}{
        Consume the remaining stream\;
        \KwRet{the coefficient vector of the zero polynomial}\;
    }

    Compute and store the full coefficient vector of
    $P_j=P_0+VQ_j$\;

    $s_j\gets k_0$\;
}

\ForEach{remaining record $(i,z,\beta)$ in the stream}{
    \If{$\beta=1$ and $i\notin I$}{
        \For{$j=1,\ldots,R$}{
            \If{$J_\ell P_j(a_i)=z$}{
                $s_j\gets s_j+1$\;
            }
        }
    }
}

$j^\star\gets
 \min\operatorname*{arg\,max}_{j\in[R]}s_j$\;

\KwRet{the stored coefficient vector of $P_{j^\star}$}\;
\end{algorithm2e}

\subsection{Analysis of the algorithm}
Fix the complete input and its order. All lists are nonempty, so the initial selection phase terminates. Let $t_0$ denote the arrival time of its last selected record. The candidates are fully expanded before the stream resumes; this moves all coefficient expansion out of postprocessing. The set $I$ and the selected jets $y_i$ are deterministic and are fixed before any candidate randomness is generated. By the selection rule,
\begin{equation}
 i\notin I
 \quad\Longrightarrow\quad
 \text{every positive record in row $i$ arrives after $t_0$}.
 \label{eq:c34:future}
\end{equation}
Indeed, an earlier positive record in an unselected row would have caused that row to be selected.

Every candidate satisfies the $k_0$ selected rows. On each
unselected row, its score is incremented exactly when its jet belongs to the corresponding list, because every address appears exactly once and \eqref{eq:c34:future} ensures that the relevant positive record has not been discarded. Thus the final counters are the exact values
$s_j=\score(P_j)$.

\paragraph{Success probability of a candidate.}
Let $\Gcal$ be the event that all $Rr\ell$ planned calls to $\mathsf{Sample}_B$ succeed, viewing their independent random tapes as sampled in advance. Each rejection trial accepts with probability $q/2^{\lceil\log q\rceil}>1/2$, so each call fails with probability at most $2^{-B}$. Conditional on $\Gcal$, the accepted field elements remain independent and uniform. We condition on $\Gcal$ throughout the following candidate analysis.

By \Cref{fact:hermite-interpolation}, the polynomials in $\Pcal_K$ satisfying the selected constraints are exactly
\[
 P=P_0+VQ,
 \qquad
 \deg Q<r\ell.
\]
Uniform $Q$ therefore gives a uniform polynomial conditioned on the
selected jets. Since $K-\ell k_0=r\ell$ and $r=4$, the conditional uniformity in \Cref{fact:hermite-interpolation} makes its jets on the
unselected rows four-wise independent and uniform. Thus, for $i \notin I$, $X_i=\1\bigl[J_\ell P(a_i)\in S_i\bigr]$ are four-wise independent Bernoulli variables with mean
$\rho=(|\Sigma|-1)/(2|\Sigma|)$.

For $S=k_0+\sum_{i\notin I}X_i$, we have $\mu=\E[S] =\frac{3n}{4}-\frac r2-\frac{n+2r}{4|\Sigma|}$ and $\sigma^2=\Var S=m\rho(1-\rho)\ge\frac n9$. The variance bound follows from $m\ge n/2$,
$1/3\le\rho<1/2$ and $|\Sigma|\ge n+1$. 

Put $Y=S-\mu$. Four-wise independence gives
\[
 \E[Y^4]
 =3\sigma^4+\sigma^2\bigl(1-6\rho(1-\rho)\bigr)
 \le4\sigma^4,
\]
where $1-6\rho(1-\rho)\le0$ on the stated range of $\rho$. Moment interpolation yields
\[
 \E|Y|
 \ge\frac{(\E Y^2)^{3/2}}{(\E Y^4)^{1/2}}
 \ge\frac\sigma2.
\]
Since $\E[Y]=0$, the variable $Y_+=\max\{Y,0\}$ satisfies $\E[Y_+]=\tfrac12\E|Y|\ge\sigma/4$ and $\E[Y_+^2]\le\sigma^2$. By the Paley--Zygmund inequality,
\[
 \Pr[Y\ge\sigma/8]
 \ge\Pr\bigl[Y_+\ge\tfrac12\E Y_+\bigr]
 \ge\frac{(\E Y_+)^2}{4\E Y_+^2}
 \ge\frac1{64}.
\]

Since $n\geq 3600$, we have $\frac{3n}{4}-\mu =\frac r2+\frac{n+2r}{4|\Sigma|} \le  \sqrt n/24 \leq \sigma/8$. Hence  
\[
 \Pr\left[S\ge \frac{3n}{4}\right] \geq \Pr[Y\ge\sigma/8]
 \ge\frac1{64}.
\]
This implies a candidate has a score of at least $3n/4$ with probability at least $1/64$. Conditional on $\Gcal$, the $R$ candidates are independent, so the probability that none reaches $3n/4$ is at most $(1-1/64)^R\le e^{-R/64} = e^{-8}$. Exact score selection returns a successful candidate whenever one exists. By the union bound and the choice of $B$,
\[
 \Pr[\Gcal^c]\le Rr\ell\,2^{-B}\le\frac1{400}.
\]
Consequently, the total failure probability is at most $e^{-8}+1/400<0.01$.

\paragraph{Space complexity}
The selected jets, $P_0$, and $V$ occupy $O(n\ell)$ field elements. Sequential Hermite interpolation uses the same workspace bound: when adding a constraint at $a_i$, correct the current interpolant by $V(X)G(X-a_i)$, where $G$ is obtained by dividing the desired jet correction by $V(a_i+T)$ modulo $T^\ell$, and then replace $V(X)$ by $V(X)(X-a_i)^\ell$. This truncated division is valid
because $V(a_i)\ne0$. The $R$ stored candidate vectors occupy $O(Rn\ell)$ field elements, and polynomial arithmetic and jet evaluation use $O(n\ell)$ field elements of reusable workspace. Row flags,
score counters, and the sampling counter require $O(n+R\log n+\log B)$ additional bits. Thus the space usage is $O\bigl(Rn\ell\log q+R\log n+\log B\bigr)$.   Since $r$ and $R$ are constants and $B=O(\log(n\ell))$, the space bound is $O(n\ell\log n)$.

\paragraph{Time complexity}
Selecting constraints, Hermite interpolation, bounded coefficient
sampling, and expanding the candidate polynomials all take polynomial
time in $n,\ell,R,B,\log q$. This also bounds the work on any arriving
record, including the record that completes the selection phase.
All scores and coefficient vectors are already stored. Selecting the largest counter costs $O(R\log n)$ bit operations, and explicitly outputting the winning vector costs $O(K\log q)$ bit operations. It requires neither further interpolation nor any additional access to the input.
Since $r$ and $R$ are constants and $B=O(\log(n\ell))$, the algorithm needs $\poly(n\ell,\log n)$ preprocessing-update-postprocessing time.

\end{document}